\documentclass[10pt]{article}
\usepackage{amssymb}
\usepackage{tikz}
\usepackage{xr}
\usetikzlibrary{arrows, calc, decorations.pathmorphing, backgrounds, positioning, fit, petri, automata}
\definecolor{grey1}{rgb}{0.5,0.5,0.5}
\usepackage[top=2.54cm, bottom=2.54cm, left=2.2cm, right=2.2cm]{geometry}
\usepackage{algcompatible}
\usepackage{algorithm}
\usepackage{algorithmicx}
\usepackage{changebar}

\usepackage{hyperref}
\hypersetup{ 
  colorlinks,
  citecolor=blue,
  linkcolor=red,
  urlcolor=magenta}

\usepackage[shortlabels]{enumitem}
\usepackage{mathtools}
\usepackage{graphicx}
\usepackage{subfigure}
\usepackage{amsmath,amsthm,bm,bbm}
\usepackage{rotating}
\usepackage{mathrsfs}
\usepackage{chngcntr}
\usepackage{apptools}
\usepackage[titletoc,title]{appendix}
\usepackage{comment}

\usepackage{xcolor}         %for colors
\definecolor{grau}{rgb}{0.8,0.8,0.8}

\newcommand{\chen}[1]{\color{orange}}
\usepackage{setspace,lscape,longtable}
\usepackage{amsmath,amsthm,bm}
\usepackage{expl3}
\usepackage{fullpage}
\usepackage{multirow}

\usepackage{titlesec}
\usepackage[english]{babel}
\usepackage[comma,authoryear]{natbib}

\bibpunct{(}{)}{;}{a}{,}{,}
\numberwithin{equation}{section}
\newtheorem{theorem}{Theorem}[section]

\newtheorem{lemma}[theorem]{Lemma}

\theoremstyle{remark}
\newtheorem{definition}[theorem]{Definition}
\newtheorem{assumption}{Assumption}

\newtheorem{result}{Result}[section]
\newtheorem{remark}{Remark}
\DeclareMathOperator*{\argmin}{arg\,min}
\DeclareMathOperator*{\argmax}{arg\,max}

\newcommand{\prob}{{\mathbb{P}}}
\newcommand{\var}{{\mathrm{var}}}
\newcommand{\expect}{\mathbb{E}}

\newcommand{\iidsim}{{\overset{\mathrm{i.i.d.}}{\sim}}}
\newcommand{\transpose}{^{\mathrm{T}}}

\newcommand{\inverse}{^{-1}}
\newcommand{\halfpower}{^{1/2}}
\newcommand{\invhalfpower}{^{-1/2}}
\newcommand{\frobenius}{_{\mathrm{F}}}
\newcommand{\twotoinfinity}{_{2\to\infty}}
\newcommand{\diff}{\mathrm{d}}

\ExplSyntaxOn
\clist_new:N \l_script_letters_clist
\clist_set:Nn \l_script_letters_clist {A,B,C,D,E,F,G,H,I,J,K,L,M,N,O,P,Q,R,S,T,U,V,W,X,Y,Z}
\clist_map_inline:Nn \l_script_letters_clist {
  \cs_new:cpn {cal#1} { \mathcal{#1} }
}
\ExplSyntaxOff

\ExplSyntaxOn
\clist_new:N \l_bold_letters_clist
\clist_set:Nn \l_bold_letters_clist {A,B,C,D,E,F,G,H,I,J,K,L,M,N,O,P,Q,R,S,T,U,V,W,X,Y,Z,
                                    a,b,c,d,e,g,h,i,j,k,l,n,o,p,q,r,s,t,u,v,w,x,y,z}
\clist_map_inline:Nn \l_bold_letters_clist {
  \cs_new:cpn {b#1} { \mathbf{#1} }
}
\ExplSyntaxOff

\newcommand{\bDelta}{{\bm{\Delta}}}
\newcommand{\bTheta}{{\bm{\Theta}}}
\newcommand{\bSigma}{{\bm{\Sigma}}}

\newcommand{\bgamma}{{\bm{\gamma}}}

\newcommand{\eye}{{\mathbf{I}}}
\newcommand{\one}{{\mathbbm{1}}}

\newcommand{\bmu}{{\bm{\mu}}}

\newcommand{\zero}{{\bm{0}}}

\newcommand{\keywords}[1]{\par\addvspace\baselineskip\noindent\enspace\ignorespaces \textbf{Keywords: }#1}

\author{
Dingbo Wu \thanks{Department of Statistics, Indiana University}
\and 
Fangzheng Xie \thanks{Department of Statistics, Indiana University}
\thanks{Correspondence should be addressed to Fangzheng Xie (fxie@iu.edu)}
}
\title{SANVI: A Fast Spectral-Assisted Network Variational Inference Method with an Extended Surrogate Likelihood Function}
\date{}
\begin{document}
\allowdisplaybreaks

\maketitle

\begin{abstract}
Bayesian inference has been broadly applied to statistical network analysis, but suffers from the expensive computational costs due to the nature of Markov chain Monte Carlo sampling algorithms. This paper proposes a novel and computationally efficient Spectral-Assisted Network Variational Inference (SANVI) method within the framework of the generalized random dot product graph. The key idea is a cleverly designed extended surrogate likelihood function that enjoys two convenient features. Firstly, it decouples the generalized inner product of latent positions in the random graph model. Secondly, it relaxes the complicated domain of the original likelihood function to the entire Euclidean space. Leveraging these features, we design a computationally efficient Gaussian variational inference algorithm via stochastic gradient descent. Furthermore, we show the asymptotic efficiency of the maximum extended surrogate likelihood estimator and the Bernstein-von Mises limit of the variational posterior distribution. 
Through extensive numerical studies, we demonstrate the usefulness of the proposed SANVI algorithm compared to the classical Markov chain Monte Carlo algorithm, including comparable estimation accuracy for the latent positions and less computational costs.
\end{abstract}

% \end{keyword}

% \tableofcontents
\section{Introduction} % (fold)
\label{sec:introduction}

Using graphs, a mathematical abstraction of real-world networks, to represent relational data, with the vertices denoting entities and the edges encoding relationships between connected entities, has been attracting attention in a broad range of applications, such as social networks \citep{girvan-newman-2002, wasserman-faust-1994-social, young-scheinerman-2007-rdpg}, biological networks \citep{girvan-newman-2002, tang-ketcha-badea-2019}, and computer networks \citep{neil-uphoff-hash-2013-6623779, rubin-adams-heard-2016-7745482}, among others.
Network analysis also connects to other fields beyond statistics, including computer science, machine learning, probability, and physics.
A variety of network models that are conformable to statistical analyses have been developed, including the renowned stochastic block model \citep{holland-1983-sbm} as well as its offspring \citep{airoldi-blei-fienberg-2008, karrer-newman-2011-sbm, lyzinski-tang-athreya-2017}, the (generalized) random dot product graph model \citep{rubin-cape-tang-2022, young-scheinerman-2007-rdpg}, the latent space model \citep{hoff-raftery-handcock-2002}, exchangeable random graphs \citep{caron-fox-2017, lei-2021-graphroot}, and graphons \citep{lovasz-2012-large}.
Meanwhile, there has also been substantial progress on the subsequent inference tasks for the latent structures of network models, such as community detection \citep{abbe-2018-cdsbm, abbe-bandeira-hall-2016-sbm, lei-rinaldo-2015-sbm, sussman-tang-fishkind-2012}, vertex classification \citep{sussman-tang-priebe-2014, tang-sussman-priebe-2013}, and network hypothesis testing \citep{lei-2016-testsbm, tang-athreya-sussman-2017-sptest, tang-athreya-sussman-2017-nptest}.

In this paper, we focus on the generalized random dot product graph (GRDPG). Informally, GRDPG assigns each vertex a low-dimensional vector called the latent position, and the connection probability between any pair of vertices is given by the generalized inner product of the associated latent positions. We defer the formal definition to Section \ref{subsec:grdpg}. GRDPG has been attracting attention because it not only has a simple low-rank structure but also is versatile as it encompasses several popular network models, such as stochastic block models \citep{holland-1983-sbm}, degree-corrected stochastic block models \citep{karrer-newman-2011-sbm}, mixed membership stochastic block models \citep{airoldi-blei-fienberg-2008}, and degree-corrected mixed membership models \citep{jin-ke-luo-2023}. GRDPG also provides building blocks for approximating general latent position random graphs \citep{lei-2021-graphroot, tang-sussman-priebe-2013}. 

Graph data is usually represented in the form of an adjacency matrix. Due to the low expected rank of the adjacency matrix generated from a GRDPG, spectral methods have been widely applied in statistical analysis of graph data, among which the adjacency spectral embedding (ASE) is a popular one. The random dot product graph (RDPG) community has been developing theory and methods based on ASE. The readers are referred to 
\cite{athreya-priebe-tang-2016, athreya-tang-park-2021, koo-tang-trosset-2023, levin-levina-2025, levin-roosta-tang-2021, li-levina-zhu-2020, lyzinski-sussman-tang-2014, rubin-cape-tang-2022, sengupta-chen-2017, sussman-tang-priebe-2014, sussman-tang-fishkind-2012, tang-athreya-sussman-2017-sptest, tang-athreya-sussman-2017-nptest, tang-priebe-2018, xie-2023-euclidean, xie-2024-spn-bernoulli, xie-wu-2023-eigen, xie-xu-2020-bayes, xie-xu-2023-os, young-scheinerman-2007-rdpg} for an incomplete list of references. 
However, it has also been observed in \cite{xie-xu-2020-bayes} that spectral estimators do not take advantage of the likelihood information of the network adjacency matrix, and likelihood-based methods for (generalized) RPDG are comparatively unexplored. This research theme aims to develop a novel likelihood-based method for learning GRDPG that is computationally efficient, numerically stable for finite-sample problems, and theoretically solid and optimal.

Recently, \cite{xie-xu-2023-os} discovered a striking fact: Spectral estimators are sub-optimal for estimating the latent positions due to the negligence of the graph likelihood structure. Specifically, \cite{xie-xu-2023-os} proposed a one-step estimator (OSE) that absorbs the network likelihood information and established that OSE improves upon ASE. Better estimation of the latent positions is not only interesting by itself but also useful for more effective subsequent inference methods, such as more powerful hypothesis testing of the equality of latent positions \citep{xie-2024-spn-bernoulli} or membership profiles in mixed membership models \citep{fan-fan-han-2022}.
% Formally, in \cite{xie-xu-2023-os}, the authors established the following results. For the adjacency spectral embedding $\widetilde{\bX} = [\widetilde{\bx}_1,\ldots,\widetilde{\bx}_n]\transpose$ and its one-step refinement $\widehat{\bX}^{(\mathrm{OS})} = [\widehat{\bx}_1^{(\mathrm{OS})},\ldots,\widehat{\bx}_n^{(\mathrm{OS})}]\transpose$, it has been showed that $\sqrt{n}(\bW\transpose\widetilde{\bx}_i - \bx_i)\overset{\calL}{\to}\mathrm{N}(\zero_d, \bSigma_i)$ and $\sqrt{n}(\bW\transpose\widehat{\bx}_i^{(\mathrm{OS})} - \bx_i)\overset{\calL}{\to}\mathrm{N}(\zero_d, \bG_i^{-1})$, but $\bSigma_i - \bG_i^{-1}$ is positive semidefinite, where $\bW\in\mathbb{O}(d)$ is the alignment matrix accounting for the orthogonal non-identifiability.
% Namely, the OSE has smaller asymptotic variance than the ASE in the vertex-wise sense. Furthermore, in \cite{xie-xu-2023-os}, it has been showed that $\|\widetilde{\bX}\bW - \bX\|_{\mathrm{F}}\geq \|\widehat{\bX}^{(\mathrm{OS})}\bW - \bX\|_{\mathrm{F}}$ asymptotically, \emph{i.e.}, the OSE has smaller global error than the ASE. These two observations establish the sub-optimality of the spectral embedding both in the vertex-wise sense and in the global sense.

Despite the large-sample optimality, OSE typically requires comparatively large network sizes to outperform ASE \citep{xie-xu-2023-os}. For small-network problems, OSE can be numerically unstable because the estimated Hessian matrix may contain negative eigenvalues. Subsequently, \cite{wu-xie-2022-sl} developed a Bayesian method for RDPG based on a cleverly-designed surrogate likelihood that retains more likelihood information than OSE does.
The Bayes estimate based on the surrogate likelihood is not only asymptotically efficient but also exhibits superior numerical stability compared to OSE and ASE, even for moderately small network sizes.

Nevertheless, the Bayesian methods, although theoretically solid and numerically competitive, are practically inconvenient. This is largely due to the expensive computational cost associated with Markov chain Monte Carlo (MCMC) sampling methods. Compared to classical MCMC methods, variational inference (VI) methods \citep{blei-kucukelbir-mcAuliffe-2017}
 % \citep{wang-blei-2019, han-yang-2019} 
have emerged as a popular alternative. Unlike MCMC, VI is optimization-based, which tends to be faster while still having comparable numerical performance. VI methods have been gaining rapid development recently. Readers are referred to \cite{bhattacharya-pati-yang-2025, katsevich-rigollet-2024, zhang-yang-2024-jrsssb, wang-blei-2019, han-yang-2019, hinton-camp-1993, jordan-ghahramani-jaakkola-1998, peterson-anderson-1987} and references therein for the recent advances of VI methods in general. For VI methods in the context of network models, \cite{loyal2024fast} and \cite{JMLR:v25:22-0514} developed structured mean-field VI methods for dynamic networks that were built upon latent space models, and they do not apply to the GRDPG framework. 

In this paper, we propose a computationally efficient spectral-assisted network variational inference (SANVI) method through a pivotal extended surrogate likelihood (ESL) function in the context of GRDPG. Given a fixed vertex, SANVI minimizes the Kullback--Leibler (KL) divergence between a candidate Gaussian distribution and the posterior distribution of the latent position of interest, where the posterior distribution is computed based on the ESL function. Note that the algorithm can be parallelized thanks to the separable structure of the ESL function. We also establish the corresponding large sample theory, including the asymptotic efficiency of the proposed estimator and the Bernstein-von Mises theorem of the variational posterior distribution, thereby generalizing and popularizing the existing framework in \cite{wu-xie-2022-sl}.

The remaining part of this paper is structured as follows. In Section \ref{sec:grdpg-esl}, we review the background of GRDPG and ASE and introduce the ESL function. In Section \ref{sec:gaussian-vb}, we introduce SANVI based on the Gaussian VI by leveraging the ESL function, establish the asymptotic properties of the variational posterior distribution, and discuss the stochastic gradient descent algorithm for the computation of SANVI. In Section \ref{sec:numeric-ex}, we demonstrate the empirical finite-sample performance of SANVI through some simulated examples and the analysis of a real-world network dataset. We conclude the paper with a discussion in Section \ref{sec:discussion}.

\noindent
\textit{Notations}: Most of the notations that we mainly use in this paper are explained in the following.
The notation $[n]$ stands for the set of consecutive integers from $1$ to $n$, that is, $[n]=\{1,\ldots,n\}$.
The symbol $\lesssim$ means an inequality up to a constant, that is, $a\lesssim b$ if $a\leq C b$ for some constant $C>0$.
The constant $C$ can depend on some other constants, of which we use subscripts to denote the dependency, e.g., $C_{\delta,\lambda}$ showing the dependency of $C$ on $\delta$ and $\lambda$. A similar definition also applies to the symbol $\gtrsim$.
The notation $\|\bx\|_2$ denotes the Euclidean norm of a vector $\bx=[x_1,\ldots,x_d]\transpose\in\mathbb{R}^d$, that is, $\|\bx\|_2=(\sum_{k=1}^d x_k^2)\halfpower$.
The $d\times d$ identity matrix is denoted by $\eye_d$. The notation $\mathbb{O}(n,d)=\{\bU\in\mathbb{R}^{n\times d}:\bU\transpose\bU=\bI_d\}$ denotes the set of all orthonormal $d$-frames in $\mathbb{R}^n$, where $d\leq n$, and we write $\mathbb{O}(d) = \mathbb{O}(d,d)$.
For a matrix $\bX=[x_{ik}]_{n\times d}$, $\sigma_k(\bX)$ denotes its $k$th largest singular value.
Matrix norms with following definitions are used: the spectral norm $\|\bX\|_2 = \sigma_1(\bX)$, the Frobenius norm $\|\bX\|\frobenius = (\sum_{i=1}^n\sum_{k=1}^d x_{ik}^2)\halfpower$, the matrix infinity norm $\|\bX\|_\infty=\max_{i\in[n]}\sum_{k=1}^d|x_{ik}|$, and the two-to-infinity norm $\|\bX\|_{2\to\infty}=\max_{i\in[n]}(\sum_{k=1}^d x_{ik}^2)\halfpower$. In particular, these norm notations apply to any Euclidean vector $\bx\in\mathbb{R}^d$ viewed as a $d\times 1$ matrix. Given two symmetric positive semidefinite matrices $\bA,\bB$ of the same dimension, we write $\bA\preceq\bB$ ($\bA\succeq \bB$, respectively) if $\bB - \bA$ ($\bA - \bB$, respectively) is positive semidefinite. 
For a vector $\bx\in\mathbb{R}^d$, the notation $[\bx]_k=x_k$ denotes its $k$th coordinate.
For a matrix $\bX\in\mathbb{R}^{n\times d}$, the notation $\bX_{i*}$ denotes its $i$th row, $\bX_{*j}$ its $j$th column, and $x_{ij}$ its $(i,j)$th entry.
We use $\{\bW_n\}_{n=1}^\infty$ to denote the sequence of orthogonal matrices aligning a sequence of estimators $\{\widehat\bX_n\}_{n=1}^\infty$ and the true value, and we may drop the subscript $n$ for simplicity of notation.

\section{Generalized random dot product graphs and the extended surrogate likelihood}
\label{sec:grdpg-esl}

\subsection{Background on generalized random dot product graphs}
\label{subsec:grdpg}
We begin by briefly reviewing GRDPG and ASE.
\begin{definition}
[Generalized random dot product graph]
Let $n,d\in\mathbb{N}_+$, $n\geq d$, $p,q\in\mathbb{N}$ with $p + q = d$, and $\eye_{p, q} = \mathrm{diag}(1,\ldots,1,-1,\ldots,-1)$ with $p$ positive ones followed by $q$ negative ones on its diagonal.
Given an $n\times d$ matrix $\bX = [\bx_1,\ldots,\bx_n]\transpose$, where $\bx_1,\ldots,\bx_n\in\mathbb{R}^d$, with first $p$ columns orthogonal to last $q$ columns, such that $\bx_i\transpose\eye_{p,q}\bx_j\in [0, 1]$ for all $i, j\in [n]=\{1,\ldots,n\}$, we say that $\bA = [A_{ij}]_{n\times n}$ is the adjacency matrix of a generalized random dot product graph, denoted as $\bA\sim\mathrm{GRDPG}(\bX)$ with signature $(p, q)$ if $A_{ij}\sim\mathrm{Bernoulli}(\bx_i\transpose\eye_{p, q}\bx_j)$ independently for all $i\leq j$, and $A_{ij} = A_{ji}$ if $i > j$.
The matrix $\bX$ is referred to as the latent position matrix, and the $d$-dimensional vector $\bx_i$ is referred to as the latent position of vertex $i$. When $q = 0$, a GRDPG is also called a random dot product graph (RDPG).
\end{definition}

% \begin{remark}[Deterministic versus stochastic latent positions]
In this paper, we consider the latent positions $\bx_1,\ldots,\bx_n$ to be deterministic parameters to be estimated.
Another slightly different modeling approach is to consider $\bx_1,\ldots,\bx_n$ as independent and identically distributed latent random variables following some distribution $F$ supported on the latent space $\calX$ (see, for example, \citealp{athreya-priebe-tang-2016, tang-athreya-sussman-2017-nptest, tang-priebe-2018}).
This random formulation of the latent positions introduces implicit homogeneity and is connected to the infinite exchangeable random graphs \citep{janson-2008-graph}.
The same homogeneity condition was retained in \cite{xie-xu-2023-os} using a Glivenko--Cantelli type condition when $\bx_1,\ldots,\bx_n$ are deterministic.
The latter Glivenko--Cantelli type condition is also relaxed in the current work, as we only require that $\sigma_d(\bX) > 0$ (see Remark \ref{remark:identifiability} below).
% \end{remark}

\begin{remark}[Nonidentifiability]
\label{remark:identifiability}
For convenience, in this work, we follow the setup in \cite{xie-2024-spn-bernoulli} and require the first $p$ columns of the latent position matrix $\bX$ to be orthogonal to the last $q$ columns. For GRDPG with more general latent position matrices, please see \cite{rubin-cape-tang-2022}.
The latent position matrix $\bX$ is not uniquely identified in the following two senses. First, any low-rank  connection probability matrix $\bP = \bX\eye_{p,q}\bX\transpose$ can have different factorizations because for any orthogonal matrix with a $(p,q)$ block structure, $\bW=\mathrm{diag}(\bW_p,\bW_q)$, where $\bW\in\mathbb{O}(d)$, $\bW_p\in\mathbb{O}(p)$, $\bW_q\in\mathbb{O}(q)$, we have $\bX\eye_{p,q}\bX\transpose = (\bX\bW)\eye_{p,q}(\bX\bW)\transpose$. Second, for any $d'>d$ and any latent position matrix $\bX\in\mathbb{R}^{n\times d}$, there exists another matrix $\bX'\in\mathbb{R}^{n\times d'}$ such that $\bX\eye_{p,q}\bX\transpose=\bX'\eye_{p+(d' - d),q}(\bX')\transpose$. The latter source of non-identifiability can be removed by requiring that $\sigma_d(\bX)>0$, while the former source is inevitable without further constraints. Thus, any estimator of the latent position matrix $\bX$ can only recover it up to an orthogonal transformation.
\end{remark}

We consider undirected and unweighted graphs, so the adjacency matrices are binary and symmetric. We allow self-loops, so the adjacency matrices may have non-zero diagonal elements. One convenient feature of GRDPG is that the edge probability matrix $\bP=\expect\bA = \bP = \bX\eye_{p,q}\bX\transpose$ is low-rank. This motivates spectral decomposition methods for learning the latent position matrix $\bX$ \citep{rubin-cape-tang-2022,sussman-tang-fishkind-2012}.

\begin{definition}[Adjacency spectral embedding]
Given $\bA\sim\mathrm{GRDPG}(\bX)$, let $\bA$ yield spectral decomposition $\bA = \sum_{i = 1}^n\lambda_i(\bA)\widehat{\bu}_i\widehat{\bu}_i\transpose$, where $|\lambda_1(\bA)|\geq\ldots\geq|\lambda_n(\bA)|$, arranged in decreasing order of absolute value, are the eigenvalues of $\bA$, $\widehat{\bu}_i$ is the eigenvector associated with $\lambda_i(\bA)$, $\widehat{\bu}_i\transpose\widehat{\bu}_j = 0$ for all $i\neq j$, and $\|\bu_i\|_2 = 1$ for all $i\in [n]$.
Pick the first $d$ eigenvalues and the correspongding eigenvectors, and rearrange them in the decreasing order of the eigenvalues as real numbers, $\lambda_{k_1}(\bA)\geq\ldots\geq\lambda_{k_d}(\bA)$.
Then the adjacency spectral embedding of $\bA$ into $\mathbb{R}^{n\times d}$ is defined as $\breve{\bX} = [\breve{\bx}_1,\ldots,\breve{\bx}_n]\transpose = \bU_\bA|\bS_\bA|^{1/2}$, where $\bU_\bA = [\widehat{\bu}_{k_1},\ldots,\widehat{\bu}_{k_d}]$, $\bS_\bA = \mathrm{diag}\{\lambda_{k_1}(\bA),\ldots,\lambda_{k_d}(\bA)\}$, and $|\bS_\bA| = \mathrm{diag}\{|\lambda_{k_1}(\bA)|,\ldots,|\lambda_{k_d}(\bA)|\}$.
Also, the signature-adjusted adjacency spectral embedding of $\bA$ into $\mathbb{R}^{n\times d}$ is defined as $\widetilde{\bX} = [\widetilde{\bx}_1,\ldots,\widetilde{\bx}_n]\transpose = \bU_\bA|\bS_\bA|^{1/2}\mathrm{sgn}(\bS_\bA)$, where $\mathrm{sgn}(\cdot)$ is the sign function and $\mathrm{sgn}(\bS_\bA)$ applies entrywise on the diagonals of $\bS_\bA$.
\end{definition}

% \begin{example}[Stochastic block model]
% Random dot product graphs have connections with the popular stochastic block model \cite{holland-1983-sbm}.
% Consider a graph with $n$ vertices that are partitioned into $K$ communities, where $K$ is assumed to be much smaller than $n$.
% Let $\tau:[n]\to[K]$ be a cluster assignment function that assigns each vertex to a unique community. Let $\bB=[B_{kl}]_{K\times K}\in(0,1)^{K\times K}$ be a symmetric probability matrix and $A_{ij}$ be the binary indicator of the existence of an edge between vertices $i$ and $j$.
% Then the stochastic block model specifies that $A_{ij}\sim \mathrm{Bernoulli}(B_{\tau(i)\tau(j)})$ independently for all $i,j\in[n]$, $i \leq j$, and $A_{ij}=A_{ji}$ for all $i > j$.
% By converting the community assignment to a matrix $\bZ=[1\{\tau(i)=k\}]_{n\times K}$, we see that the expected adjacency matrix $\bZ\bB\bZ\transpose$ is symmetric and of low rank.
% Furthermore, if $\bB$ has rank $d\leq K$, with $p$ positive eigenvalues and $q$ negative eigenvalues, and can be factorized as $\bB=\bV\eye_{p,q}\bV\transpose$ for a $K\times d$ matrix $\bV$, then $\bA$ can be seen as an adjacency matrix generated by the generalized random dot product graph with latent position matrix $\bX=\bZ\bV$ and signature $(p,q)$, that is, $\bA\sim\mathrm{GRDPG}(\bZ\bV)$ with signature $(p,q)$.
% \end{example}

\subsection{The extended surrogate likelihood function}
\label{subsec:esl}

We now introduce the extended surrogate likelihood function for GRDPG.
The motivation is that the exact likelihood function has a complicated structure, bringing challenges for developing the theory and computation of maximum likelihood estimation.
The difficulty partially comes from the fact that GRDPG belongs to a curved exponential family, and the theory of the maximum likelihood estimation is much more difficult in curved exponential families than in canonical ones (see, for example, Section 2.3 in \citealp{bickel-doksum-2015}). 

% For the sake of generality, we introduce an $n$-dependent sparsity factor $\rho_n\in(0, 1]$ and consider $\bA\sim\mathrm{GRDPG}(\rho_n^{1/2}\bX)$ throughout. The sparsity factor $\rho_n$ governs the average expected degree of $\bA\sim\mathrm{GRDPG}(\rho_n^{1/2}\bX)$ through the quantity $n\rho_n$. 
Consider the log-likelihood function of $\bA\sim\mathrm{GRDPG}(\bX)$:
\[
\ell_\bA(\bX) =
\sum_{1\leq i\leq j\leq n}\{A_{ij}\log(\bx_i\transpose\eye_{p,q}\bx_j) + (1-A_{ij})\log(\bx_i\transpose\eye_{p,q}\bx_j)\}.
\]
The parameter space is defined by $\{\bX = [\bx_1,\ldots,\bx_n]\transpose\in\mathbb{R}^{n\times d}:0 < \bx_i\transpose\eye_{p,q}\bx_j < 1\text{ for all }i,j\}$, which is a complicated set whose boundary renders the maximum likelihood estimation intractable, both computationally and analytically. In addition, the log-likelihood function has an unbounded gradient over the boundary.

For the sake of generality, we introduce an $n$-dependent sparsity factor $\rho_n\in(0, 1]$ that governs the average expected degree of GRDPG through the quantity $n\rho_n$. Note that by taking $\rho_n\to 0$ as $n\to\infty$, we allow the modeling of sparse random graphs that are more practical for real-world network data. 
To distinguish a generic latent position $\bx_i\in\mathbb{R}^d$ and its true value associated with the data generating distribution, let $\rho_n^{1/2}\bx_{0i}$ denote the ground truth of $\bx_i$, $i\in[n]$, and $\bX_0 = [\bx_{01},\ldots,\bx_{0n}]\transpose$.
We first consider the log-likelihood function of a single $\bx_i$ when the remaining latent positions $\{\bx_{0j}\}_{j\neq i}$ are accessible:
\begin{equation}
\label{eqn:oracle-loglik}
\begin{aligned}
\ell_{0in}(\bx_i)
&= \sum_{j\neq i}^n\{A_{ij}\log(\rho_n^{1/2}\bx_i\transpose\eye_{p,q}\bx_{0j}) + (1 - A_{ij})\log(1 - \rho_n^{1/2}\bx_i\transpose\eye_{p,q}\bx_{0j})\}\\
&\quad + \{A_{ii}\log(\bx_i\transpose\eye_{p,q}\bx_i) + (1 - A_{ii})\log(1 - \bx_i\transpose\eye_{p,q}\bx_i)\}.
\end{aligned}
\end{equation}
We refer to $\ell_{0in}(\bx_i)$ in \eqref{eqn:oracle-loglik} as the oracle log-likelihood function.
% because it requires the true values of the remaining $\bx_j$'s with $j\neq i$. 
Theorem 2 in \cite{xie-xu-2023-os} established the consistency and asymptotic normality of the maximizer of the oracle log-likelihood function $\ell_{0in}(\bx_i)$ in \eqref{eqn:oracle-loglik}. Nevertheless, the oracle log-likelihood is not computable because $\{\bx_{0j}\}_{j\neq i}$ are not accessible in practice.
Following the idea in \cite{wu-xie-2022-sl}, we replace the unknown latent positions together with the signature by the corresponding rows of the signature-adjusted adjacency spectral embedding. Formally, let $\widetilde{\bx}_j$ be the $j$th row of the signature-adjusted adjacency spectral embedding $\widetilde{\bX}$, $j\in [n]$. Then, we obtain the following approximation to the oracle log-likelihood:
\begin{align}
\label{eqn:loglik-approximation}
\ell_{0in}(\bx_i) \approx \sum_{j = 1}^n\{A_{ij}\log(\bx_i\transpose\widetilde{\bx}_j) + (1 - A_{ij})\log(1 - \bx_i\transpose\widetilde{\bx}_j)\}.
\end{align}
Note that the last term in $\ell_{0in}$ is replaced by $A_{ii}\log(\bx_i\transpose\widetilde{\bx}_i) + (1 - A_{ii})\log(1 - \bx_i\transpose\widetilde{\bx}_i)$ for convenience, which is asymptotically unimportant. The above approximation can be made precise by the uniform consistency of the adjacency spectral embedding: There exists a $d\times d$ orthogonal $\bW$ such that $\|\widetilde{\bX}\bW - \rho_n^{1/2}\bX_0\|_{2\to\infty} = O\{\sqrt{(\log n)/n}\}$ with high probability \citep{lyzinski-sussman-tang-2014, xie-2024-spn-bernoulli}.

The complication of the oracle log-likelihood function primarily comes from the constraint that $\bx_i\transpose\bx_j\in(0,1)$ for all $i,j\in[n]$. Nonetheless, the approximation step \eqref{eqn:loglik-approximation} does not fully resolve it. 
% For the approximation of the oracle log-likelihood function, although separate in $i\in[n]$ so that the latent position of each vertex can be estimated separately, there is still a similarly constraint for each $i\in[n]$, namely, $\bx_i\transpose\widetilde\bx_j\in(0,1)$ for all $j\in[n]$.
To address this technical challenge, \cite{wu-xie-2022-sl} proposed a surrogate likelihood method for learning the latent position matrix $\bX$ for RDPG 
% Specifically, the local log-likelihood function for the latent position of a single vertex, $\bx_i$, is
% $\ell_{in}(\bx_i) = \sum_{j = 1}^n\{A_{ij}\log(\bx_i\transpose\bx_j) + (1 - A_{ij})\log(1 - \bx_i\transpose\bx_j)\}
% $, $i\in [n]$.
% Nonetheless, the feasible set $\{\bx_i\in\mathbb{R}^d:\bx_i\transpose\bx_j\in [0, 1],j\in [n]\}$ is cumbersome for both theoretical analysis and practical computation, and it is therefore desirable to design a more tractable approximation to the log-likelihood with a convenient feasible set for $\bx_i$.
% To achieve this, in \cite{wu-xie-2022-sl}, the authors first replaced $\bx_j$ with the $j$th row of the ASE $\widetilde{\bx}_j$, $j\in [n]$.
% The resulting local approximation to $\ell_{in}(\bx_i)$ only depends on $\bx_i$.
by applying a Taylor's expansion of the term $\log(\bx_i\transpose\widetilde{\bx}_j)$ (note that in RDPG, $\eye_{p,q} = \eye_p$), such that the constraint $\bx_i\transpose\bx_j\in (0, 1)$ is relaxed to $\|\bx_i\|_2\leq 1$ due to the uniform consistency of ASE. The surrogate log-likelihood function for the entire graph is formed by taking the sum of the individual surrogate log-likelihood functions for each vertex. In particular, by doing so, the surrogate log-likelihood has a separable structure for each individual latent position $\bx_i$ and provides immediate convenience for both theoretical analysis and practical computation. 
% \cite{wu-xie-2022-sl} further established the large sample properties of the frequentist and Bayes estimators and developed a stochastic gradient descent algorithm for the corresponding computation. 

\cite{wu-xie-2022-sl} observed that Bayesian methods are typically comparable and sometimes outperform the frequentist point estimator for RPDG, such as ASE and OSE. Additionally, Bayesian methods offer a natural and principled approach for uncertainty quantification. One disadvantage of the Bayesian method proposed in \cite{wu-xie-2022-sl} is the computational expense due to the nature of MCMC. This practical inconvenience motivates us to develop a computationally efficient VI method for GRDPG and to provide it with the necessary theoretical guarantee. 

The surrogate likelihood proposed \cite{wu-xie-2022-sl} is only defined on a compact subset $\prod_{i = 1}^n\{\bx_i\in\mathbb{R}^d:\|\bx_i\|_2\leq 1\}$ of the Euclidean space. In this work, we adopt the Gaussian VI and take the variational distribution family to be the space of all Gaussian distributions. The detailed formulation is deferred to Section \ref{sec:gaussian-vb}, but one requirement of Gaussian VI is that the target posterior distribution needs to be supported over $\mathbb{R}^d$. Hence, extending the feasible set of the surrogate likelihood to the entire Euclidean space is necessary.
In addition, the surrogate likelihood derivation relies on Taylor's expansion argument, which is necessary to drop the constraint that $\bx_i\transpose\widetilde{\bx}_j > 0$ for all $j\in [n]$.
However, this approximation step could still be rough in moderate and small network problems.
The extended surrogate likelihood to be developed will address the above issues while preserving the attractive features of the surrogate likelihood, including separability and log-concavity. 

For $\bA\sim\mathrm{GRDPG}(\bX)$, the local log-likelihood for a single latent position $\bx_i$ is
\[
\ell_{in}(\bx_i;\{\bx_j\}_{j\neq i}) = \sum_{j = 1}^n\{A_{ij}\log(\bx_i\transpose\eye_{p,q}\bx_j) + (1 - A_{ij})\log(1 - \bx_i\transpose\eye_{p,q}\bx_j)\}.
\]
Rather than using the $j$th row of ASE $\breve{\bx}_j$ directly, we use the signature-adjusted adjacency spectral embedding, $\widetilde{\bx}_j=\mathrm{sgn}(\bS_\bA)\breve{\bx}_j$, to replace $\eye_{p, q}\bx_j$, since $(p, q)$ may be unknown in practice.
% It is clear that $\mathrm{sgn}(\bS_\bA)$ is an ideal proxy for the unknown $\eye_{p, q}$ because the signature $(p, q)$ agrees with the signs of the eigenvalues of $\bP = \bX\eye_{p, q}\bX\transpose$ and the eigenvalues of $\bA$ concentrate around the eigenvalues of $\bP$ by Weyl's inequality. 
Next, observe that the functions $\log(t)$ and $\log(1 - t)$ are both well defined over $[\tau, 1 - \tau]$ for a small threshold $\tau > 0$ , but not the entire $\mathbb{R}$.
It is desirable that these functions can be extended beyond this interval while certain regularities, such as differentiability and smoothness, are preserved.
For this purpose, let $\tau_n$ be a small positive number that may depend on $n$, and we define $\psi_n(t)$ with the following properties: $\psi_n(t) = \log(t)$ on $[\tau_n,1]$; $\psi_n$ is a quadratic function for $t<\tau_n$ and $t>1$; $\psi_n$ is twice continuously differentiable over $\mathbb{R}$.
Formally,
\begin{align*}
\psi_n(t) &=
\left\{\begin{aligned}
&\log(t),&\quad&\text{if }\tau_n\leq t \leq 1,\\
&-{t^2}/{(2\tau_n^2)} + {2t}/{\tau_n} + (\log \tau_n - 3/2),&\quad&\text{if }t \leq \tau_n,\\
&-{t^2}/{2} + 2t - {3}/{2},&\quad&\text{if }t > 1.
\end{aligned}\right.
\end{align*}
% and it is easy to see that
% \begin{align*}
% \psi_n(1-t) &=
% \left\{\begin{aligned}
% &\log(1 - t),&\quad&\text{if }0\leq t\leq 1 - \tau_n,\\
% &-{t^2}/{(2\tau_n^2)} + {(1 - 2\tau_n)t}/{\tau_n^2} + \log \tau_n + {(3\tau_n - 1)(1 - \tau_n)}/{(2\tau_n^2)},&\quad&\text{if }t > 1 - \tau_n,\\
% &-\frac{t^2}{2} - t,&\quad&\text{if }t < 0.
% \end{aligned}\right.
% \end{align*}
% Let $\log(t)$ take the value $-\infty$ on $t\in(-\infty,0]$, then we can also write
% \[
% \psi_n(t) = \max\left\{-\frac{t^2}{2\tau_n^2} + \frac{2t}{\tau_n} + (\log\tau_n - 3/2),\, \log(t)\right\}
% \]
With these modifications of $\log(t)$ and $\log(1-t)$, we then define the local extended surrogate log-likelihood (ESL) function for the latent position of a single vertex $\bx_i$ for GRDPG as
\begin{align}\label{eqn:def-esl}
\widehat{\ell}_{in}(\bx_i) = \sum_{j = 1}^n\{A_{ij}\psi_n(\bx_i\transpose\widetilde{\bx}_j) + (1 - A_{ij})\psi_n(1-\bx_i\transpose\widetilde{\bx}_j)\}.
\end{align}
The global ESL function for the entire graph is defined as $\widehat{\ell}_n(\bX) = \sum_{i = 1}^n\widehat{\ell}_{in}(\bx_i)$. Here, the term ``extended'' means that the domain of the target function \eqref{eqn:def-esl} is extended to the entire $\mathbb{R}^d$ without constraint, and the term ``surrogate'' means that the unknown latent positions are replaced by their signature-adjusted ASE.

\section{Spectral-assisted network variational inference}
\label{sec:gaussian-vb}
We now leverage the ESL function to develop SANVI. To begin with, we first consider the posterior distribution of the latent positions $\bx_1,\ldots,\bx_n$ associated with the ESL function \eqref{eqn:def-esl}. Note that the spectral assistance occurs directly in the formulation of \eqref{eqn:def-esl} since ($\widetilde{\bx}_1,\ldots,\widetilde{\bx}_n$) are spectral estimators. 

It should be noted that we do not use the exact likelihood for the entire graph, as it is difficult to analyze and not separable in $i\in[n]$. Additionally, we do not use the oracle likelihood, since the true latent positions are unknown.
Instead, we substitute the signature-adjusted ASE for the unknown latent positions together with the signature in the oracle likelihood (hence the term ``surrogate'').
This idea of using a general statistical criterion function to replace the likelihood in the Bayes formula when the exact likelihood function is not available or intractable for analysis or computation is not new, and among the literature, an influential work is \cite{chernozhukov-hong-2003}. 

Formally, given independent prior distributions with densities $\pi_i(\bx_i)$ over $\mathbb{R}^d$, $i\in [n]$ for the latent positions $\bx_1,\ldots,\bx_n$, the posterior distribution of the latent position $\bx_i$ of a single vertex $i$ associated with the ESL function \eqref{eqn:def-esl} has the following density function up to a normalizing constant:
\begin{align}\label{eqn:esl-posterior}
\pi_{in}(\bx_i\mid\bA) \propto \exp\{\widehat{\ell}_{in}(\bx_i)\}\pi_i(\bx_i).
% \pi_{in}(\bx_i\mid\bA) = \frac{\exp\{\widehat{\ell}_{in}(\bx_i)\}\pi_i(\bx_i)}{\int_{\mathbb{R}^d}\exp\{\widehat{\ell}_{in}(\bx_i)\}\pi_i(\bx_i)\mathrm{d}\bx_i}.
\end{align}
The joint posterior density of the latent position matrix $\bX$ of the entire graph takes the product form $\pi_n(\bX\mid\bA) = \prod_{i = 1}^n\pi_{in}(\bx_i\mid\bA)$ thanks to the separable structure of $\widehat{\ell}_n(\bX)$.
The exact computation of the posterior distribution in \eqref{eqn:esl-posterior} typically relies on MCMC and is generally inconvenient, even though the separable structure permits parallelization. Instead, we resort to VI methods and focus on the Gaussian VI. 
% The survey \cite{blei-kucukelbir-mcAuliffe-2017} provides a review of the variational Bayes in general.

The goal of VI is to find a distribution $q(\bX)\in\calQ$ for the latent position matrix $\bX$, where $\calQ$ is a collection of candidate distributions over $\bX$ that are tractable to compute, such that the Kullback-Leibler (KL) divergence between $q(\bX)$ and the posterior distribution $\pi_n(\bX\mid\bA)$ is minimized. Formally, VI solves $\min_{q\in\calQ}D_{\mathrm{KL}}(q(\cdot)\|\pi_n(\cdot\mid\bA))$, where $D_{\mathrm{KL}}$ denotes the KL divergence. Since $\pi_n(\bX\mid\bA)$ factorizes as $\prod_i\pi_{in}(\bx_i\mid\bA)$, it is also reasonable to require that $\calQ$  reduces to the class of all product distributions of $\bx_i$'s: $\calQ = \{\prod_{i = 1}^nq_i(\bx_i):q_i\in\calP\}$, where $\calP$ is some distribution class for $\bx_i$. 

Specialized to the Gaussian VI, we take $\calP$ as the class of all $d$-dimensional (non-degenerate) multivariate Gaussian distributions, namely,
\[
\calQ = \left\{\prod_{i = 1}^n\calN(\bx_i\mid\bmu_i, \bSigma_i):\bmu_i\in\mathbb{R}^d, \bSigma_i\in\mathbb{M}_+(d),i\in [n]\right\},
\]
where $\calN(\bx_i\mid\bmu_i, \bSigma_i)$ denotes the multivariate Gaussian distribution of $\bx_i$ with mean $\bmu_i$ and covariance matrix $\bSigma_i$, and $\mathbb{M}_+(d)$ denotes the class of all $d\times d$ symmetric positive definite matrices.
Notationally, we use $\phi_d(\bx_i\mid\bmu_i, \bSigma_i)$ to denote the density function of the Gaussian distribution $\calN(\bx_i\mid\bmu_i, \bSigma_i)$.
It then follows that the Gaussian VI for the posterior distribution of vertex $i$ associated with the ESL function solves the following optimization problem:
\begin{align}\label{eqn:gvb-kl-div}
\min_{\bmu_i\in\mathbb{R}^d,\bSigma_i\in\mathbb{M}_+(d)}
{D}_{\mathrm{KL}}\left(\phi_d(\bx_i\mid\bmu_i, n^{-1}\bL_i\bL_i\transpose)\|\pi_{in}(\bx_i\mid\bA)\right),\quad i \in [n].
\end{align}
We call the Gaussian distribution with parameters being the solution to the above optimization problem the variational posterior distribution, and we call its mean parameter the variational inference estimator. We next introduce the computation and the theory of SANVI.

\subsection{Computation algorithm}
\label{subsec:computation}

% We first introduce the SANVI algorithm. 
Following the idea in \cite{xu-campbell-2023} and \cite{kucukelbir-tran-ranganath-2017}, we reparameterize the covariance matrix $\bSigma_i$ of $\calN(\bx_i\mid\bmu_i, \bSigma_i)$ using the Cholesky factorization $\bSigma_i = (1/n)\bL_i\bL_i\transpose$, where $\bL_i$ is a lower triangular matrix with positive diagonal entries.
With the change of variable $\bx_i = \bmu_i + n^{-1/2}\bL_i\bz_i$ where $\bz_i\sim\mathrm{N}_d(\zero_d, \eye_d)$, a simple algebra shows that the objective function of VI is 
\begin{align*}
&D(\phi_d(\bx_i\mid\bmu_i, n^{-1}\bL_i\bL_i\transpose)\|\pi_{in}(\bx_i\mid\bA))
\\
% &= \expect_{\calN(\bx_i\mid\bmu_i, \bSigma_i)}\left\{\log\frac{\phi_d(\bx_i\mid\bmu_i, n^{-1}\bL_i\bL_i\transpose)}{\pi_{in}(\bx_i\mid\bA_{})}\right\}\\
% &= \expect_{\bz_i}\left\{\log\frac{\det(2\pi\bSigma_i)\invhalfpower\phi_d(\bz_i|\zero_d,\eye_d)}{\pi_{in}(\bmu_i + n^{-1/2}\bL_i\bz_i\mid\bA_{})}\right\}  \\
% &= \expect_{\bz_i}\left\{\log\frac{\det(2\pi\bSigma_i)\invhalfpower\exp(-\frac{1}{2}\|\bz_i\|_2^2)}{\pi_{in}(\bmu_i + n^{-1/2}\bL_i\bz_i\mid\bA_{})}\right\} \\
&= - \log\det(\bL_i) - \frac{d}{2}\log(2\pi) - \expect_{\bz_i}\bigg(\frac{1}{2}\|\bz_i\|_2^2\bigg) \\
&\quad - \expect_{\bz_i}\left\{\widehat{\ell}_{in}\left(\bmu_i + \frac{1}{\sqrt{n}}\bL_i\bz_i\right) + \log\pi_i\left(\bmu_i + \frac{1}{\sqrt{n}}\bL_i\bz_i\right)\right\} + \log(d_{in}),
\end{align*}
where $d_{in}=\int_{\mathbb{R}^d}\exp\{\widehat{\ell}_{in}(\bx_i)\}\pi_i(\bx_i)\mathrm{d}\bx_i$ is the marginal density of the data matrix $\bA$. Dropping the terms that do not depend on $\bmu_i$ and $\bL_i$, we define the Gaussian VI objective function
\begin{align}
\label{eqn:vb-objective-population}
F_{in}(\bmu_i,\bSigma_i) =
- \log\det(\bL_i) - \expect_{\bz_i}\left\{\widehat{\ell}_{in}\left(\bmu_i + \frac{1}{\sqrt{n}}\bL_i\bz_i\right) + \log\pi_i\left(\bmu_i + \frac{1}{\sqrt{n}}\bL_i\bz_i\right)\right\},
\end{align}
where $\widehat{\ell}_{in}(\bx_i)$ is the ESL function defined in \eqref{eqn:def-esl}, and $\calL$ is the class of all $d\times d$ lower-triangular matrices with positive diagonals.
Then the optimization problem \eqref{eqn:gvb-kl-div} is equivalent to
\begin{align}
\label{eqn:vb-objective-optimization}
\min_{\bmu_i\in\mathbb{R}^d, \bL_i\in\calL} F_{in}(\bmu_i,\bSigma_i).
\end{align}
Denote the noisy version of the objective function in  \eqref{eqn:vb-objective-population} by
\begin{align}
\label{eqn:vb-objective-sample}
f_{in}(\bmu_i, \bL_i, \bz_i) = - \log\det(\bL_i) - \widehat{\ell}_{in}\left(\bmu_i + \frac{1}{\sqrt{n}}\bL_i\bz_i\right) - \log\pi_i\left(\bmu_i + \frac{1}{\sqrt{n}}\bL_i\bz_i\right).
\end{align}
Then, a simple algebra shows that
\begin{equation}
\label{eqn:vb-objective-gradient}
\begin{aligned}
\frac{\partial f_{in}}{\partial\bmu_i}(\bmu_i, \bL_i, \bz_i)
&= \left.-\frac{\partial\widehat{\ell}_{in}}{\partial\bx_i}\left(\bx_i\right) - \frac{\partial}{\partial\bx_i}\log\pi_{i}\left(\bx_i\right)\right|_{\bx_i=\bmu_i + \frac{1}{\sqrt{n}}\bL_i\bz_i},\\
\frac{\partial f_{in}}{\partial\bL_i}(\bmu_i, \bL_i, \bz_i)
&= -\,\mathrm{diag}(\bL_i)^{-1} - \frac{1}{\sqrt{n}}\mathrm{tril}\left\{\frac{\partial\widehat{\ell}_{in}}{\partial\bx_i}\left(\bx_i\right)\bz_i\transpose\right\} \\
&\quad \left. - \frac{1}{\sqrt{n}}\mathrm{tril}\left\{\frac{\partial}{\partial\bx_i}\log\pi_{i}\left(\bx_i\right)\bz_i\transpose\right\}\right|_{\bx_i=\bmu_i + \frac{1}{\sqrt{n}}\bL_i\bz_i},
\end{aligned}
\end{equation}
where $\mathrm{tril}(\bB)$ replaces the upper triangular entries (excluding diagonals) of a $d\times d$ matrix $\bB$ with zeros.
% , namely,
% \begin{align*}
% \mathrm{tril}\left(
% \begin{bmatrix*}
% b_{11} & b_{12} & b_{13} & \ldots & b_{1d}\\
% b_{21} & b_{22} & b_{23} & \ldots & b_{2d}\\
% b_{31} & b_{32} & b_{33} & \ldots & b_{3d}\\
% \vdots & \vdots & \vdots &        & \vdots\\
% b_{d1} & b_{d2} & b_{d3} & \ldots & b_{dd}\\
% \end{bmatrix*}
% \right)
% = 
% \begin{bmatrix*}
% b_{11} & 0      & 0      & \ldots & 0     \\
% b_{21} & b_{22} & 0      & \ldots & 0     \\
% b_{31} & b_{32} & b_{33} & \ldots & 0     \\
% \vdots & \vdots & \vdots &        & \vdots\\
% b_{d1} & b_{d2} & b_{d3} & \ldots & b_{dd}\\
% \end{bmatrix*}
% \end{align*}

Below, 
Theorem \ref{thm:vb-objective-fun} establishes the strong convexity of $F_{in}(\bmu_i,\bL_i)$.

\begin{theorem}\label{thm:vb-objective-fun}
Suppose Assumption \ref{assumption:esl-grdpg} and Assumption \ref{assumption:prior} hold.
Then $F_{in}(\bmu_i,\bL_i)$ viewed as a function from $\mathbb{R}^d\times\calL_{d\times d}$ to $\mathbb{R}$ is strongly convex with probability at least $1-n^{-c}$ for all $n\geq N_{c,\delta,\lambda}$ depending on $c,\delta,\lambda$, and 
$\nabla_{\bmu_i,\bL_i}F_{in}(\bmu_i,\bL_i) = \expect_\bz\left[\nabla_{\bmu_i,\bL_i}f_{in}(\bmu_i,\bL_i, \bz)\right]$.
\end{theorem}

The derivative of $\log\det\bL_{i}$ with respect to $\bL_{i}$ is $\mathrm{diag}(L_{11}^{-1},\ldots,L_{dd}^{-1})$, and  $L_{kk}^{-1}$ is unbounded as $L_{kk}$ approaches 0 from the right. In practical implementation, to avoid unbounded gradients and improve the numerical stability of our gradient-based algorithm, we borrow the idea in \cite{xu-campbell-2023} to modify the gradient of $\bL_i$ as follows. Let $c_n$ be a positive number that depends on $n$, and consider the function
\[
\widetilde{h}_{n}(x)
= \left\{\begin{aligned}
& \frac{c_n}{c_nx+1}  \quad&\mbox{if }x > 0,\\
& -c_n^2x + c_n &\mbox{if }x \leq 0.
\end{aligned}\right.
\]
The function $\widetilde{h}_n(x)$ has a continuous derivative at $x=0$ and  asymptotically equals $\frac{1}{x}$ as $x$ goes to positive infinity. With this modification, the scaled gradient of $\log\det\bL_{i}$ with respect to $\bL_{i}$ is defined as the $d\times d$ diagonal matrix whose $k$th diagonal element is $\widetilde{h}_n(L_{kk})$. 
% \begin{align*}
% \widetilde{h}_{n}(\mathrm{diag}(\bL_{i}))
% =
% \begin{bmatrix*}
% \widetilde{h}_{n}(L_{11}) & & & \\
%  & \widetilde{h}_{n}(L_{22}) & & \\
%  & & \ddots & \\
%  & & & \widetilde{h}_{n}(L_{dd})\\
% \end{bmatrix*}.
% \end{align*}
Then the scaled gradient of $f_{in}$ with respect to $\bL_i$ is defined as the $d\times d$ matrix 
\begin{align*}
\widetilde{\nabla}_{\bL_i}f_{in}(\bmu_i, \bL_i, \bz_i)
&= -\,\widetilde{h}_{n}(\mathrm{diag}(\bL_i)) - \frac{1}{\sqrt{n}}\mathrm{tril}
\left.\left\{\frac{\partial\widehat{\ell}_{in}(\bx_i)}{\partial\bx_i}\bz_i\transpose
% \right\} 
% \\&\quad
 % \left. - \frac{1}{\sqrt{n}}\mathrm{tril}\left\{
+  \frac{\partial\log\pi_{i}\left(\bx_i\right)}{\partial\bx_i}\bz_i\transpose\right\}\right|_{\bx_i=\bmu_i + \frac{1}{\sqrt{n}}\bL_i\bz_i}.
\end{align*}
We adopt the Adam scheme in \cite{kingma-ba-2017-adam} to define the update step for stochastic gradient descent. See Algorithm \ref{alg:vbsl} for the detailed SANVI computation algorithm.

\begin{algorithm}[htbp]
\caption{Stochastic gradient descent for SANVI}
\label{alg:vbsl}
\begin{algorithmic}[1]
\State \textbf{Input:} The adjacency matrix $\bA = [A_{ij}]_{n\times n}$ and the embedding dimension $d$.

\State \textbf{Set:} $\tau\in(0,\frac{1}{2})$, batch size $1\leq{}s\leq{}n$, step size $\alpha_0>0$, exponential decay rates for the moments of gradients $\beta_1,\beta_2\in[0,1)$, constant $\epsilon_0=10^{-8}$.

\State Compute the spectral decomposition
$\bA = \sum_{i=1}^n \widehat\lambda_i \widehat\bu_i \widehat\bu_j\transpose$,
where $|\widehat{\lambda}_1| \geq |\widehat{\lambda}_2| \geq \ldots \geq |\widehat{\lambda}_n|$, and $\widehat{\bu}_i\transpose \widehat{\bu}_j = \mathbbm{1}(i=j)$ for all $i, j \in[n]$.

\State Compute the signature-adjusted ASE:
\[
  \widetilde{\bX} = [\widetilde{\bx}_1, \ldots, \widetilde{\bx}_n]\transpose = \left[\mathrm{sign}(\widehat{\lambda}_1)|\widehat{\lambda}_1|^{1/2}\widehat{\bu}_1, \ldots, \mathrm{sign}(\widehat{\lambda}_d)|\widehat{\lambda}_d|^{1/2}\widehat{\bu}_d\right].
  % \cdot \mathrm{diag}(|\widehat{\lambda}_1|^{1/2}, \ldots, |\widehat{\lambda}_d|^{1/2})
  % \cdot \mathrm{diag}(\mathrm{sign}(\widehat{\lambda}_1), \ldots, \mathrm{sign}(\widehat{\lambda}_d)).
\]
Let $\widetilde{p}_{ij} = \widetilde{\bx}_i\transpose\mathrm{diag}(\mathrm{sign}(\widehat{\lambda}_1), \ldots, \mathrm{sign}(\widehat{\lambda}_d))\widetilde{\bx}_j$ for all $i,j\in[n]$.

% \State Let $\widetilde{G}_{in} = \frac{1}{n}\sum_{j=1}^n \frac{\widetilde{\bx}_j\widetilde{\bx}_j\transpose}{\widetilde{p}_{ij}(1-\widetilde{p}_{ij})}$, and compute the Cholesky decomposition of its inverse $\widetilde{G}_{in}^{-1}=\widetilde{\bL}_{i}\widetilde{\bL}_{i}\transpose$, for all $i\in[n]$.

\State For $i = 1,2,\ldots,n$

\State \quad Compute the Cholesky decomposition $(\sum_{j = 1}^n\widetilde{\bx}_j\widetilde{\bx}_j\transpose/\{n\widetilde{p}_{ij}(1 - \widetilde{p}_{ij})\})^{-1}=\widetilde{\bL}_{i}\widetilde{\bL}_{i}\transpose$.

\State \quad  Set the iteration counter $t = 0$.

\State \quad Initialize gradient moments $m_{\bmu_i,1}^{(0)}=\bm{0}_d$, $m_{\bmu_i,2}^{(0)}=\bm{0}_d$, $m_{\bL_i,1}^{(0)}=\bm{0}_{d\times{}d}$, $m_{\bL_i,2}^{(0)}=\bm{0}_{d\times{}d}$.

\State \quad Initialize $\bmu_i^{(0)} = \widetilde{\bx}_i$ and $\bL_i^{(0)} = \widetilde{\bL}_{i}$.

\State \quad  While not converge

\State \quad\quad Set $t \longleftarrow t+1$.

\State \quad  \quad Sample $\bz_1^{(t-1)}, \bz_2^{(t-1)}, \ldots, \bz_{s}^{(t-1)}\sim\mathrm{N}_d(\zero_d, \eye_d)$ independently.

\State \quad\quad Compute
\begin{align*}
m_{\bmu_i,1}^{(t)} &= \beta_1 m_{\bmu_i,1}^{(t-1)} + (1-\beta_1) \frac{1}{s\sqrt{n}}\sum_{k = 1}^s\frac{\partial f_{in}}{\partial\bmu_i}(\bmu_i^{(t-1)}, \bL_i^{(t-1)}, \bz_s^{(t-1)}),\\
m_{\bL_i,1}^{(t)} &= \beta_1 m_{\bL_i,1}^{(t-1)} + (1-\beta_1) \frac{1}{s\sqrt{n}}\sum_{k = 1}^s\widetilde{\nabla}_{\bL_i}f_{in}(\bmu_i^{(t-1)}, \bL_i^{(t-1)}, \bz_s^{(t-1)}),\\
m_{\bmu_i,2}^{(t)} &= \beta_2 m_{\bmu_i,2}^{(t-1)} + (1-\beta_2) \frac{1}{s\sqrt{n}}\sum_{k = 1}^s\left\{\frac{\partial f_{in}}{\partial\bmu_i}(\bmu_i^{(t-1)}, \bL_i^{(t-1)}, \bz_s^{(t-1)})\right\}^{\odot2},\\
m_{\bL_i,2}^{(t)} &= \beta_2 m_{\bL_i,2}^{(t-1)} + (1-\beta_2) \frac{1}{s\sqrt{n}}\sum_{k = 1}^s\left\{\widetilde{\nabla}_{\bL_i}f_{in}(\bmu_i^{(t-1)}, \bL_i^{(t-1)}, \bz_s^{(t-1)})\right\}^{\odot2},
\end{align*}
% \[
% $\widehat{\mathbf{x}}_i^{(t+1)} = \widehat{\mathbf{x}}_i^{(t)} + \alpha_t\bar{\mathbf{g}}^{(t)}(\widehat{\mathbf{x}}_i^{(t)})$.
% \]
\quad\quad{}where $\odot2$ denotes the entry-wise square of a vector or a matrix.

% \State \quad\quad If $\|\widehat{\bx}_i^{(t + 1)}\|_2 > 1$, then set $\alpha_t \longleftarrow \alpha_t/2$ and go to line 12

\State \quad\quad Update
\begin{align*}
\bmu_i^{(t)} &= \bmu_i^{(t-1)} - \frac{\alpha_0m_{\bmu_i,1}^{(t)} / (1-\beta_1^t)}{\sqrt{m_{\bmu_i,2}^{(t)} / (1-\beta_2^t)} + \epsilon_0},\quad
\bL_i^{(t)} = \bL_i^{(t-1)} - \frac{\alpha_0m_{\bL_i,1}^{(t)} / (1-\beta_1^t)}{\sqrt{m_{\bL_i,2}^{(t)} / (1-\beta_2^t)} + \epsilon_0},
\end{align*}
\quad\quad{}where the division and square root are computed entry-wise for vectors or matrices.

\State \quad End While

\State \quad Set $\widehat{\bx}_i = \bmu_i^{(t)}$ and $\widehat{\bG}_i = (\bL_i\bL_i\transpose)^{-1}$.

\State End For

\State \textbf{Output: } $\widehat{\bX} = [\widehat{\bx}_1, \ldots, \widehat{\bx}_n]\transpose$ and $(\widehat{\bG}_i)_{i = 1}^n$.

\end{algorithmic}
\end{algorithm}

\subsection{Theoretical properties}
\label{subsec:theory}
We now introduce and establish the asymptotic properties of the variational posterior distribution whose parameters solve the optimization problem \eqref{eqn:vb-objective-population}. Several assumptions are necessary before we state the main results. 

\begin{assumption}
\label{assumption:esl-grdpg}
The following conditions hold:
\begin{enumerate}[(a)]
\item $d$, $p$, and $q$ are constant integers with $d\geq{}1$, $p\geq{}1$, $q\geq{}0$, and $d=p+q$.
\item $\|\bx_{0i}\|_2 \in [\sqrt{\delta}, \sqrt{1-\delta}]$ for all $i\in[n]$, and $\bx_{0i}\transpose\eye_{p,q}\bx_{0j} \in [\delta, 1-\delta]$ for all $i,j\in[n]$, for a constant $\delta\in(0,1/2)$.
\item The eigenvalues of $(1/n)\sum_{i=1}^n\bx_{0i}\bx_{0i}\transpose$, $\sigma_{0,1}\geq\sigma_{0,2}\geq\ldots\geq\sigma_{0,d}$, satisfy either $\sigma_{0,k}=\sigma_{0,k+1}$ or $\sigma_{0,k}-\sigma_{0,k+1}>\lambda$ where $1\leq k\leq d-1$, and $\sigma_{0,d}>\lambda$, for a positive constant $\lambda$ for all $n\geq{}d$.
\item $\rho_n\in(0,1]$ for all $n$, $\lim_{n\to\infty}\rho_n$ exists with $(\log n)/(n\rho_n)\to 0$ as $n\to\infty$.
\item $\delta^2 < \tau_n/\rho_n < \delta/2$ for all $n$.
\item The first $p$ columns of $\bX_0$ are orthogonal to the last $q$ columns of $\bX_0$, where $\bX_0 = [\bx_{01}, \ldots, \bx_{0n}]\transpose$.
\item $\bA\sim\mathrm{GRDPG}(\rho_n^{1/2}\bX_0,\,\eye_{p,q})$.
\end{enumerate}
\end{assumption}
In Assumption \ref{assumption:esl-grdpg} above, items (b) and (c) are standard, and item (d) is a weak requirement on the network sparsity (also see \citealp{xie-2024-spn-bernoulli}). Item (f) can be made without loss of generality by Sylvester's law of inertia. Item (e) guarantees that the true values of $\expect_0A_{ij}$'s stay inside the truncated interval $[\tau_n, 1 - \tau_n]$ and requires that the truncation level $\tau_n$ is not too small. 

\begin{assumption}
\label{assumption:prior}
The prior densities $\pi_i(\bx_i)$, $i\in[n]$, which are independent, satisfy the following conditions, where $C,c>0$ are absolute constants:
\begin{enumerate}[(a)]
\item $0<\pi_i(\bx_i)\leq C$, for all $\bx_i\in\mathbb{R}^d$, and $\pi_i(\rho_n\halfpower\bW\bx_{0i})\geq c$, for all $i\in[n]$;
\item $\|\partial\log\pi_i/\partial\bx_i(\rho_n\halfpower\bW\bx_{0i})\|_2 \leq C$, for all $i\in[n]$;
\item $\log\pi_i(\bx_i)$ is concave in $\bx_i$, and $\|\partial^2\log\pi_i/\partial\bx_i\partial\bx_i\transpose(\bx_i)\|_2 \leq C$, for all $\bx_i\in\mathbb{R}^d$, for all $i\in[n]$.
\end{enumerate}
\end{assumption}
Assumption \ref{assumption:prior} above lists several standard requirements for the prior distribution, such as the prior thickness in a neighborhood of the truth and log-concavity. In particular, $\pi_i(\bx_i)$ can be taken as a multivariate Gaussian distribution with a bounded mean vector and a covariance matrix (in spectra). 

We now establish the asymptotic properties of the maximum extended surrogate likelihood estimator (MESLE). The MESLE provides the theoretical foundations for the Gaussian VI, and it is also theoretically appealing by itself. For convenience, denote by 
$\bG_{in}(\bx_i) = (1/n)\sum_{j = 1}^n{\rho_n\bx_{0j}\bx_{0j}\transpose}/\{\rho_n\halfpower\bx_i\transpose\bx_{0j}(1 - \rho_n\halfpower\bx_i\transpose\bx_{0j})\}$
and let $\bG_{0in} = \bG_{in}(\rho_n\halfpower\bx_{0i})$. Note that $\bG_{0in}$ is precisely the Fisher information matrix for $\bx_i$ at $\rho_n\halfpower\bx_{0i}$. 

\begin{theorem}\label{thm:mesle}
Suppose Assumption \ref{assumption:esl-grdpg} holds.
For each $i\in[n]$, let $\widehat\bx_i=\argmax_{\bx_i\in\mathbb{R}^d}\widehat\ell_{in}(\bx_i)$ be the maximizer of the ESL function.
Then, there exists an orthogonal matrix $\bW\in\mathbb{O}(d)$ depending on $n$, and for any $c>0$, there exist a constant integer $N_{c,\delta,\lambda}$ and a constant $C_{c,\delta,\lambda}$ depending on $c,\,\delta,\,\lambda$, such that for all $n\geq N_{c,\delta,\lambda}$,
\begin{align*}
\prob\left(
\widehat\bx_i{\;exists\;and\;is\;unique\;for\;all\;} i\in[n]
\right)> 1-n^{-c},\\
\prob\left\{\max_{i\in[n]}
\left\Vert
\bW\transpose\widehat\bx_i - \rho_n\halfpower\bx_{0i}
\right\Vert_2
< C_{c,\delta,\lambda}\sqrt{\frac{\log n}{n}}
\right\}
> 1-n^{-c},
\end{align*}
and
% \[
$\sqrt{n}\bG_{0in}\halfpower(\bW\transpose\widehat\bx_i-\rho_n\halfpower\bx_{0i}) \overset{\calL}{\to} \mathrm{N}_d(\zero_d,\,\eye_d)$ as $n\to\infty$.
% \]
If furthermore $(\log n)^4/(n\rho_n)\to 0$ as $n\to\infty$, then 
% \[
$\sum_{i=1}^n\Vert \bW\transpose\widehat\bx_i-\rho_n\halfpower\bx_{0i}\Vert_2^2
- \sum_{i=1}^n\mathrm{tr}(\bG_{0in}\inverse)/n \overset{\prob}{\to} 0$ as $n\to\infty$.
% \]
\end{theorem}

Theorem \ref{thm:bernstein-von-mises} below is a Bernstein-von-Mises theorem for the posterior distribution \eqref{eqn:esl-posterior} associated with the ESL function. It says that the posterior distribution converges in total variation distance to a normal distribution centered at the MESLE with covariance being the inverse Fisher information matrix scaled by $1/{n}$. For technical considerations, we impose the condition $\rho_n=1$, although it is possible to relax it and let $\rho_n\to 0$ as $n\to\infty$.

\begin{theorem}\label{thm:bernstein-von-mises}
Suppose Assumption \ref{assumption:esl-grdpg} and Assumption \ref{assumption:prior} hold, and assume $\rho_n=1$ for all $n$.
For each $i\in[n]$, let $\widehat\bx_i=\argmax_{\bx_i\in\mathbb{R}^d}\widehat\ell_{in}(\bx_i)$ be the maximizer of the ESL function (MESLE), and let $\bW$ be the orthogonal alignment matrix in Theorem \ref{thm:mesle}. Then, with probability at least $1-n^{-c}$,
\begin{align*}
\int_{\mathbb{R}^d}\left|
\pi_{in}(\bx_i\mid\bA)
% \frac{\pi_i(\bx_i)e^{\widehat{\ell}_{in}(\bx_i)}}{\int_{\mathbb{R}^d}\pi_i(\bt_i)e^{\widehat{\ell}_{in}(\bt_i)}\mathrm{d}\bt_i}
 - \phi_d(\bx_i\mid\widehat{\bx}_i, (n\bW\bG_{0in}\bW\transpose)^{-1})
\right|\mathrm{d}\bx_i
\lesssim_{c,\delta,\lambda} \frac{1}{\log n}.
\end{align*}
\end{theorem}
Since the variational posterior distribution is a minimizer of the KL divergence from the family of Gaussian distributions to the true posterior distribution, intuitively, the distance between the variational posterior distribution (\emph{i.e.}, solution to the problem \eqref{eqn:gvb-kl-div}) and the posterior distribution defined in \eqref{eqn:esl-posterior} should be small. With Theorem \ref{thm:bernstein-von-mises}, we can make this intuition precise and establish the following Bernstein-von Mises theorem for VI.
\begin{theorem}\label{thm:vb-posterior}
Suppose the conditions in Theorem \ref{thm:bernstein-von-mises} hold. Let
\[
q_{in}^*(\bx_i)=\argmin_{q\in\calQ_d}D_{\mathrm{KL}}(q(\bx_i)\|\pi_{in}(\bx_i|\bA))
\]
be the variational posterior distribution, where $\calQ_d$ denotes the family of all $d$-dimensional Gaussian distributions. Then
\[
\int_{\mathbb{R}^d} \left|q_{in}^*(\bx_i) - \phi_d(\bx_i\mid\widehat{\bx}_i, (n\bW\bG_{0in}\bW\transpose)^{-1})\right| \diff\bx_i
\lesssim_{c,\delta,\lambda} \sqrt{\frac{1}{\log n}}
\]
with probability at least $1-n^{-c}$.
\end{theorem}
We also provide the asymptotic normality of the variational posterior mean as a point estimator in Theorem \ref{thm:vb-posterior-mean} below.
\begin{theorem}\label{thm:vb-posterior-mean}
Suppose the conditions in Theorem \ref{thm:bernstein-von-mises} hold. Let $\bx_i^*$ be the variational posterior mean of  $q_{in}^*(\bx_i)$. Then, $\sqrt{n}\bG_{0in}\halfpower(\bW\transpose\bx_i^*-\rho_n\halfpower\bx_{0i}) \overset{\calL}{\to} \mathrm{N}_d(\zero_d,\,\eye_d)$ as $n\to\infty$. 
% If furthermore $(\log n)^4/(n\rho_n)\to 0$, then 
% $\sum_{i=1}^n\Vert \bW\transpose\bx_i^*-\rho_n\halfpower\bx_{0i} \Vert_2^2
% - \sum_{i=1}^n\mathrm{tr}(\bG_{0in}\inverse)/n \overset{\prob}{\to} 0$ as $n\to\infty$.
\end{theorem}

\section{Numerical examples}
\label{sec:numeric-ex}
In this section, we study the finite-sample numerical performance of SANVI in several simulated examples of GRDPG. For comparison, we implement the following competing estimates: ASE, OSE developed by \cite{xie-xu-2023-os}, Bayes estimate (BE) as the posterior mean from MCMC with the ESL function, and SANVI. For both BE and SANVI, the improper uniform prior distribution over the Euclidean space is used. We evaluate the performance of an estimator $\widehat{\bX}$ by computing the sum of squared errors (the global error) between the aligned estimated latent positions and the true value counterparts, defined as $\mathrm{SSE}(\widehat{\bX},\,\bX_0) = \inf_{\bW\in\mathbb{O}(d)}\|\widehat{\bX}\bW-\bX_0\|_{\mathrm{F}}^2$. Besides the simulated examples, we also apply the proposed method to a real-world graph dataset. We then discuss briefly the computation time of the proposed algorithm.

We consider four examples of GRDPG to investigate the numerical performance of the proposed estimate in various scenarios: a rank-two stochastic block model, a rank-two degree-corrected stochastic block model, a generic rank-two RDPG, and a generic rank-three GRDPG.
For each example, several sample sizes are considered: $n=$ 1000, 3000, 5000, 7000, 9000, and 10000. We take the truncation parameter in the ESL function to be $\tau=\mathrm{min}(0.001, e^{1.5}/n)$, where $n$ denotes the number of vertices in the graph. For the parameters in the stochastic gradient descent for SANVI, we take the batch size $s=2$, the step size $\alpha_0=0.01$, decay rates for the moments of gradients $\beta_1=0.01$, $\beta_2=0.95$, and the maximum number of iterations to be $1000$. For the MCMC sampling, we use the Metropolis-Hastings algorithm with a Gaussian random walk. The total length of the chain is set to be $3000$, with a thinning of 2 and then a burn-in of 500, giving a set of $1000$ draws to compute the posterior mean. The covariance of the random walk is tuned so that the acceptance rates of the chains for most vertices among the $n$ vertices of a graph lie between $20\%$ and $30\%$. The experiments are repeated for $100$ times for each simulated scenario.

\subsection{A stochastic block model example}
\label{subsec:numeric-sbm}
As the first simulated example, consider a rank-two stochastic block model in the context of random dot product graphs with five blocks with latent positions $\bv_1=[0.3,0.3]$, $\bv_2=[0.5,0.5]$, $\bv_3=[0.7,0.7]$, $\bv_4=[0.3,0.7]$, $\bv_5=[0.7,0.3]$. Each vertex is randomly assigned to one of the five blocks with equal probability. Arrange the five latent positions as the rows of a $5\times2$ matrix $\bB$, and let $\bZ$ be an $n\times{}5$ matrix whose $i$th row $\bz_i\transpose$ encodes the block membership of vertex $i$, i.e., the $k$th entry of $\bz_i$ is $1$ if vertex $i$ belongs to block $k$ and $0$ otherwise. Conditional on the block assignments of all the vertices, we have $A_{ij}\iidsim\mathrm{Bernoulli}(\bv_{k_i}\transpose\bv_{k_j})$ for all $i,j\in[n]$, and $\bA\sim\mathrm{RDPG}(\bZ\bB)$.

The sums of squared errors of the four estimates are summarized in Table \ref{table:sse-sbm}. While OSE is numerically unstable and has sums of squared errors larger than all other estimates due to the latent position $\bv_3=[0.7,\,0.7]$ being close to the unit circle (two vertices in this block have an edge probability of $0.98$), BE and SANVI are nevertheless numerically stable and perform better than ASE. The paired two-sample $t$-tests between the sums of squared errors of ASE and those of BE, and those of SANVI, respectively, are performed, and the $p$-values are listed in Table \ref{table:ttest-sbm}, from which we can see that the two estimates that are based on the ESL function indeed have smaller sums of squared errors.

\begin{table}[t]
  \centering
  \begin{tabular}{l | c c c c}
     \hline
  Estimate & ASE & OSE & BE & SANVI \\
  \hline
    \multirow{2}{*}{$n=1000$} & 8.352   & 19.611   & 7.658   & 7.739 \\
                              & (0.367) & (11.116) & (0.358) & (0.359) \\
        \hline
    \multirow{2}{*}{$n=3000$} & 7.880   & 41.459   & 7.535   & 7.564 \\
                              & (0.210) & (28.124) & (0.199) & (0.203) \\
        \hline
    \multirow{2}{*}{$n=5000$} & 7.834   & 39.936   & 7.588   & 7.593 \\
                              & (0.145) & (29.470) & (0.142) & (0.141) \\
        \hline
    \multirow{2}{*}{$n=7000$} & 7.913   & 39.742   & 7.727   & 7.721 \\
                              & (0.129) & (38.149) & (0.130) & (0.130) \\
        \hline
    \multirow{2}{*}{$n=9000$} & 7.833   & 23.114   & 7.690   & 7.685 \\
                              & (0.115) & (18.278) & (0.116) & (0.116) \\
        \hline
    \multirow{2}{*}{$n=10000$}& 7.774   & 19.486   & 7.650   & 7.643 \\
                              & (0.094) & (17.361) & (0.092) & (0.093) \\
     \hline
  \end{tabular}
  \caption{The sums of squared errors (and their standard errors over 100 repetitions, in parenthesis) of ASE, OSE, BE, and SANVI, respectively, in the example of stochastic block model.}
  \label{table:sse-sbm}
\end{table}

\begin{table}[htbp]
  \centering
  \begin{tabular}{l c c c c}
     \hline
     $n$ & $n=1000$ & $n=3000$ & $n=5000$ & $n=7000$ \\
     \hline
   ASE vs BE  & $2.7\times10^{-82}$ & $7.1\times10^{-87}$ & $5.5\times10^{-81}$ & $3.7\times10^{-84}$  \\
   ASE vs SANVI & $2.2\times10^{-76}$ & $4.9\times10^{-81}$ & $6.9\times10^{-75}$ & $8.1\times10^{-85}$  \\
     \hline
       & $n=9000$ & $n=10000$\\
       \hline
     ASE vs BE  & $7.3\times10^{-84}$ & $3.4\times10^{-78}$\\
     ASE vs SANVI & $4.4\times10^{-76}$ & $1.1\times10^{-76}$\\
     \hline
  \end{tabular}
  \caption{The $p$-values of paired two-sample $t$-tests of the sums of squared errors of ASE with BE and SANVI, respectively, in the example of the stochastic block model.}
  \label{table:ttest-sbm}
\end{table}

\subsection{A degree-corrected stochastic block model example}
\label{subsec:numeric-dcsbm}
In this example, we consider a rank-two degree-corrected stochastic block model in the context of RDPG with two blocks. Specifically, let $\bv_1=[{3\sqrt{10}}/{10}, {\sqrt{10}}/{10}]$ and $\bv_2=[{\sqrt{10}}/{10}, {3\sqrt{10}}/{10}]$. Each vertex is randomly assigned to a block with equal probability and then assigned a degree-corrected parameter (weight) $\theta_i$ that follows $\mathrm{Uniform}(0.05,0.95)$. Arrange the latent positions of the two blocks as the rows of a $2\times2$ matrix $\bB$, let $\bZ$ be an $n\times{}2$ matrix whose $i$th row $\bz_i\transpose$ encodes the block membership of vertex $i$, i.e., the $k$th entry of $\bz_i$ is $1$ if vertex $i$ belongs to block $k$ and $0$ otherwise, and let $\bTheta$ be an $n\times{}n$ diagonal matrix whose $(i,i)$th entry $\theta_i$ is the degree-corrected parameter of vertex $i$. Then, conditional on the block assignments and the degree corrections of all the vertices, we have the edge indicator $A_{ij}\iidsim\mathrm{Bernoulli}(\theta_{i}\theta_{j}\bv_{k_i}\transpose\bv_{k_j})$ for all $i,j\in[n]$, and $\bA\sim\mathrm{RDPG}(\bTheta\bZ\bB)$.

The sums of squared errors of the four estimates are summarized in Table \ref{table:sse-dcsbm}. We can see that for large samples, the three likelihood-based estimates (OSE, BE, and SANVI) all have smaller sums of squared errors than ASE does, and in the case of $n=1000$, the two estimates based on the ESL function still perform well. This phenomenon empirically validates the statement that the likelihood-based estimates improve upon spectral-based ASE, since the latter does not incorporate likelihood information in the graph. The $p$-values given by the paired $t$-tests on the sums of squared errors of ASE and the other three estimates are listed in Table \ref{table:ttest-dcsbm}, which quantitatively verifies the smaller errors given by the likelihood-based estimates.

\begin{table}[t]
  \centering
  \begin{tabular}{l | c c c c}
     \hline
  Estimate & ASE & OSE & BE & SANVI \\
  \hline
    \multirow{2}{*}{$n=1000$} & 3.323   & 3.307   & 3.097   & 3.082 \\
                              & (0.133) & (0.443) & (0.137) & (0.138) \\
        \hline
    \multirow{2}{*}{$n=3000$} & 3.340   & 3.157   & 3.161   & 3.157 \\
                              & (0.074) & (0.068) & (0.069) & (0.069) \\
        \hline
    \multirow{2}{*}{$n=5000$} & 3.317   & 3.140   & 3.148   & 3.146 \\
                              & (0.059) & (0.058) & (0.058) & (0.058) \\
        \hline
    \multirow{2}{*}{$n=7000$} & 3.302   & 3.124   & 3.134   & 3.133 \\
                              & (0.049) & (0.046) & (0.046) & (0.046) \\
        \hline
    \multirow{2}{*}{$n=9000$} & 3.330   & 3.149   & 3.160   & 3.160 \\
                              & (0.040) & (0.038) & (0.038) & (0.038) \\
        \hline
    \multirow{2}{*}{$n=10000$}& 3.291   & 3.117   & 3.128   & 3.128 \\
                              & (0.042) & (0.040) & (0.041) & (0.041) \\
     \hline
  \end{tabular}
  \caption{The sums of squared errors (and their standard errors over 100 repetitions, in parenthesis) of ASE, OSE, BE, and SANVI, respectively, in the example of degree-corrected stochastic block model.}
  \label{table:sse-dcsbm}
\end{table}

\begin{table}[htbp]
  \centering
  \begin{tabular}{l c c c c}
     \hline
     & $n=1000$ & $n=3000$ & $n=5000$ & $n=7000$ \\
     \hline
   ASE vs OSE & $0.7              $ & $3.9\times10^{-87}$ & $1.4\times10^{-101}$& $1.1\times10^{-105}$ \\
   ASE vs BE  & $1.1\times10^{-65}$ & $1.2\times10^{-85}$ & $3.8\times10^{-98}$ & $1.8\times10^{-102}$  \\
   ASE vs SANVI & $9.8\times10^{-66}$ & $2.1\times10^{-86}$ & $9.6\times10^{-99}$ & $2.6\times10^{-102}$  \\
     \hline
       & $n=9000$ & $n=10000$\\
       \hline
     ASE vs OSE & $4.2\times10^{-115}$ & $8.1\times10^{-111}$ \\
     ASE vs BE & $1.3\times10^{-108}$ & $1.7\times10^{-106}$\\
     ASE vs SANVI & $3.2\times10^{-110}$ & $3.5\times10^{-107}$\\
     \hline
     \
  \end{tabular}
  \caption{The $p$-values of paired two-sample $t$-tests of the sums of squared errors of ASE with OSE, BE, and SANVI, respectively, in the example of the degree-corrected stochastic block model.}
  \label{table:ttest-dcsbm}
\end{table}

\subsection{A two-dimensional latent curve example}
\label{subsec:numeric-2drdpg}
Now we consider a rank-two generic RDPG whose latent positions are drawn from a latent curve in $\mathbb{R}^2$, parameterized as $[0.15\sin(\pi{}t)+0.6,\,0.15\cos(\pi{}t)+0.6]\transpose$, for $0<t\leq{}1$, where the $n$ latent positions $\bx_i$ for $i\in[n]$ are then obtained by taking $t={i}/{n}$ for $i\in[n]$. Then, we take $A_{ij}\iidsim{}\mathrm{Bernoulli}(\bx_{0i}\transpose\bx_{0j})$ for all $i,j\in[n]$. Writing $\bX_0=[\bx_1,\,\ldots,\,\bx_n]\transpose$, we then have $\bA\sim\mathrm{RDPG}(\bX_0)$.

The sums of squared errors of the four estimates are summarized in Table \ref{table:sse-2drdpg}. We can see that the errors are obviously larger than those in the previous simulated examples, due to the complex nature of the latent positions obtained from a curve, in comparison to the relatively simple nature of the latent positions in a stochastic block model. Similar to the case in the stochastic block model example, while OSE is numerically unstable in the presence of latent positions close to the unit circle, the two estimates based on the ESL function are nevertheless numerically stable and perform relatively well compared to ASE. From Table \ref{subsec:numeric-2drdpg}, we observe that the sums of squared errors of BE and SANVI are approximately $5\%$ to $15\%$ less than those of ASE. This indicates that the likelihood information in the ESL function facilitates the estimation of latent positions with complex structures. The $p$-values given by the paired two-sample $t$-tests between the sums of squared errors of ASE and those of the two ESL-based estimates are listed in Table \ref{table:ttest-2drdpg}, which quantitatively supports this observation.

\begin{table}[t]
  \centering
  \begin{tabular}{l | c c c c}
     \hline
  Estimate & ASE & OSE & BE & SANVI \\
  \hline
    \multirow{2}{*}{$n=1000$} & 31.800  & 41.128   & 28.713  & 28.655 \\
                              & (0.381) & (11.628) & (0.490) & (0.546)  \\
        \hline
    \multirow{2}{*}{$n=3000$} & 64.216  & 86.080   & 60.203  & 60.119 \\
                              & (0.371) & (19.760) & (0.521) & (0.556)  \\
        \hline
    \multirow{2}{*}{$n=5000$} & 85.716  & 114.552  & 80.924  & 80.711 \\
                              & (16.277)& (31.286) & (16.267)& (16.077) \\
        \hline
    \multirow{2}{*}{$n=7000$} & 28.950  & 76.529   & 25.113  & 25.420 \\
                              & (0.995) & (51.900) & (0.848) & (0.872)  \\
        \hline
    \multirow{2}{*}{$n=9000$} & 25.198  & 73.932   & 21.978  & 22.260 \\
                              & (0.644) & (52.166) & (0.559) & (0.572)  \\
        \hline
    \multirow{2}{*}{$n=10000$}& 24.153  & 71.077   & 21.112  & 21.383 \\
                              & (0.517) & (51.425) & (0.441) & (0.452) \\
     \hline
  \end{tabular}
  \caption{The sums of squared errors (and their standard errors over 100 repetitions, in parenthesis) of ASE, OSE, BE, and SANVI, respectively, in the example of rank-two latent curve.}
  \label{table:sse-2drdpg}
\end{table}

\begin{table}[htbp]
  \centering
  \begin{tabular}{l c c c c}
     \hline
     & $n=1000$ & $n=3000$ & $n=5000$ & $n=7000$ \\
     \hline
   ASE vs BE  & $4.4\times10^{-91}$ & $1.2\times10^{-99}$ & $1.9\times10^{-105}$ & $1.1\times10^{-127}$\\
   ASE vs SANVI & $9.2\times10^{-83}$ & $1.8\times10^{-95}$ & $6.0\times10^{-101}$ & $1.5\times10^{-125}$ \\
     \hline
      & $n=9000$ & $n=10000$\\
      \hline
      ASE vs BE & $1.7\times10^{-136}$ & $2.2\times10^{-139}$ \\
      ASE vs SANVI & $8.0\times10^{-133}$ & $9.6\times10^{-135}$\\
      \hline
  \end{tabular}
  \caption{The $p$-values of paired two-sample $t$-tests of the sums of squared errors of ASE with BE and SANVI, respectively, in the example of the rank-two latent curve.}
  \label{table:ttest-2drdpg}
\end{table}

\subsection{A three-dimensional latent curve example}
\label{subsec:numeric-3drdpg}
In this example, we consider a rank-three generic GRDPG, with signature $(2,1)$, whose latent positions are drawn from a latent curve in $\mathbb{R}^3$, parameterized as $[0.15\sin(2\pi{}t)+0.6,\,0.15\cos(2\pi{}t)+0.6,\,0.15\cos(4\pi{}t)]\transpose$, for $0<t\leq{}1$, where the $n$ latent positions are then obtained by taking $t={i}/{n}$ for $i\in[n]$. In particular, the resulting edge probability matrix $\bP=\bX_0\eye_{2,1}\bX_0\transpose$, where $\bX_0=[\bx_1,\,\ldots,\,\bx_n]\transpose$ and $\eye_{2,1}=\mathrm{diag}(1,\,1,\,-1)$, is an indefinite matrix. We then take $A_{ij}\iidsim{}\mathrm{Bernoulli}(\bx_{0i}\transpose\eye_{2,1}\bx_{0j})$ for all $i,j\in[n]$, and $\bA\sim\mathrm{GRDPG}(\bX_0)$ with signature $(2, 1)$.

The sums of squared errors of the four estimates are summarized in Table \ref{table:sse-3drdpg}. As in the previous example of the rank-two latent curve, the errors are relatively large due to the complex nature of the latent positions obtained from a curve. BE and SANVI still give relatively smaller sums of squared errors compared to ASE, and the $p$-values from the paired two-sample $t$-tests of their sums of squared errors are listed in Table \ref{table:ttest-3drdpg}.

\begin{table}[t]
  \centering
  \begin{tabular}{l | c c c c}
     \hline
  Estimate & ASE & OSE & BE & SANVI \\
  \hline
    \multirow{2}{*}{$n=1000$}& 74.147 & 87.363  & 69.920 & 70.058 \\
                             & (2.581)  & (12.521)  & (2.559)  & (2.589)  \\
        \hline
    \multirow{2}{*}{$n=3000$}& 58.578 & 107.439 & 55.517 & 55.835 \\
                             & (8.160)  & (37.719)  & (7.940)  & (7.954)  \\
        \hline
    \multirow{2}{*}{$n=5000$}& 46.232 & 107.456 & 44.070 & 44.260 \\
                             & (1.039)  & (38.946)  & (0.993)  & (0.996)  \\
        \hline
    \multirow{2}{*}{$n=7000$}& 42.373 & 112.582 & 40.593 & 40.747 \\
                             & (0.590)  & (51.598)  & (0.558)  & (0.565)  \\
        \hline
    \multirow{2}{*}{$n=9000$}& 40.291 & 123.629 & 38.719 & 38.860 \\
                             & (0.448)  & (63.700)  & (0.413)  & (0.424)  \\
        \hline
    \multirow{2}{*}{$n=10000$}& 39.572 & 132.974 & 38.079 & 38.205 \\
                              & (0.369)  & (75.372)  & (0.354)  & (0.358) \\
     \hline
  \end{tabular}
  \caption{The sums of squared errors (and their standard errors over 100 repetitions, in parenthesis) of ASE, OSE, BE, and SANVI, respectively, in the example of rank-three latent curve.}
  \label{table:sse-3drdpg}
\end{table}

\begin{table}[htbp]
  \centering
  \begin{tabular}{l c c c c}
     \hline
     & $n=1000$ & $n=3000$ & $n=5000$ & $n=7000$ \\
     \hline
   ASE vs BE  & $3.0\times10^{-95}$ & $1.6\times10^{-100}$ & $2.1\times10^{-112}$ & $1.7\times10^{-118}$ \\
   ASE vs SANVI & $5.0\times10^{-87}$ & $1.5\times10^{-95}$  & $1.1\times10^{-105}$ & $2.3\times10^{-106}$  \\
     \hline
      & $n=9000$ & $n=10000$\\
      \hline
      ASE vs BE & $1.3\times10^{-115}$ & $9.7\times10^{-124}$ \\
      ASE vs SANVI & $5.3\times10^{-107}$ & $2.4\times10^{-114}$\\
      \hline
  \end{tabular}
  \caption{The $p$-values of paired two-sample $t$-tests of the sums of squared errors of ASE with BE and SANVI, respectively, in the example of rank-three latent curve.}
  \label{table:ttest-3drdpg}
\end{table}

\subsection{Analysis of a real-world graph dataset}
\label{subsec:numeric-real}
We finally apply the proposed algorithm on a real-world network of political blogs \citep{adamic-glance-2005-pb}. The network corresponds to the hyperlinks of blogs regarding U.S. politics after the 2004 presidential election. These blogs are manually classified as either liberal or conservative, which we use as the true value of community labels. After following the rule of thumb by extracting the largest connected component and converting the resulting network with undirected edges, we obtain an $1224\times{}1224$ adjacency matrix with 33430 entries being 1 and others being 0. We apply MCMC and stochastic gradient Algorithm in \ref{alg:vbsl} to compute the BE and SANVI estimates, together with ASE and OSE as the competitors. We choose the embedding dimension to be $d=2$ since there are two true communities in the network. These latent position estimates are then applied to the Gaussian-mixture-model-based clustering to obtain the set of estimated community labels, which we compare against the true community labels via the adjusted Rand index (ARI). The results are listed in Table \ref{table:ari-polblogs}, along with the corresponding computation time. Clearly, the two estimates that are based on the ESL function, i.e., BE and SANVI, are more accurate in terms of recovering the liberal-versus-conservative community structure in the network of these political blogs.

\begin{table}[htbp]
  \centering
  \begin{tabular}{c c c c c}
    \hline
    Estimate                & ASE & OSE & BE & SANVI \\
    \hline
    Adjusted Rand Index     & 0.1321 & 0.0416 & 0.4374 & 0.3117 \\
    Computation Time (seconds)   & 0.05  & 0.05  & 45.20  & 16.64 \\
    \hline
  \end{tabular}
  \caption{The adjusted Rand indices computed from the four estimates and the corresponding computation time for Section \ref{subsec:numeric-real}.}
  \label{table:ari-polblogs}
\end{table}

\subsection{Discussion on computation time}
\label{subsec:compute-time}
A graph that has $n$ vertices gives an $n\times{}n$ adjacency matrix that has $n^2$ entries. Estimating GRDPG involves finding a $d$-dimensional representation for each vertex, resulting in an $n\times{}d$ latent position matrix. We decompose this into $n$ subproblems, each of which finds a $d$-dimensional latent position. Each subproblem is $O(n)$ in time, so the entire problem is $O(n^2)$ in time.

Among the four estimates considered above, ASE and OSE require little time in computation, since the former is just the spectral decomposition truncated at the first $d$ dimensions of the adjacency matrix, and the latter is just a one-step update of the former. Methods based on the likelihood function typically require optimization and/or sampling, which are often computationally intensive. MCMC sampling is useful and often yields good results in various statistical problems; however, it is known to suffer from a long mixing time in some high-dimensional cases. Identifying a stopping criterion for an MCMC sampler is also nontrivial. VI tries to deal with this issue by turning the sampling problem into an optimization problem that requires relatively less computation time.

Specialized to the simulated examples above, while both BE and SANVI perform relatively as well as each other, VI requires less computation time than MCMC sampling. The relationship of computation time and sample size in the example of the stochastic block is given in Figure \ref{fig:time-plot-sbm}, with two quadratic curves fitted for the points corresponding to the MCMC sampler for BE and the stochastic gradient descent algorithm for SANVI, respectively. We can see that although both algorithms are $O(n^2)$ in time, the optimization-based stochastic gradient descent requires around only $20\%$ of that of the sampling-based MCMC algorithm.

\begin{figure}[htbp]
\centering
\includegraphics[width = 0.75\textwidth]{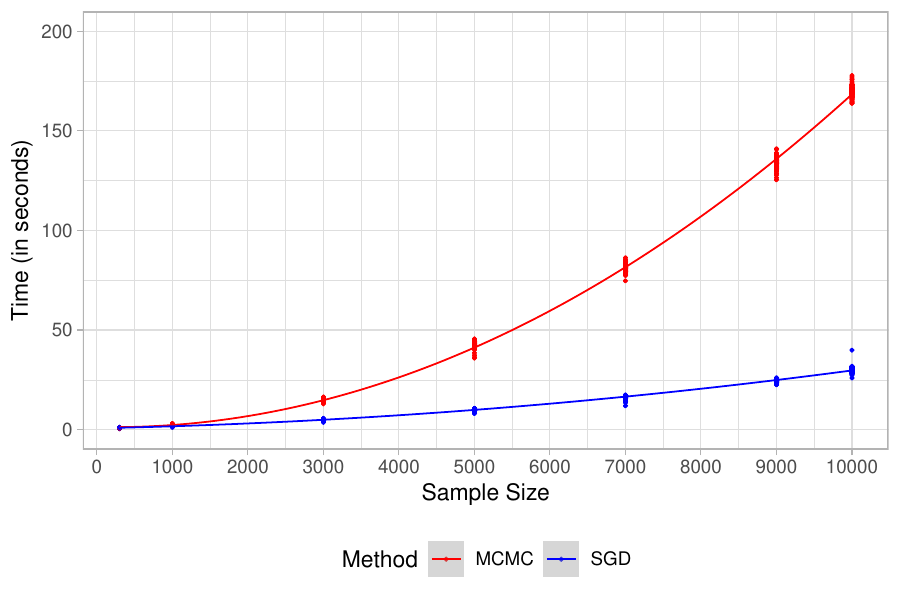}
\caption{The relationship of running time and sample size, in the example of stochastic block model in Section \ref{subsec:numeric-sbm}, where the sample sizes are $n=300$, $1000$, $3000$, $5000$, $7000$, $9000$, and $10000$, with 100 repetitions. Here, two quadratic curves are fitted using the points corresponding to the MCMC for BE and the stochastic gradient descent algorithm (SGD) for SANVI, respectively.}
\label{fig:time-plot-sbm}
\end{figure}

\section{Discussion}
\label{sec:discussion}
In this paper, we propose an ESL function for GRDPG and leverage it to develop a computationally efficient spectral-assisted network variational inference method (SANVI). 
We establish the asymptotic properties of the point estimator, namely, MESLE, including its existence and uniqueness, its consistency at the rate of $\sqrt{n}$ with high probability, its asymptotic normality and efficiency, and, globally for all vertices, the consistency of the global error. The MESLE is theoretically interesting and computationally appealing in its own right. For SANVI, we establish the Bernstein-von Mises theorem of the variational posterior distribution and the asymptotic normality and efficiency of the variation inference estimator. 

We also provide a stochastic gradient descent algorithm for implementing the computation of SANVI. Numerical study shows that, measured in terms of the global error, the point estimate of SANVI (variational posterior mean) is numerically comparable to BE, the latter of which is computed via a classical MCMC sampler. The computation time of the stochastic gradient descent algorithm for SANVI, although still in $O(n^2)$ and is the same as the MCMC sampler due to the intrinsic properties of our setting, is indeed much less than the computation time of the classical MCMC algorithm.

\appendix

\section{Preliminary Results}

This section contains some preliminary results that will be used in the proofs of the main results.

\begin{theorem}
% [Concentration of signatured-adjusted adjacency spectral embedding]
\label{thm:signed-ase}
Suppose Assumption \ref{assumption:esl-grdpg} holds. Let $\widetilde\bX$ denote the signature-adjusted adjacency spectral embedding. Then
\[
\widetilde\bX\bW - \rho_n^{1/2}\bX_0\eye_{p,q} = \rho_n^{-1/2} (\bA-\bP)\bX_0 (\bX_0\transpose\bX_0)\inverse + \bR_{\widetilde\bX},
\]
where, for any $c > 0$, there exists a constant $N_{c,\delta,\lambda}\in\mathbb{N}_+$ depending on $c,\delta,\lambda$, such that for all $n \geq N_{c,\delta,\lambda}$,
\begin{align*}
\|\widetilde\bX\bW - \rho_n^{1/2}\bX_0\eye_{p,q}\|\twotoinfinity &\lesssim_{c,\delta,\lambda} \sqrt{\frac{\log n}{n}}, \quad
\|\bR_{\widetilde\bX}\|\twotoinfinity \lesssim_{c,\delta,\lambda} \frac{\log n}{n\rho_n\halfpower},
\end{align*}
with probability at least $1 - n^{-c}$.
\end{theorem}

\begin{remark}
Theorem \ref{thm:signed-ase} is a generalization of Corollary 4.1 in \cite{xie-2024-spn-bernoulli} to generalized random dot product graphs in our settings. The proof is mostly identical to its original version, with slight modifications such as the presence of the signature matrix $\eye_{p,q}$ and the different definition of the orthogonal alignment matrix $\bW$. We provide a proof here. For more theory on the entrywise limit theorems for eigenvectors of signal-plus-noise matrix models with weak signals, please refer to \cite{xie-2024-spn-bernoulli}.
\end{remark}

\begin{proof}
% [Proof of Theorem \ref{thm:signed-ase}]
We clarify some notations first. Let $\bX_{0+}$ denote the first $p$ columns of $\bX_0$ (those corresponding to the positive part of the signature), and $\bX_{0-}$ the last $q$ columns of $\bX_0$ (those corresponding to the negative part of the signature), that is, $\bX_0=[\bX_{0+},\;\bX_{0-}]$. Define $\Delta_n=(1/n)\bX_0\transpose\bX_0$, then by the assumption that $\bX_{0+}$ is orthogonal to $\bX_{0-}$, we have
\[
\bDelta_n = \begin{bmatrix} \frac{1}{n}\bX_{0+}\transpose\bX_{0+} & \\ & \frac{1}{n}\bX_{0-}\transpose\bX_{0-} \end{bmatrix}.
\]
Perform the eigendecomposition of $\bP=\rho_n\bX_0\eye_{p,q}\bX_0\transpose$, we have $\bP=\bU_\bP \bS_\bP \bU_\bP\transpose$. Group by positive eigenvalues and negative eigenvalues, we have
\[
\bU_\bP = [\bU_{\bP+},\; \bU_{\bP-}],
\quad
\bS_\bP = \begin{bmatrix} \bS_{\bP+} & \\ & \bS_{\bP-} \end{bmatrix}.
\]
Let $\bW_{\bX_{0+}}\in\mathbb{O}(p)$ and $\bW_{\bX_{0-}}\in\mathbb{O}(q)$ be the orthogonal matrices such that $\rho_n\halfpower\bX_{0+}=\bU_{\bP+}\bS_{\bP+}\halfpower\bW_{\bX_{0+}}$ and $\rho_n\halfpower\bX_{0-}=\bU_{\bP-}(-\bS_{\bP-})\halfpower\bW_{\bX_{0-}}$. It is easy to see that $\rho_n\bX_0\eye_{p,q}\bX_0\transpose = \bU_\bP \bS_\bP \bU_\bP\transpose$. 
% \begin{align*}
% \rho_n\bX_0\eye_{p,q}\bX_0\transpose
% &= [\rho_n\halfpower\bX_{0+},\;\rho_n\halfpower\bX_{0-}] \begin{bmatrix} \eye_p & \\ & -\eye_q \end{bmatrix} \begin{bmatrix}\rho_n\halfpower\bX_{0+}\transpose \\ \rho_n\halfpower\bX_{0-}\transpose\end{bmatrix} \\
% &= [\bU_{\bP+}\bS_{\bP+}\halfpower\bW_{\bX_{0+}},\;\bU_{\bP-}(-\bS_{\bP-})\halfpower\bW_{\bX_{0-}}] \begin{bmatrix} \eye_p & \\ & -\eye_q \end{bmatrix}\\ 
% &\quad\times \begin{bmatrix}(\bU_{\bP+}\bS_{\bP+}\halfpower\bW_{\bX_{0+}})\transpose \\ (\bU_{\bP-}(-\bS_{\bP-})\halfpower\bW_{\bX_{0-}})\transpose\end{bmatrix} \\
% &= [\bU_{\bP+},\; \bU_{\bP-}] \begin{bmatrix} \bS_{\bP+} & \\ & \bS_{\bP-} \end{bmatrix} \begin{bmatrix}\bU_{\bP+}\transpose \\ \bU_{\bP-}\transpose\end{bmatrix} 
% % \\&
% = \bU_\bP \bS_\bP \bU_\bP\transpose.
% \end{align*}
Let $\hat\lambda_1\geq\hat\lambda_2\geq\ldots\geq\hat\lambda_d$ be the first $d$ eigenvalues of $\bA$ that are largest in absolute value, and let $\hat\sigma_1\geq\hat\sigma_2\geq\ldots\geq\hat\sigma_d$ be the first $d$ singular values of $\bA$. Note that the numbering for $\hat\lambda_k$ and $\hat\sigma_k$ are different in order. Recall the definition of the adjacency spectral embedding $\breve\bX = \bU_\bA |\bS_\bA|\halfpower$, and the definition of the signature-adjusted adjacency spectral embedding $\widetilde\bX = \bU_\bA |\bS_\bA|\halfpower \mathrm{sgn}(\bS_\bA)$, where $\bS_\bA$ is the diagonal matrix with $\hat\lambda_k$, $k\in[d]$, arranged in the order of real numbers, $\mathrm{sgn}(\bS_\bA)$ is the diagonal matrix with the signs (+1 and -1) of the corresponding eigenvalues, and $\bU_\bA$ is the matrix with the corresponding eigenvectors as columns. By grouping the positive eigenvalues and negative eigenvalues respectively, we can write
\[
\bU_\bA = [\bU_{\bA+},\; \bU_{\bA-}],
\quad
\bS_\bA = \begin{bmatrix} \bS_{\bA+} & \\ & \bS_{\bA-} \end{bmatrix}.
\]
Define $\breve\bX_+ = \bU_{\bA+}\bS_{\bA+}\halfpower$ and $\breve\bX_- = \bU_{\bA-}(-\bS_{\bA-})\halfpower$, and $\widetilde\bX_+ = \bU_{\bA+}\bS_{\bA+}\halfpower$ and $\widetilde\bX_- = -\bU_{\bA-}(-\bS_{\bA-})\halfpower$, we have
\[
\breve\bX = [\breve\bX_+,\; \breve\bX_-]
\quad \mathrm{and} \quad
\widetilde\bX = [\widetilde\bX_+,\; \widetilde\bX_-].
\]
Note that $\widetilde\bX=\breve\bX\eye_{p,q}$.
Let $\bU_{\bP+}\transpose\bU_{\bA+}$ and $\bU_{\bP-}\transpose\bU_{\bA-}$ yield the singular value decompositions $\bW_{1+}\bSigma_+\bW_{2+}\transpose$ and $\bW_{1-}\bSigma_-\bW_{2-}\transpose$, respectively, and define $\bW_+^*=\bW_{1+}\bW_{2+}\transpose$ and $\bW_-^*=\bW_{1-}\bW_{2-}\transpose$, and let $\bW^*=\mathrm{diag}(\bW_+^*,\;\bW_-^*)$. Then the orthogonal alignment matrix between $\breve\bX_+$ and $\rho_n\halfpower\bX_{0+}$ is selected as $(\bW_+^*)\transpose\bW_{\bX_{0+}}$, and the same for that between $\widetilde\bX_+$ and $\rho_n\halfpower\bX_{0+}$; the orthogonal alignment matrix between $\breve\bX_-$ and $\rho_n\halfpower\bX_{0-}$ is selected as $(\bW_-^*)\transpose\bW_{\bX_{0-}}$, and that between $\widetilde\bX_-$ and $\rho_n\halfpower\bX_{0-}$ is selected as $-(\bW_-^*)\transpose\bW_{\bX_{0-}}$. So the orthogonal alignment matrix between the adjacency spectral embedding $\breve\bX$ and $\rho_n\halfpower\bX_0$ is the block diagonal matrix $\mathrm{diag}((\bW_+^*)\transpose\bW_{\bX_{0+}},\;(\bW_-^*)\transpose\bW_{\bX_{0-}})$, and that between $\widetilde\bX$ and $\rho_n\halfpower\bX_0$ is $\mathrm{diag}((\bW_+^*)\transpose\bW_{\bX_{0+}},\;-(\bW_-^*)\transpose\bW_{\bX_{0-}})$. The $\bW$ in the statement of this lemma is $\mathrm{diag}((\bW_+^*)\transpose\bW_{\bX_{0+}},\;(\bW_-^*)\transpose\bW_{\bX_{0-}})$ because we are aligning $\widetilde\bX$ and $\rho_n\halfpower\bX_0\eye_{p,q}$.

\noindent To prove the theorem, we follow the proofs in \cite{xie-2024-spn-bernoulli}. We first present a useful result for random graphs.
\begin{result}
% [Spectral norm concentration for random graphs]
\label{result:spec-norm-concentration}
Suppose Assumption \ref{assumption:esl-grdpg} holds. Let $\bP=\rho_n\bX_0\eye_{p,q}\bX_0\transpose$. Then for any $c>0$, there exists some constant $K_c>$ depending on $c$, such that $\|\bA-\bP\|_2 \leq K_c(n\rho_n)\halfpower$ with probability at least $1-n^{-c}$. This follows exactly from Theorem 5.2 in \cite{lei-rinaldo-2015-sbm}.
\end{result}

\noindent We need to verify the Assumptions 1-5 in \cite{xie-2024-spn-bernoulli} for our setup. By our definition of generalized random dot product graphs, Assumptions 1-3 in \cite{xie-2024-spn-bernoulli} automatically holds. We now verify Assumptions 4 in \cite{xie-2024-spn-bernoulli}. Fix an arbitrary constant $c\geq 1$. Write $\bA=\bP+\bE$. Let $\be_i$ denote the unit basis vector whose $i$th coordinate is one and the rest of coordinates are zeros. Let $\bE^{(m)}$ denote the matrix constructed by replacing the $m$th row and $m$th column of $\bE$ by their expected values which are zeros. Define the function $\phi(x)=(2+\beta_c)(\max(\log(1/x),\,1))\inverse\lambda_d(\bDelta_n)\inverse$ for a constant $\beta_c>0$ that satisfies $\beta_c n\rho_n \geq (c+2)\log{n}$. By Lemma S6.1 in \cite{xie-2024-spn-bernoulli}, for any deterministic $\bV\in\mathbb{R}^{n\times d}$, we have
\begin{align*}
&\prob\left\{ \|\be_i\transpose\bE\bV\|_2 \leq n\rho_n\lambda_d(\bDelta_n)\|\bV\|\twotoinfinity\phi\left(\frac{\|\bV\|\frobenius}{\sqrt{n}\|\bV\|\twotoinfinity}\right) \right\}
%  \\
% &\quad= \prob\left\{ \|\be_i\transpose\bE\bV\|_2 \leq \frac{n\rho_n\lambda_d(\bDelta_n)(2+\beta_c)\|\bV\|\twotoinfinity}{\lambda_d(\bDelta_n)\max(\log(\sqrt{n}\|\bV\|\twotoinfinity/\|\bV\|\frobenius),\,1)} \right\} 
% % \\&
% \geq 1 - 2d\exp(-\beta_c n\rho_n) \\
% &\quad
\geq 1 - c_0 n^{-(1+\xi)}
\end{align*}
where $\xi=c\geq 1$ and $c_0=2$. To show that the same concentration bound also holds for $\|\be_i\transpose\bE^{(m)}\bV\|_2$, we simply observe that $[\bE^{(m)}]_{im}$ can be viewed as a centered Bernoulli random variable whose success probability is zero. Then applying Lemma S6.1 in \cite{xie-2024-spn-bernoulli} leads to
\[
\prob\left\{ \|\be_i\transpose\bE\bV\|_2 \leq n\rho_n\lambda_d(\bDelta_n)\|\bV\|\twotoinfinity\phi\left(\frac{\|\bV\|\frobenius}{\sqrt{n}\|\bV\|\twotoinfinity}\right) \right\}
\geq
1 - c_0 n^{-(1+\xi)}
\]
where $\xi=c\geq 1$ and $c_0=2$. We now verify Assumptions 5 in \cite{xie-2024-spn-bernoulli}. By Result \ref{result:spec-norm-concentration}, there exists a constant $K_c\geq 1$ that depends on $c$ such that $\prob(\|\bE\|_2 \leq K_c(n\rho_n)\halfpower) \geq 1-n^{-c}$. Let $\kappa(\bDelta_n)=\lambda_1(\bDelta_n)/\lambda_d(\bDelta_n)$ be the condition number of $\bDelta_n$. Then with
\[
\gamma = \frac{\max(3K_c,\,\|\bX_0\|\twotoinfinity^2)}{(n\rho_n)\halfpower\lambda_d(\bDelta_n)}
= \frac{3K_c}{(n\rho_n)\halfpower\lambda_d(\bDelta_n)},
\]
we immediately see that
\[
32\kappa(\bDelta_n)\max(\gamma,\,\phi(\gamma))
\lesssim_c
\frac{\kappa(\bDelta_n)}{\lambda_d(\bDelta_n)}\max\left\{ \frac{1}{(n\rho_n)\halfpower},\,\frac{1}{\log(n\rho_n\lambda_d(\bDelta_n)^2)} \right\}
\to 0,
\]
which shows that the Assumption 5 in \cite{xie-2024-spn-bernoulli} holds with $\zeta=c\geq 1$ and $c_0=1$. The five Assumptions in \cite{xie-2024-spn-bernoulli} are thus verified for our setting.

\noindent Write
\begin{align*}
\widetilde\bX\bW - \rho_n^{1/2}\bX_0\eye_{p,q}
% &= \widetilde\bX\bW - \rho_n^{1/2}\bX_0\eye_{p,q}\bX_0\transpose\bX_0 (\bX_0\transpose\bX_0)\inverse \\
% &= \widetilde\bX\bW - \rho_n^{-1/2} (\rho_n\bX_0\eye_{p,q}\bX_0\transpose)\bX_0 (\bX_0\transpose\bX_0)\inverse \\
% &= \widetilde\bX\bW - \rho_n^{-1/2} \bP\bX_0 (\bX_0\transpose\bX_0)\inverse \\
% &= \rho_n^{-1/2} (\bA-\bP)\bX_0 (\bX_0\transpose\bX_0)\inverse + (\widetilde\bX\bW - \rho_n^{-1/2} \bA\bX_0 (\bX_0\transpose\bX_0)\inverse) \\
&= \rho_n^{-1/2} (\bA-\bP)\bX_0 (\bX_0\transpose\bX_0)\inverse + \bR_{\widetilde\bX},
\end{align*}
which can be view as a sum of two terms. Then by Lemma S2.1 in \cite{xie-2024-spn-bernoulli} we have
\begin{align*}
\|\rho_n^{-1/2} (\bA-\bP)\bX_0 (\bX_0\transpose\bX_0)\inverse\|\twotoinfinity
% &= \|\bE\bU_\bP\bS_\bP\invhalfpower\|\twotoinfinity 
% \\&
% \leq 
% \max_{m\in[n]} \|\be_m\transpose\bE\bU_\bP\|_2\|\bS_\bP\invhalfpower\|_2 \\
% &\lesssim_c \frac{(n\rho_n\log{n})\halfpower\|\bU_\bP\|\twotoinfinity}{(n\rho_n)\halfpower\lambda_d(\bDelta_n)\halfpower} 
% % \\&
\leq
\frac{(\log{n})\halfpower}{\lambda_d(\bDelta_n)\halfpower} \|\bU_\bP\|\twotoinfinity
\end{align*}
for all $n\geq N_c$ that depends on $c$ with probability at least $1-n^{-c}$, and by Theorem 3.2 in \cite{xie-2024-spn-bernoulli} we have
\begin{align*}
\|\bR_{\widetilde\bX}\|\twotoinfinity
% &= \|\widetilde\bX\bW - \rho_n^{-1/2} \bA\bX_0 (\bX_0\transpose\bX_0)\inverse\|\twotoinfinity \\
&\lesssim_c 
% \frac{\|\bU_\bP\|\twotoinfinity}{(n\rho_n)^{1/2}\lambda_d(\bDelta_n)^{2}}\max\left\{
% \frac{(\log n)^{1/2}}{\lambda_d(\bDelta_n)^2}, \frac{\kappa(\bDelta_n)}{\lambda_d(\bDelta_n)^2}, \log n
% \right\} 
% \\&
% = 
\frac{\log n\|\bU_\bP\|\twotoinfinity}{(n\rho_n)^{1/2}\lambda_d(\bDelta_n)^{2}}
\end{align*}
for all $n\geq N_c$ that depends on $c$ with probability at least $1-n^{-c}$, and we also have
\[
\|\bU_\bP\|\twotoinfinity
\leq \|\rho\halfpower\bX_0\|\twotoinfinity \||\bS_\bP|\invhalfpower\|_2
\leq \sqrt{\frac{\rho_n}{n\rho_n\lambda_d(\bDelta_n)}}
\leq \frac{1}{\sqrt{n\lambda}}.
\]
So, with the assumption that $(\log n)/(n\rho_n)\to 0$, we have
\begin{align*}
\|\widetilde\bX\bW - \rho_n^{1/2}\bX_0\eye_{p,q}\|\twotoinfinity
&\lesssim_c 
% \frac{(\log{n})\halfpower}{\lambda_d(\bDelta_n)\halfpower} \|\bU_\bP\|\twotoinfinity + \frac{\log n}{(n\rho_n)^{1/2}\lambda_d(\bDelta_n)^{2}}\|\bU_\bP\|\twotoinfinity \\
% &\leq \lambda_d(\bDelta_n)\inverse \sqrt{\frac{\log n}{n}} + \lambda_d(\bDelta_n)^{-5/2} \sqrt{\frac{\log n}{n\rho_n}} \sqrt{\frac{\log n}{n}} 
% \\&
% \lesssim_{c,\lambda}
 \sqrt{\frac{\log n}{n}},\quad
% \end{align*}
% and
% \begin{align*}
\|\bR_{\widetilde\bX}\|\twotoinfinity
% &\lesssim_c \frac{\log n}{(n\rho_n)^{1/2}\lambda_d(\bDelta_n)^{2}}\|\bU_\bP\|\twotoinfinity 
% \\&
\lesssim_{c,\lambda} \frac{\log n}{n\rho_n\halfpower},
\end{align*}
for all $n\geq N_c$ that depends on $c,\lambda$ with probability at least $1-n^{-c}$.
\end{proof}

\begin{lemma}[Some frequently used results]
\label{lemma:frequently-used-results}
Suppose Assumption \ref{assumption:esl-grdpg} holds.
Let $\breve\bX$ denote the adjacency spectral embedding, and $\widetilde\bX$ the signature-adjusted adjacency spectral embedding.
Let $p_{0ij} = \rho_n\bx_{0i}\transpose\eye_{p,q}\bx_{0j}$, and $\widetilde{p}_{ij} = \breve\bx_i\transpose\widetilde{\bx}_j$, $i,j\in [n]$.
Then for any $c > 0$, there exists a constant $N_{c,\delta,\lambda}\in\mathbb{N}_+$ depending on $c,\delta,\lambda$, such that for all $n \geq N_{c,\delta,\lambda}$, the following hold with probability at least $1 - n^{-c}$:
\begin{align*}
(a)\;& \frac{\delta}{2}\rho_n\halfpower \leq \min_{j\in[n]}\|\widetilde\bx_j\|_2 \leq \max_{j\in[n]}\|\widetilde\bx_j\|_2 \leq (1-\frac{\delta}{2})\rho_n^{1/2},\\
(b)\;& \max_{i,j\in[n]}|\widetilde p_{ij} - p_{0ij}| \lesssim_{c,\delta,\lambda} \rho_n\halfpower\sqrt{\frac{\log n}{n}},\\
(c)\;& \frac{\delta}{2}\rho_n \leq \min_{i,j\in[n]}\widetilde p_{ij} \leq \max_{i,j\in[n]}\widetilde p_{ij} \leq (1-\frac{\delta}{2})\rho_n,\\
(d)\;& \max_{j\in[n]}\|\bW\transpose\widetilde\bx_j\widetilde\bx_j\transpose\bW - \rho_n\eye_{p,q}\bx_{0j}\bx_{0j}\transpose\eye_{p,q}\|_2 \lesssim_{c,\delta,\lambda} \rho_n\halfpower\sqrt{\frac{\log n}{n}},\\
(e)\;& \max_{i,j\in[n]}\sup\left\{ |\bx_i\transpose\widetilde\bx_j - p_{0ij}|\middle|\,||\bW\transpose\bx_i-\rho_n\halfpower\bx_{0i}||_2\leq C_{c,\delta,\lambda}\sqrt{\frac{\log n}{n}}\right\} \lesssim_{c,\delta,\lambda} \rho_n\halfpower\sqrt{\frac{\log n}{n}},\\
(f)\;& \min_{i,j\in[n]}\inf\left\{\bx_i\transpose\widetilde\bx_j\middle|\,||\bW\transpose\bx_i-\rho_n\halfpower\bx_{0i}||_2\leq C_{c,\delta,\lambda}\sqrt{\frac{\log n}{n}}\right\} \geq \frac{\delta}{2}\rho_n \\
(g)\;& \max_{i,j\in[n]}\sup\left\{\bx_i\transpose\widetilde\bx_j\middle|\,||\bW\transpose\bx_i-\rho_n\halfpower\bx_{0i}||_2\leq C_{c,\delta,\lambda}\sqrt{\frac{\log n}{n}}\right\} \leq (1-\frac{\delta}{2})\rho_n,\\
\end{align*}
\end{lemma}

\begin{proof}
% [Proof of Lemma \ref{lemma:frequently-used-results}]
We prove the results one by one.
For simplicity of notation, in the proof of this lemma, the results are stated to hold with probability at least $1 - n^{-c}$ for all $n\geq{}{N}_{c,\delta,\lambda}$ for some large constant integer ${N}_{c,\delta,\lambda}$ that depends on $c,\delta,\lambda$, where $c>0$ is an arbitrary positive constant. Also, the results that hold for a single $i\in[n]$ with probability at least $1-n^{-c}$ can be strengthened to hold for all $i\in[n]$ by taking a union bound over $i\in[n]$.

\noindent{}For (a), by assumption we have $(\log n)/(n\rho_n)\to 0$, so we can pick an ${N}_{c,\delta,\lambda}$ large enough such that for all $n\geq{}{N}_{c,\delta,\lambda}$, we have $C_{c,\delta,\lambda}(\log n)/(n\rho_n) < 1-\delta/2-\sqrt{1-\delta}$ (this is because $(1-\delta/2)^2 = 1-\delta+\delta^2/4 > 1-\delta$ and recall that $\delta\in(0,1/2)$), and we also have $C_{c,\delta,\lambda}(\log n)/(n\rho_n) < \sqrt{\delta}-\delta/2$ (this is because
$1-\delta/2-\sqrt{1-\delta}
= 1-\sqrt{1-\delta} -\sqrt{\delta} + \sqrt{\delta} - \delta/2
= 1 - \sqrt{(\sqrt{1-\delta}+\sqrt{\delta})^2} + \sqrt{\delta} - \delta/2
= 1 - \sqrt{1+2\sqrt{\delta(1-\delta)}} + \sqrt{\delta} - \delta/2
\leq \sqrt{\delta} - \delta/2$ and recall that $\delta\in(0,1/2)$).
Then by triangle inequality and Theorem \ref{thm:signed-ase},
\begin{align*}
\min_{j\in[n]} \|\widetilde\bx_j\|_2
&\geq \min_{i,j\in[n]}\|\rho_n\halfpower\bx_{0j}\|_2 - \max_{j\in[n]}\| \bW\transpose\widetilde\bx_j-\rho_n\halfpower\bx_{0j} \|_2
\geq 
% \sqrt{\delta}\rho_n\halfpower - (\sqrt{\delta} - \frac{\delta}{2})\rho_n\halfpower
% = 
\frac{\delta}{2}\rho_n\halfpower,\\
\max_{j\in[n]} \|\widetilde\bx_j\|_2
&\leq \max_{i,j\in[n]}\|\rho_n\halfpower\bx_{0j}\|_2 + \max_{j\in[n]}\| \bW\transpose\widetilde\bx_j-\rho_n\halfpower\bx_{0j} \|_2
\leq 
% \sqrt{1-\delta}\rho_n\halfpower + (1-\frac{\delta}{2}-\sqrt{1-\delta})\rho_n\halfpower
% = 
(1-\frac{\delta}{2})\rho_n\halfpower.
\end{align*}

\noindent For (b), by triangle inequality, Cauchy--Schwarz inequality, and Theorem \ref{thm:signed-ase},
\begin{align*}
\max_{i,j\in[n]}|\widetilde p_{ij} - p_{0ij}|
% &\leq \max_{i,j\in[n]}|\breve\bx_i\transpose\bW(\bW\transpose\widetilde\bx_j-\rho_n\halfpower\eye_{p,q}\bx_{0j})| + \max_{i,j\in[n]}|(\bW\transpose\breve\bx_i-\rho_n\halfpower\bx_{0i})\transpose\rho_n\halfpower\eye_{p,q}\bx_{0j}|\\
&\leq (\max_{j\in[n]}\|\widetilde\bx_j\|_2 + \rho_n\halfpower) \|\widetilde\bX\bW-\rho_n\halfpower\bX_0\|\twotoinfinity
\lesssim_{c,\lambda} \rho_n\halfpower\sqrt{\frac{\log n}{n}}.
\end{align*}

\noindent{}For (c), similar to (a), we can pick an ${N}_{c,\delta,\lambda}$ such that for all $n\geq{}{N}_{c,\delta,\lambda}$, we have\\ $C_{c,\delta,\lambda}(\log n)/(n\rho_n) < \delta/2$. Then by triangle inequality and the previous result,
\begin{align*}
\min_{i,j\in[n]}\widetilde p_{ij}
&\geq \min_{i,j\in[n]}p_{0ij} - \max_{i,j\in[n]}|\widetilde p_{ij} - p_{0ij}|
\geq \delta\rho_n - \frac{\delta}{2}\rho_n
= \frac{\delta}{2}\rho_n,\\
% \]
% and
% \[
\max_{i,j\in[n]}\widetilde p_{ij}
&\leq \max_{i,j\in[n]}p_{0ij} + \max_{i,j\in[n]}|\widetilde p_{ij} - p_{0ij}|
\leq (1-\delta)\rho_n + \frac{\delta}{2}\rho_n
= (1-\frac{\delta}{2})\rho_n.
\end{align*}

\noindent{}For (d), by triangle inequality, Cauchy--Schwarz inequality, and Theorem \ref{thm:signed-ase},
\begin{align*}
\max_{j\in[n]}\|\bW\transpose\widetilde\bx_j\widetilde\bx_j\transpose\bW - \rho_n\eye_{p,q}\bx_{0j}\bx_{0j}\transpose\eye_{p,q}\|_2
% &\leq \max_{j\in[n]}\|\bW\transpose\widetilde\bx_j(\widetilde\bx_j\transpose\bW-\rho_n\halfpower\bx_{0j}\transpose\eye_{p,q})\|_2 \\
% &\quad\quad\quad + \max_{j\in[n]}\|(\bW\transpose\widetilde\bx_j-\rho_n\halfpower\eye_{p,q}\bx_{0j})\rho_n\halfpower\bx_{0j}\transpose\eye_{p,q}\|_2\\
% &\leq (\max_{j\in[n]}\|\widetilde\bx_j\|_2 + \rho_n\halfpower)\|\widetilde\bX\bW-\rho_n\halfpower\bX_0\|\twotoinfinity 
% \\&
\lesssim_{c,\lambda} \rho_n\halfpower\sqrt{\frac{\log n}{n}}.
\end{align*}

\noindent{}For (e), by triangle inequality, Cauchy--Schwarz inequality, and Theorem \ref{thm:signed-ase},
\begin{align*}
\max_{j\in[n]}|\bx_i\transpose\widetilde\bx_j - p_{0ij}|
% &\leq \max_{j\in[n]}|\bx_i\transpose\bW(\bW\transpose\widetilde\bx_j - \rho_n\halfpower\eye_{p,q}\bx_{0j})| + \max_{j\in[n]}|(\bx_i\transpose\bW - \rho_n\halfpower\bx_{0i}\transpose)\rho_n\halfpower\eye_{p,q}\bx_{0j}| \\
% &\leq \max_{j\in[n]}|(\bW\transpose\bx_i-\rho_n\halfpower\bx_{0i})\transpose(\bW\transpose\widetilde\bx_j - \rho_n\halfpower\eye_{p,q}\bx_{0j})| \\
% &\quad\quad\quad\quad + \max_{j\in[n]}|\rho_n\halfpower\bx_{0i}\transpose(\bW\transpose\widetilde\bx_j - \rho_n\halfpower\eye_{p,q}\bx_{0j})| \\
% &\quad\quad\quad\quad + \max_{j\in[n]}|(\bW\transpose\bx_i - \rho_n\halfpower\bx_{0i})\transpose\rho_n\halfpower\eye_{p,q}\bx_{0j}| \\
&\lesssim_{c,\delta,\lambda} 
% \|\bW\transpose\bx_i-\rho_n\halfpower\bx_{0i}\|_2\sqrt{\frac{\log n}{n}} + \rho_n\halfpower\sqrt{\frac{\log n}{n}} + \|\bW\transpose\bx_i-\rho_n\halfpower\bx_{0i}\|_2\rho_n\halfpower \\
% &\asymp_{c,\delta,\lambda} 
\rho_n\halfpower\left(\sqrt{\frac{\log n}{n}} + \|\bW\transpose\bx_i-\rho_n\halfpower\bx_{0i}\|_2\right),
\end{align*}
so
\[
\max_{j\in[n]}\sup_{\bx_i\in B(\rho_n\halfpower\bW\bx_{0i}, C_{c,\delta,\lambda}\sqrt{(\log n)/(n)})}|\bx_i\transpose\widetilde\bx_j - p_{0ij}|
\lesssim_{c,\delta,\lambda} \rho_n\halfpower\sqrt{\frac{\log n}{n}}.
\]
With a union bound, the result holds with maximum over $i\in[n]$.

\noindent{}For (f), similar to (a), we can pick an ${N}_{c,\delta,\lambda}$ such that for all $n\geq{}{N}_{c,\delta,\lambda}$, we have\\
 $C_{c,\delta,\lambda}(\log n)/(n\rho_n) < \delta/2$. Then by triangle inequality and the previous result,
\begin{align*}
\min_{j\in[n]}\inf_{\bx_i\in B(\rho_n\halfpower\bW\bx_{0i}, C_{c,\delta,\lambda}\sqrt{\frac{\log n}{n}})}\bx_i\transpose\widetilde\bx_j
% &\geq \min_{j\in[n]}p_{0ij} - \max_{j\in[n]}\sup_{\bx_i\in B(\rho_n\halfpower\bW\bx_{0i}, C_{c,\delta,\lambda}\sqrt{\frac{\log n}{n}})}|\bx_i\transpose\widetilde\bx_j - p_{0ij}| \\
&\geq \delta\rho_n - \frac{\delta}{2}\rho_n
=\frac{\delta}{2}\rho_n.
\end{align*}
With a union bound, the result holds with minimum over $i\in[n]$.

\noindent{}For (g), similar to (a), we can pick an ${N}_{c,\delta,\lambda}$ such that for all $n\geq{}{N}_{c,\delta,\lambda}$, we have\\
 $C_{c,\delta,\lambda}(\log n)/(n\rho_n) < \delta/2$. Then by triangle inequality and the previous result,
\begin{align*}
\max_{j\in[n]}\sup_{\bx_i\in B(\rho_n\halfpower\bW\bx_{0i}, C_{c,\delta,\lambda}\sqrt{\frac{\log n}{n}})}\bx_i\transpose\widetilde\bx_j
% &\leq \max_{j\in[n]}p_{0ij} + \max_{j\in[n]}\sup_{\bx_i\in B(\rho_n\halfpower\bW\bx_{0i}, C_{c,\delta,\lambda}\sqrt{\frac{\log n}{n}})}|\bx_i\transpose\widetilde\bx_j - p_{0ij}| \\
&\leq (1-\delta)\rho_n + \frac{\delta}{2}\rho_n
=(1-\frac{\delta}{2})\rho_n.
\end{align*}
With a union bound, the result holds with maximum over $i\in[n]$.
\end{proof}

\begin{lemma}[Some results with $A_{ij}$]
\label{lemma:results-with-Aij}
Suppose Assumption \ref{assumption:esl-grdpg} holds.
Denote by $p_{0ij} = \rho_n\bx_{0i}\transpose\eye_{p,q}\bx_{0j}$. Let $\alpha_{ijn}$ be a two-dimensional array of real numbers such that $|\alpha_{ijn}| \leq C_\delta\alpha_n$ for all $n$ where $C_\delta$ is a constant that depends on $\delta$ and $\alpha_n$ is a function of $n$.
Then for any $c > 0$, there exists a constant $N_{c,\delta,\lambda}\in\mathbb{N}_+$ depending on $c,\delta,\lambda$, such that for all $n \geq N_{c,\delta,\lambda}$, the following hold with probability at least $1 - n^{-c}$:
\begin{align*}
(a)\;& \bigg\vert \frac{1}{n}\sum_{j=1}^n(A_{ij}-p_{0ij})\alpha_{ijn}\bigg\vert \lesssim_{c,\delta} \alpha_n\rho_n\sqrt{\frac{\log n}{n\rho_n}}, \quad
(b)\; \Vert\bA\Vert_{\infty} \lesssim_{c,\delta} n\rho_n.\\
(c)\;& \Vert\bA-\rho_n\bX_0\eye_{p,q}\bX_0\transpose\Vert_{\infty} \lesssim_{c,\delta} n\rho_n,\quad
(d)\; \bigg\Vert \frac{1}{n}\sum_{j=1}^n(A_{ij}-p_{0ij})\alpha_{ijn}\bx_{0j}\bigg\Vert_2 \lesssim_{c,\delta,\lambda} \alpha_n\rho_n\sqrt{\frac{\log n}{n\rho_n}}, \\
(e)\;& \bigg\Vert \frac{1}{n}\sum_{j=1}^n(A_{ij}-p_{0ij})\alpha_{ijn}\bx_{0j}\bx_{0j}\transpose\bigg\Vert_2 \lesssim_{c,\delta,\lambda} \alpha_n\rho_n\sqrt{\frac{\log n}{n\rho_n}},  \quad
(f)\; \frac{1}{n}\sum_{j=1}^n A_{ij}\bx_{0j}\bx_{0j}\transpose \succeq \frac{1}{2}\delta\lambda\rho_n\eye_d,  \\
(g)\;& \frac{1}{n}\sum_{j=1}^n A_{ij}\widetilde\bx_j\widetilde\bx_j\transpose \succeq \frac{1}{4}\delta\lambda\rho_n^2\eye_d,
\end{align*}
With a union bound over $i\in[n]$, the results above hold with maximum over $i\in[n]$.
\end{lemma}

\begin{proof}
% [Proof of Lemma \ref{lemma:results-with-Aij}]
For simplicity of notation, in the proof of this lemma, the results are stated to hold with probability at least $1 - n^{-c}$ for all $n\geq{}{N}_{c,\delta,\lambda}$ for some large constant integer ${N}_{c,\delta,\lambda}$ that depends on $c,\delta,\lambda$, where $c>0$ is an arbitrary positive constant.
\noindent{}For (a), by Bernstein's inequality,
\begin{align*}
\prob\bigg( \bigg\vert \frac{1}{n}\sum_{j=1}^n(A_{ij}-p_{0ij})\alpha_{ijn}\bigg\vert \geq t \bigg)
% &\leq 2\exp\left( - \frac{3 n^2 t^2}{6C_\delta^2\alpha_n^2\sum_{j=1}^n p_{0ij}(1-p_{0ij}) + 2C_\delta\alpha_nnt} \right) \\
&\leq 2\exp\bigg( - \frac{3 n^2 t^2}{6C_\delta^2\alpha_n^2n\rho_n + 2C_\delta\alpha_nnt} \bigg).
\end{align*}
Let $c>0$ be given, and by the assumption that $(\log n)/(n\rho_n) \to 0$, we have $\sqrt{\log n}\leq C_2\sqrt{n\rho_n}$ for a constant $C_2$ for all sufficiently large $n$. Take $t=C_1C_\delta\alpha_n\rho_n\halfpower\sqrt{(\log n)/n}$, where $C_1$ is a constant that depends on $c$ and $C_2$ that satisfies $-3C_1^2/(6+2C_1C_2)<-(\log 2)/(\log n)-c$ for all sufficiently large $n$, we have
\begin{align*}
\prob\bigg( \bigg\vert \frac{1}{n}\sum_{j=1}^n(A_{ij}-p_{0ij})\alpha_{ijn}\bigg\vert \geq t \bigg)
% &\leq 2\exp\left( - \frac{3C_1^2C_\delta^2\alpha_n^2 n\rho_n\log n}{6C_\delta^2\alpha_n^2n\rho_n + 2C_1C_\delta^2\alpha_n^2\sqrt{n\rho_n\log n}} \right)  \\
&\leq 2\exp\bigg( - \frac{3C_1^2 n\rho_n\log n}{6n\rho_n + 2C_1C_2n\rho_n} \bigg)
%  \\
% &= 2\exp\left( - \frac{3C_1^2 \log n}{6 + 2C_1C_2} \right)  
% \\&
\leq 2n^{-(\log 2)/(\log n)-c} = n^{-c},
\end{align*}
so 
% \[
$\vert (1/n)\sum_{j=1}^n(A_{ij}-p_{0ij})\alpha_{ijn}\vert \lesssim_{c,\delta} \alpha_n\rho_n\sqrt{{\log n}/(n\rho_n)}$
% \]
for all $n \geq N_{c}$ with probability at least $1 - n^{-c}$. With a union bound, we can take maximum over $i\in[n]$ and the bound still holds.

\noindent For (b), by triangle inequality, the previous result, and the assumption that $(\log n)/(n\rho_n)\to 0$ as $n\to\infty$,
\begin{align*}
||\bA||_{\infty}
&= \max_{i\in[n]}\bigg\vert \sum_{j=1}^n A_{ij}\bigg\vert 
% \\&
\leq \max_{i\in[n]}\bigg\vert \sum_{j=1}^n(A_{ij}-p_{0ij})\bigg\vert + \max_{i\in[n]}\bigg\vert \sum_{j=1}^n p_{0ij}\bigg\vert 
% \\&
\lesssim_{c,\delta} 
% \sqrt{n\rho_n\log n} + n\rho_n
% \asymp_{c,\delta} 
n\rho_n.
\end{align*}

\noindent{}For (c),
\[
\left\Vert\bA-\rho_n\bX_0\eye_{p,q}\bX_0\transpose\right\Vert_{\infty}
\leq \left\Vert\bA\right\Vert_{\infty} + \left\Vert\rho_n\bX_0\eye_{p,q}\bX_0\transpose\right\Vert_{\infty}
\lesssim_{c,\delta} n\rho_n.
\]

\noindent{}For (d), by the assumption that $\|\bx_{0j}\|_2\in[\sqrt{\delta},\sqrt{1-\delta}]$ for all $j\in[n]$, and the previous result (a), we have
\begin{align*}
\bigg\Vert \frac{1}{n}\sum_{j=1}^n(A_{ij}-p_{0ij})\alpha_{ijn}\bx_{0j}\bigg\Vert_2
&\leq \sum_{k=1}^d \bigg\vert \frac{1}{n}\sum_{j=1}^n(A_{ij}-p_{0ij})\alpha_{ijn}x_{0jk}\bigg\vert 
% \\&
\lesssim_{c,\delta} 
% d \alpha_n\rho_n\sqrt{\frac{\log n}{n\rho_n}} 
% \\&
% \asymp_{c,\delta,\lambda} 
\alpha_n\rho_n\sqrt{\frac{\log n}{n\rho_n}},
\end{align*}
in which we note that the chosen embedding dimension $d$ implicitly depends on $\lambda$.

\noindent{}For (e), similar to (c), we have
\begin{align*}
\bigg\Vert \frac{1}{n}\sum_{j=1}^n(A_{ij}-p_{0ij})\alpha_{ijn}\bx_{0j}\bx_{0j}\transpose\bigg\Vert_2
% &\leq \left\Vert \frac{1}{n}\sum_{j=1}^n(A_{ij}-p_{0ij})\alpha_{ijn}\bx_{0j}\bx_{0j}\transpose\right\Vert\frobenius \\
% &\leq \sum_{k_1,k_2=1}^d\sum_{k_2=1}^d \left\vert \frac{1}{n}\sum_{j=1}^n(A_{ij}-p_{0ij})\alpha_{ijn}x_{0jk_1}x_{0jk_2}\right\vert 
% \\&
\lesssim_{c,\delta}
 % d^2 \alpha_n\rho_n\sqrt{\frac{\log n}{n\rho_n}} 
% \\&
% \asymp_{c,\delta,\lambda} 
\alpha_n\rho_n\sqrt{\frac{\log n}{n\rho_n}}.
\end{align*}

\noindent{}For (f), we have
\begin{align*}
\frac{1}{n}\sum_{j=1}^n A_{ij}\bx_{0j}\bx_{0j}\transpose 
% &= \frac{1}{n}\sum_{j=1}^n p_{0ij}\bx_{0j}\bx_{0j}\transpose + \frac{1}{n}\sum_{j=1}^n(A_{ij}-p_{0ij})\bx_{0j}\bx_{0j}\transpose \\
&\succeq \delta\lambda\rho_n\eye_d + \frac{1}{n}\sum_{j=1}^n(A_{ij}-p_{0ij})\bx_{0j}\bx_{0j}\transpose 
% \\&
\succeq \frac{1}{2}\delta\lambda\rho_n\eye_d
\end{align*}
for all $n$ large enough such that $C_{c,\delta,\lambda}\sqrt{(\log n)/(n\rho_n)}<(1/2)\delta\lambda$, by assumption and by the previous result (e).

\noindent{}For (g), with the previous result (f), we have
\begin{align*}
&
\frac{1}{n}\sum_{j=1}^n A_{ij}\widetilde\bx_j\widetilde\bx_j\transpose  
% \\
% &\quad = \frac{1}{n}\sum_{j=1}^n A_{ij}\rho_n\bW\eye_{p,q}\bx_{0j}\bx_{0j}\transpose\eye_{p,q}\bW\transpose + \frac{1}{n}\sum_{j=1}^n A_{ij}\left(\widetilde\bx_j\widetilde\bx_j\transpose - \rho_n\bW\eye_{p,q}\bx_{0j}\bx_{0j}\transpose\eye_{p,q}\bW\transpose\right) \\
% &\quad 
\succeq \frac{1}{2}\delta\lambda\rho_n^2\eye_d + \frac{1}{n}\sum_{j=1}^n A_{ij}(\widetilde\bx_j\widetilde\bx_j\transpose - \rho_n\bW\eye_{p,q}\bx_{0j}\bx_{0j}\transpose\eye_{p,q}\bW\transpose),
\end{align*}
where
\begin{align*}
&
% \quad 
\bigg\|\frac{1}{n}\sum_{j=1}^n A_{ij}\left(\widetilde\bx_j\widetilde\bx_j\transpose - \rho_n\bW\eye_{p,q}\bx_{0j}\bx_{0j}\transpose\eye_{p,q}\bW\transpose\right)\bigg\|_2 
% \\
% &\quad\leq \frac{1}{n}\|\bA\|_\infty \max_{j\in[n]} \left\|\bW\transpose\widetilde\bx_j\widetilde\bx_j\transpose\bW - \rho_n\eye_{p,q}\bx_{0j}\bx_{0j}\transpose\eye_{p,q}\right\|_2 
% \\&
\lesssim_{c,\delta,\lambda} \rho_n^2\sqrt{\frac{\log n}{n\rho_n}}
\end{align*}
by the previous result (b) and Lemma \ref{lemma:frequently-used-results}, so
% \[
$(1/n)\sum_{j=1}^n A_{ij}\widetilde\bx_j\widetilde\bx_j\transpose
\succeq (1/4)\delta\lambda\rho_n^2\eye_d$
% \]
for all $n$ large enough such that $C_{c,\delta,\lambda}\sqrt{(\log n)/(n\rho_n)}<(1/4)\delta\lambda$.
\end{proof}

\begin{lemma}[Concentration of the gradient]
\label{lemma:concentration-gradient}
Suppose Assumption \ref{assumption:esl-grdpg} holds.
Then for any $c > 0$, there exists a constant $N_{c,\delta,\lambda}\in\mathbb{N}_+$ depending on $c,\delta,\lambda$, such that for all $n \geq N_{c,\delta,\lambda}$,
\[
\max_{i\in[n]}
\sup_{\|\bx_i\|_2\leq \rho_n\halfpower}
\left\Vert
\frac{1}{n}\bW\transpose\frac{\partial\widehat\ell_{in}}{\partial\bx_i}(\bW\bx_i) -
\frac{\partial M_{in}}{\partial\bx_i}(\bx_i)
\right\Vert_2
\lesssim_{c,\delta,\lambda}
\sqrt{\frac{\log n}{n}}
\]
with probability at least $1 - n^{-c}$, where
\begin{align*}
M_{in}(\bx_i) = \frac{1}{n}\sum_{j = 1}^n\{\rho_n\bx_{0i}\transpose\eye_{p,q}\bx_{0j}\psi'_n(\rho_n^{1/2}\bx_i\transpose\eye_{p, q}\bx_{0j})
% \\ 
% &\quad\quad\quad\quad
 - (1 - \rho_n\bx_{0i}\transpose\eye_{p, q}\bx_{0j})\psi_n'(1 - \rho_n^{1/2}\bx_i\transpose\eye_{p, q}\bx_{0j})\}
\end{align*}
% \begin{align*}
% \frac{1}{n}\frac{\partial\widehat{\ell}_{in}}{\partial\bx_i}(\bx_i) &= \frac{1}{n}\sum_{j = 1}^n\bigg\{A_{ij}\psi_n'(\bx_i\transpose\widetilde{\bx}_j) - (1 - A_{ij})\psi_n'(1 - \bx_i\transpose\widetilde{\bx}_j)\bigg\}\widetilde{\bx}_j,\\
% \frac{1}{n}\frac{\partial M_{in}}{\partial\bx_i}(\bx_i) & = \frac{1}{n}\sum_{j = 1}^n\bigg\{\rho_n\bx_{0i}\transpose\eye_{p,q}\bx_{0j}\psi'_n(\rho_n^{1/2}\bx_i\transpose\eye_{p, q}\bx_{0j})\\ 
% &\quad\quad\quad\quad - (1 - \rho_n\bx_{0i}\transpose\eye_{p, q}\bx_{0j})\psi_n'(1 - \rho_n^{1/2}\bx_i\transpose\eye_{p, q}\bx_{0j})\bigg\}\rho_n^{1/2}\eye_{p, q}\bx_{0j},
% \end{align*}
and $\widetilde\bX$ denotes the signature-adjusted adjacency spectral embedding.
\end{lemma}

\begin{proof}
% [Proof of Lemma \ref{lemma:concentration-gradient}]
In the proof of this lemma, the large probability bounds with probability at least $1 - n^{-c}$ are stated with respect to all $n\geq{}{N}_{c,\delta,\lambda}$ for some large constant integer ${N}_{c,\delta,\lambda}$ that depends on $c,\delta,\lambda$, where $c>0$ is an arbitrary positive constant. Let $p_{0ij} = \rho_n\bx_{0i}\transpose\eye_{p,q}\bx_{0j}$, and let $g_n(t) = \psi_n'(t)$ for simplicity of notation.
Then it is easy to see that $g_n(t)>0$, that $g'_n(t) = - \tau_n^{-2}\mathbbm{1}(t < \tau_n) - t^{-2}\mathbbm{1}(t\in[\tau_n, 1]) -\mathbbm{1}(t > 1)$, which implies that $g_n(t)$ is decreasing in $t$, that $g_n'(t)$ is constant on $t<\tau_n$ or $t > 1$ and increasing on $t\in[\tau_n, 1]$, and that $-1 \leq \tau_n^2 g'_n(t) \leq -\tau_n^2$.

\noindent{}Note that $\sqrt{1-\delta} < 1-\delta/2$, which is because $1-\delta < 1-\delta+\delta^2/4 = (1-\delta/2)^2$ and recall that $\delta\in(0,1/2)$;
and also note that $\delta/2\leq 1-(1-\delta/2)\rho_n$, which is because $\delta/2=(\delta/2)(1-\rho_n+\rho_n)=(\delta/2)(1-\rho_n)+(\delta/2)\rho_n\leq 1-\rho_n+(\delta/2)\rho_n=1-(1-\delta/2)\rho_n$; so we have $\tau_n<1-(1-\delta/2)\rho_n$, which is because $\tau_n<(\delta/2)\rho_n\leq \delta/2$ by assumption.

\noindent{}We have
% \[
$\max_{j\in[n]}\sup_{\|\bx_i\|_2\leq \rho_n\halfpower}|\rho_n\halfpower\bx_i\transpose\eye_{p,q}\bx_{0j}| \leq\sqrt{1-\delta}\rho_n < (1-\delta/2)\rho_n$
% \]
by assumption,\\
% \[
$\max_{j\in[n]}\sup_{\|\bx_i\|_2\leq \rho_n\halfpower}|\bx_i\transpose\bW\transpose\widetilde\bx_j|\leq \rho_n\halfpower\max_{j\in[n]}\|\widetilde\bx_j\|_2 \leq (1-\delta/2)\rho_n$
% \]
for all $n\geq N_{c,\delta,\lambda}$ that depends on $c,\delta,\lambda$ with probability at least $1-n^{-c}$ by Lemma \ref{lemma:frequently-used-results},\\
% \[
$\max_{j\in[n]}\sup_{\|\bx_i\|_2\leq \rho_n\halfpower}
|\bx_i\transpose\bW\transpose\widetilde\bx_j - \rho_n\halfpower\bx_i\transpose\eye_{p,q}\bx_{0j}|
\leq
% \rho_n\halfpower\max_{j\in[n]}\|\bW\transpose\widetilde\bx_j - \rho_n\halfpower\eye_{p,q}\bx_{0j}\|_2
% \leq
C_{c,\delta,\lambda}\rho_n\halfpower\sqrt{(\log n)/n}$
% \]
for all $n\geq N_{c,\delta,\lambda}$ and a constant $C_{c,\delta,\lambda}$ that depend on $c,\delta,\lambda$ with probability at least $1-n^{-c}$ by Theorem \ref{thm:signed-ase}, and
% \[
$
\max_{j\in[n]}\sup_{\|\bx_i\|_2\leq \rho_n\halfpower}
|\rho_n\halfpower\bx_i\transpose\eye_{p,q}\bx_{0j} + \theta(\bx_i\transpose\bW\transpose\widetilde\bx_j - \rho_n\halfpower\bx_i\transpose\eye_{p,q}\bx_{0j}) |
% \leq
% \sqrt{1-\delta}\rho_n + C_{c,\delta,\lambda}\rho_n\halfpower\sqrt{(\log n)/n}
< (1-\delta/2)\rho_n
$
% \]
for all $n\geq N_{c,\delta,\lambda}$ such that $C_{c,\delta,\lambda}\sqrt{(\log n)/(n\rho_n)} < 1-\delta/2 - \sqrt{1-\delta}$ with probability at least $1-n^{-c}$ by the results above and by the assumption that $(\log n)/(n\rho_n)\to 0$, where $\theta\in(0,1)$.

\noindent{}Write
\begin{equation}\label{eqn:gradient-decomp}
\begin{split}
&\frac{1}{n}\bW\transpose\frac{\partial\widehat\ell_{in}}{\partial\bx_i}(\bW\bx_i) - \frac{\partial M_{in}}{\partial\bx_i}(\bx_i) \\
% &\quad= \frac{1}{n} \sum_{j=1}^n \left\{ A_{ij}g_n(\bx_i\transpose\bW\transpose\widetilde\bx_j) - (1-A_{ij})g_n(1-\bx_i\transpose\bW\transpose\widetilde\bx_j) \right\}\bW\transpose\widetilde\bx_j \\
% &\quad\quad - \frac{1}{n} \sum_{j=1}^n \left\{ p_{0ij}g_n(\rho_n\halfpower\bx_i\transpose\eye_{p,q}\bx_{0j}) - (1-p_{0ij})g_n(1-\rho_n\halfpower\bx_i\transpose\eye_{p,q}\bx_{0j}) \right\}\rho_n\halfpower\eye_{p,q}\bx_{0j} \\
&\quad= \frac{1}{n}\sum_{j=1}^n A_{ij}\left\{ g_n(\bx_i\transpose\bW\transpose\widetilde\bx_j) - g_n(\rho_n\halfpower\bx_i\transpose\eye_{p,q}\bx_{0j}) \right\}\bW\transpose\widetilde\bx_j \\
&\quad\quad- \frac{1}{n}\sum_{j=1}^n (1-A_{ij})\left\{ g_n(1-\bx_i\transpose\bW\transpose\widetilde\bx_j) - g_n(1-\rho_n\halfpower\bx_i\transpose\eye_{p,q}\bx_{0j}) \right\}\bW\transpose\widetilde\bx_j \\
&\quad\quad+ \frac{1}{n}\sum_{j=1}^n A_{ij}g_n(\rho_n\halfpower\bx_i\transpose\eye_{p,q}\bx_{0j})(\bW\transpose\widetilde\bx_j-\rho_n\halfpower\eye_{p,q}\bx_{0j}) \\
&\quad\quad- \frac{1}{n}\sum_{j=1}^n (1-A_{ij})g_n(1-\rho_n\halfpower\bx_i\transpose\eye_{p,q}\bx_{0j}) (\bW\transpose\widetilde\bx_j-\rho_n\halfpower\eye_{p,q}\bx_{0j}) \\
&\quad\quad+ \frac{1}{n}\sum_{j=1}^n (A_{ij}-p_{0ij})g_n(\rho_n\halfpower\bx_i\transpose\eye_{p,q}\bx_{0j})\rho_n\halfpower\eye_{p,q}\bx_{0j} \\
&\quad\quad+ \frac{1}{n}\sum_{j=1}^n (A_{ij}-p_{0ij})g_n(1-\rho_n\halfpower\bx_i\transpose\eye_{p,q}\bx_{0j})\rho_n\halfpower\eye_{p,q}\bx_{0j} ,
\end{split}
\end{equation}
which can be viewed as a sum of six terms. For simplicity of notation, in the remaining of the proof of this lemma, the large probability bounds are stated with respect to all $n\geq{}{N}_{c,\delta,\lambda}$ for some large constant integer ${N}_{c,\delta,\lambda}$ that depends on $c,\delta,\lambda$.

\noindent{}For the first term, with probability at least $1-n^{-c}$,
\begin{align*}
&\sup_{\|\bx_i\|_2\leq \rho_n\halfpower} \left\Vert \frac{1}{n}\sum_{j=1}^n A_{ij}\left\{ g_n(\bx_i\transpose\bW\transpose\widetilde\bx_j) - g_n(\rho_n\halfpower\bx_i\transpose\eye_{p,q}\bx_{0j}) \right\}\bW\transpose\widetilde\bx_j \right\Vert_2 \\
&\quad\leq
% \frac{1}{n} \|\bA\|_\infty \cdot \sup_{\|\bx_i\|_2\leq 1} \max_{j\in[n]}\left|g_n(\bx_i\transpose\bW\transpose\widetilde\bx_j) - g_n(\rho_n\halfpower\bx_i\transpose\eye_{p,q}\bx_{0j})\right| \cdot \max_{j\in[n]} \|\widetilde\bx_j\|_2 \\
% &\quad= 
\frac{1}{n} \|\bA\|_\infty \cdot \sup_{\|\bx_i\|_2\leq 1} \max_{j\in[n]}\left|g_n'(\rho_n\halfpower\bx_i\transpose\eye_{p,q}\bx_{0j} + \theta(\bx_i\transpose\bW\transpose\widetilde\bx_j - \rho_n\halfpower\bx_i\transpose\eye_{p,q}\bx_{0j}))\right| \\
&\quad\quad\quad\quad \cdot \left| \bx_i\transpose\bW\transpose\widetilde\bx_j - \rho_n\halfpower\bx_i\transpose\eye_{p,q}\bx_{0j} \right| \cdot \max_{j\in[n]} \|\widetilde\bx_j\|_2 \\
&\quad\leq \frac{1}{n} \|\bA\|_\infty \cdot \frac{1}{\tau_n^2} \cdot \sup_{\|\bx_i\|_2\leq \rho_n\halfpower}\|\bx_i\|_2 \cdot \max_{j\in[n]} \left\Vert \bW\transpose\widetilde\bx_j - \rho_n\halfpower {\eye_{p,q}\bx_{0j}} \right\Vert_2 \cdot \max_{j\in[n]} \|\widetilde\bx_j\|_2 
% \\&\quad
\lesssim_{c,\delta,\lambda} 
% \frac{1}{n} \cdot n\rho_n \cdot \frac{1}{\rho_n^2} \cdot \rho_n\halfpower \cdot \sqrt{\frac{\log n}{n}} \rho_n\halfpower \\
% &\quad\asymp_{c,\delta,\lambda}
 \sqrt{\frac{\log n}{n}},
\end{align*}
by Cauchy-Schwarz inequality, mean value theorem, the properties of the function $g_n(t)$, the assumption that $\delta^2\rho_n<\tau_n$, Theorem \ref{thm:signed-ase}, Lemma \ref{lemma:frequently-used-results}, and Lemma \ref{lemma:results-with-Aij}.

\noindent{}For the second term, with probability at least $1-n^{-c}$,
\begin{align*}
&\sup_{\|\bx_i\|_2\leq \rho_n\halfpower} \left\Vert \frac{1}{n}\sum_{j=1}^n (1-A_{ij})\left\{ g_n(1-\bx_i\transpose\bW\transpose\widetilde\bx_j) - g_n(1-\rho_n\halfpower\bx_i\transpose\eye_{p,q}\bx_{0j}) \right\}\bW\transpose\widetilde\bx_j \right\Vert_2 \\
&\quad\leq 
% \sup_{\|\bx_i\|_2\leq \rho_n\halfpower}\max_{j\in[n]}\left|g_n(1-\bx_i\transpose\bW\transpose\widetilde\bx_j) - g_n(1-\rho_n\halfpower\bx_i\transpose\eye_{p,q}\bx_{0j})\right| \cdot \max_{j\in[n]} \|\widetilde\bx_j\|_2 \\
% &\quad= 
\sup_{\|\bx_i\|_2\leq \rho_n\halfpower}\max_{j\in[n]} \left|g_n'(1-\rho_n\halfpower\bx_i\transpose\eye_{p,q}\bx_{0j} - \theta(\bx_i\transpose\bW\transpose\widetilde\bx_j - \rho_n\halfpower\bx_i\transpose\eye_{p,q}\bx_{0j}))\right| \\
&\quad\quad\quad\quad \cdot \left| \bx_i\transpose\bW\transpose\widetilde\bx_j - \rho_n\halfpower\bx_i\transpose\eye_{p,q}\bx_{0j} \right| \cdot \max_{j\in[n]} \|\widetilde\bx_j\|_2 \\
&\quad\leq \left|g_n'(1-(1-\delta/2)\rho_n)\right| \cdot \sup_{\|\bx_i\|_2\leq\rho_n\halfpower}\|\bx_i\|_2 \cdot \max_{j\in[n]} \left\Vert \bW\transpose\widetilde\bx_j - \rho_n\halfpower\eye_{p,q}\bx_{0j} \right\Vert_2 \cdot \max_{j\in[n]} \|\widetilde\bx_j\|_2 \\
&\quad\leq \frac{4}{\delta^2} \cdot \sup_{\|\bx_i\|_2\leq\rho_n\halfpower}\|\bx_i\|_2 \cdot \max_{j\in[n]} \left\Vert \bW\transpose\widetilde\bx_j - \rho_n\halfpower\eye_{p,q}\bx_{0j} \right\Vert_2 \cdot \max_{j\in[n]} \|\widetilde\bx_j\|_2 
% \\&\quad
\lesssim_{c,\delta,\lambda} \rho_n\sqrt{\frac{\log n}{n}},
\end{align*}
by Cauchy-Schwarz inequality, mean value theorem, the properties of the function $g_n(t)$, the result that $\tau_n < (\delta/2)\rho_n \leq \delta/2 \leq 1-(1-\delta/2)\rho_n$ shown above, Theorem \ref{thm:signed-ase}, and Lemma \ref{lemma:frequently-used-results}.

\noindent{}For the third term, with probability at least $1-n^{-c}$,
\begin{align*}
&\sup_{\|\bx_i\|_2\leq \rho_n\halfpower} \left\Vert \frac{1}{n}\sum_{j=1}^n A_{ij}g_n(\rho_n\halfpower\bx_i\transpose\eye_{p,q}\bx_{0j})(\bW\transpose\widetilde\bx_j-\rho_n\halfpower\eye_{p,q}\bx_{0j})\right\Vert_2 \\
&\quad\leq \frac{1}{n}\|\bA\|_\infty \sup_{\|\bx_i\|_2\leq \rho_n\halfpower}\max_{j\in[n]}g_n(\rho_n\halfpower\bx_i\transpose\eye_{p,q}\bx_{0j}) \cdot \max_{j\in[n]}\left\Vert \bW\transpose\widetilde\bx_j-\rho_n\halfpower\eye_{p,q}\bx_{0j} \right\Vert_2 \\
&\quad\leq \frac{1}{n}\|\bA\|_\infty \cdot g_n\left(-(1-\delta/2)\rho_n\right) \cdot \max_{j\in[n]}\left\Vert \bW\transpose\widetilde\bx_j-\rho_n\halfpower\eye_{p,q}\bx_{0j} \right\Vert_2
 % \\
% &\quad= \frac{1}{n}\|\bA\|_\infty \cdot \frac{(1-\delta/2)\rho_n+2\tau_n}{\tau_n^2} \cdot \max_{j\in[n]}\left\Vert \bW\transpose\widetilde\bx_j-\rho_n\halfpower\eye_{p,q}\bx_{0j} \right\Vert_2 \\
% &\quad\leq \frac{1}{n}\|\bA\|_\infty \cdot \left(\frac{1-\delta/2}{\delta^4}\rho_n\inverse + \frac{2}{\delta^2}\rho_n^{-1}\right) \cdot \max_{j\in[n]}\left\Vert \bW\transpose\widetilde\bx_j-\rho_n\halfpower\eye_{p,q}\bx_{0j} \right\Vert_2 \\
% &\quad
\lesssim_{c,\delta,\lambda}
 % \frac{1}{n} \cdot n\rho_n \cdot \rho_n\inverse \cdot \sqrt{\frac{\log n}{n}} 
% \\&\quad
% \asymp_{c,\delta,\lambda} 
\sqrt{\frac{\log n}{n}},
\end{align*}
where the inequalities and equalities follow from Cauchy-Schwarz inequality, triangle inequality, the properties of the function $g_n(t)$, the assumption that $\delta^2 < \tau_n/\rho_n < \delta/2$, Theorem \ref{thm:signed-ase}, Lemma \ref{lemma:frequently-used-results}, and Lemma \ref{lemma:results-with-Aij}.

\noindent{}For the fourth term, with probability at least $1-n^{-c}$,
\begin{align*}
&\sup_{\|\bx_i\|_2\leq \rho_n\halfpower} \left\Vert \frac{1}{n}\sum_{j=1}^n (1-A_{ij})g_n(1-\rho_n\halfpower\bx_i\transpose\eye_{p,q}\bx_{0j})(\bW\transpose\widetilde\bx_j-\rho_n\halfpower\eye_{p,q}\bx_{0j})\right\Vert_2 \\
&\quad\leq \sup_{\|\bx_i\|_2\leq \rho_n\halfpower}\max_{j\in[n]}g_n(1-\rho_n\halfpower\bx_i\transpose\eye_{p,q}\bx_{0j}) \cdot \max_{j\in[n]}\left\Vert \bW\transpose\widetilde\bx_j-\rho_n\halfpower\eye_{p,q}\bx_{0j} \right\Vert_2 \\
&\quad\leq 
% g_n(1-(1-\delta/2)\rho_n) \cdot \max_{j\in[n]}\left\Vert \bW\transpose\widetilde\bx_j-\rho_n\halfpower\eye_{p,q}\bx_{0j} \right\Vert_2 \\
% &\quad= 
\frac{1}{1-(1-\delta/2)\rho_n} \cdot \max_{j\in[n]}\left\Vert \bW\transpose\widetilde\bx_j-\rho_n\halfpower\eye_{p,q}\bx_{0j} \right\Vert_2
 % \\
% &\quad\leq (2/\delta) \cdot \max_{j\in[n]}\left\Vert \bW\transpose\widetilde\bx_j-\rho_n\halfpower\eye_{p,q}\bx_{0j} \right\Vert_2 
% \\&
\lesssim_{c,\delta,\lambda} \sqrt{\frac{\log n}{n}} ,
\end{align*}
where the inequalities and equalities follow from Cauchy-Schwarz inequality, triangle inequality, the properties of the function $g_n(t)$, the assumption that $\delta^2 < \tau_n/\rho_n < \delta/2$, and Theorem \ref{thm:signed-ase}.

\noindent{}For the fifth term and the sixth term, some methods in empirical processes are needed. We first define a stochastic process indexed by $\bx_i$ in the closed ball that is centered at origin and of radius $\rho_n\halfpower$, then use the results in Chapter 2.2 of \cite{vaart-wellner-2023-weak} to compute the bounds on the Orlicz $\psi_1$ norm for the supremum of the process, and then use a Bernstein-type inequality (Theorem 12.2 in \cite{boucheron-lugosi-massart-2013}) to obtain a tail probability bound for the supremum of the process.
We now show the fifth term. Let $J_{ijk}(\bx_i) = (A_{ij}-p_{0ij})g_n(\rho_n\halfpower\bx_i\transpose\eye_{p,q}\bx_{0j})\rho_n\halfpower x_{0jk}$, and for each $k\in[d]$, define a stochastic processes $J_{ink}(\bx_i) = (1/n)\sum_{j=1}^n J_{ijk}(\bx_i)$, where $\bx_i\in B(\zero_d,\rho_n\halfpower)=\{\bx_i\in\mathbb{R}^d:\|\bx_i\|_2\leq \rho_n\halfpower\}$. Then we have
% \[
$\expect\left[ J_{ijk}(\bx_i) - J_{ijk}(\bx_i') \right] = 0$,
% \]
\begin{align*}
& |J_{ijk}(\bx_i) - J_{ijk}(\bx_i')| \\
% &\quad= \left|(A_{ij}-p_{0ij})g_n'\left(\theta\rho_n\halfpower\bx_{0j}\transpose\eye_{p,q}\bx_i + (1-\theta)\rho_n\halfpower\bx_{0j}\transpose\eye_{p,q}\bx_i'\right)\rho_n\halfpower\bx_{0j}\transpose\eye_{p,q}(\bx_i-\bx_i') \rho_n\halfpower x_{0jk}\right|  \\
&\quad\leq \max_{j\in[n]} \left|g_n'\left(\theta\rho_n\halfpower\bx_{0j}\transpose\eye_{p,q}\bx_i + (1-\theta)\rho_n\halfpower\bx_{0j}\transpose\eye_{p,q}\bx_i'\right)\right| \cdot \max_{j\in[n]} \rho_n\halfpower \|\bx_{0j}\|_2 \rho_n\halfpower |x_{0jk}| \|\bx_i-\bx_i'\|_2 \\
&\quad\leq \tau_n^{-2}\rho_n\|\bx_i-\bx_i'\|_2
 % \\&\quad
\leq \frac{1}{\delta^4}\rho_n\inverse \|\bx_i-\bx_i'\|_2,
\end{align*}
and similarly, $\expect |J_{ijk}(\bx_i) - J_{ijk}(\bx_i')|^2\leq \delta^{-8}\rho_n^{-1}\|\bx_i - \bx_i'\|_2^2$
% \begin{align*}
% &\expect |J_{ijk}(\bx_i) - J_{ijk}(\bx_i')|^2
 % \\
% &\quad= \expect\left|(A_{ij}-p_{0ij})g_n'\left(\theta\rho_n\halfpower\bx_{0j}\transpose\eye_{p,q}\bx_i + (1-\theta)\rho_n\halfpower\bx_{0j}\transpose\eye_{p,q}\bx_i'\right)\rho_n\halfpower\bx_{0j}\transpose\eye_{p,q}(\bx_i-\bx_i') \rho_n\halfpower x_{0jk}\right|^2  \\
% &\quad\leq \max_{j\in[n]}p_{0ij}(1-p_{0ij}) \cdot \max_{j\in[n]} \left|g_n'\left(\theta\rho_n\halfpower\bx_{0j}\transpose\eye_{p,q}\bx_i + (1-\theta)\rho_n\halfpower\bx_{0j}\transpose\eye_{p,q}\bx_i'\right)\right|^2
% \\ &\quad\quad \times 
% \max_{j\in[n]} \rho_n\|\bx_{0j}\|_2^2 \rho_n |x_{0jk}|^2 \|\bx_i-\bx_i'\|_2^2  \\
% &\quad\leq \rho_n \tau_n^{-4} \rho_n^2 \|\bx_i-\bx_i'\|_2^2
% \\&\quad
% \leq \frac{1}{\delta^8}\rho_n\inverse \|\bx_i-\bx_i'\|_2^2,
% \end{align*}
for all $\bx_i\,\bx_i'\in B(\zero_d,\rho_n\halfpower)$ and all $j\in[n]$. Then for any $\bx_i,\,\bx_i'\in B(\zero_d,\rho_n\halfpower)$, by Bernstein's inequality,
\begin{align*}
&\prob\left\{ |J_{ink}(\bx_i) - J_{ink}(\bx_i')| \geq t \right\} \\
% &\quad= \prob\bigg\{ \bigg|\frac{1}{n}\sum_{j=1}^n (A_{ij}-p_{0ij})g_n'\left(\theta\rho_n\halfpower\bx_{0j}\transpose\eye_{p,q}\bx_i + (1-\theta)\rho_n\halfpower\bx_{0j}\transpose\eye_{p,q}\bx_i'\right)\\ 
% &\qquad\qquad\times \rho_n\halfpower\bx_{0j}\transpose\eye_{p,q}(\bx_i-\bx_i') \rho_n\halfpower x_{0jk}\bigg| \geq t \bigg\} \\
% &\quad\leq 2\exp\left\{ \frac{-3n^2t^2}{(6/\delta^8)n\rho_n\inverse\|\bx_i-\bx_i'\|_2^2 + (2/\delta^4)n\rho_n\inverse\|\bx_i-\bx_i'\|_2 t} \right\}  \\
% &\quad= 2\exp\left\{ \frac{-t^2}{2C_\delta^2(n\rho_n)\inverse\|\bx_i-\bx_i'\|_2^2 + (2/3)C_\delta(n\rho_n)\inverse\|\bx_i-\bx_i'\|_2 t} \right\}  \\
&\quad\leq 2\exp\left\{ -\min\left( \frac{t^2}{4C_\delta^2(n\rho_n)\inverse\|\bx_i-\bx_i'\|_2^2},\, \frac{t}{(4/3)C_\delta(n\rho_n)\inverse\|\bx_i-\bx_i'\|_2}\right) \right\},
\end{align*}
{where we take $C_\delta = \delta^{-4}$.} We then consider the case where $|J_{ink}(\bx_i) - J_{ink}(\bx_i')|\leq 3C_\delta\|\bx_i-\bx_i'\|_2$ (the sub-Gaussian part) and the case where $|J_{ink}(\bx_i) - J_{ink}(\bx_i')| > 3C_\delta\|\bx_i-\bx_i'\|_2$ (the sub-exponential part) separately, noting that both the two bounds hold for all $t>0$, i.e.,
\begin{align*}
&\prob\Big\{ |J_{ink}(\bx_i) - J_{ink}(\bx_i')| \cdot \one\left( |J_{ink}(\bx_i) - J_{ink}(\bx_i')|\leq 3C_\delta\|\bx_i-\bx_i'\|_2 \right) \geq t \Big\} \\
&\quad\leq 2\exp\left\{ \frac{-t^2}{4C_\delta^2(n\rho_n)\inverse\|\bx_i-\bx_i'\|_2^2} \right\}, \\
&\prob\Big\{ |J_{ink}(\bx_i) - J_{ink}(\bx_i')| \cdot \one\left( |J_{ink}(\bx_i) - J_{ink}(\bx_i')| > 3C_\delta\|\bx_i-\bx_i'\|_2 \right) \geq t \Big\} \\
&\quad\leq 2\exp\left\{ \frac{-t}{(4/3)C_\delta(n\rho_n)\inverse\|\bx_i-\bx_i'\|_2} \right\}.
\end{align*}
Recall that the Orlicz $\psi_p$ norm of a random variable $X$ is $\|X\|_{\psi_p}=\inf\{c>0:\expect[\psi_p(|x|/c)]\leq 1\}$ with $\psi_p(x)=e^{x^p}-1$. Then, by sub-additivity, the fact that $\|X\|_{\psi_1}\leq(1/\sqrt{\log 2})\|X\|_{\psi_2}$ (Problem 2.2.5 in \cite{vaart-wellner-2023-weak}), and Lemma 2.2.1 in \cite{vaart-wellner-2023-weak}, we can bound the Orlicz $\psi_1$ norm of $J_{ink}(\bx_i) - J_{ink}(\bx_i')$, i.e.,
\begin{align*}
&\left\| J_{ink}(\bx_i) - J_{ink}(\bx_i') \right\|_{\psi_1}  \\
&\quad\leq \left\| |J_{ink}(\bx_i) - J_{ink}(\bx_i')| \cdot \one\left( |J_{ink}(\bx_i) - J_{ink}(\bx_i')|\leq 3C_\delta\|\bx_i-\bx_i'\|_2 \right) \right\|_{\psi_1} \\
&\quad\quad + \left\| |J_{ink}(\bx_i) - J_{ink}(\bx_i')| \cdot \one\left( |J_{ink}(\bx_i) - J_{ink}(\bx_i')| > 3C_\delta\|\bx_i-\bx_i'\|_2 \right) \right\|_{\psi_1} \\
&\quad\leq \frac{1}{\sqrt{\log 2}}\left\| |J_{ink}(\bx_i) - J_{ink}(\bx_i')| \cdot \one\left( |J_{ink}(\bx_i) - J_{ink}(\bx_i')|\leq 3C_\delta\|\bx_i-\bx_i'\|_2 \right) \right\|_{\psi_2} \\
&\quad\quad + \left\| |J_{ink}(\bx_i) - J_{ink}(\bx_i')| \cdot \one\left( |J_{ink}(\bx_i) - J_{ink}(\bx_i')| > 3C_\delta\|\bx_i-\bx_i'\|_2 \right) \right\|_{\psi_1} \\
&\quad\leq \sqrt{(12/\log 2)}C_\delta(n\rho_n)\invhalfpower\|\bx_i-\bx_i'\|_2 + 4C_\delta(n\rho_n)\inverse\|\bx_i-\bx_i'\|_2
 % \\&\quad
 = C_\delta K(n) \|\bx_i-\bx_i'\|_2,
\end{align*}
where $C_\delta = 1/\delta^4$ (note that $C_\delta>1$) and $K(n)=\sqrt{(12/\log 2)}(n\rho_n)\invhalfpower + 4(n\rho_n)\inverse$.

\noindent{}Define a metric $d_J(\bx_i, \bx_i')=C_\delta K(n) \|\bx_i-\bx_i'\|_2$ on $B(\zero_d,\rho_n\halfpower)$, then the diameter of $B(\zero_d,\rho_n\halfpower)$ under $d_J$ is $2C_\delta \rho_n\halfpower{}K(n)$, and the packing number of the metric space $(B(\zero_d,\rho_n\halfpower), d_J)$ satisfies $D(\epsilon, d_J)\leq (2C_\delta \rho_n\halfpower{}K(n)/\epsilon)^d$. Then by Corollary 2.2.5 of \cite{vaart-wellner-2023-weak},
\begin{align*}
&\bigg\| \sup_{\bx_i,\,\bx_i'\in B(\zero_d,\rho_n\halfpower)} \bigg|J_{ink}(\bx_i) - J_{ink}(\bx_i')\bigg| \bigg\|_{\psi_1} \\
&\quad\leq C\int_0^{\mathrm{diam}B(\zero_d,\rho_n\halfpower)} \log(1 + D(\epsilon,d_j))\diff\epsilon 
% \\&
\leq 
% C\int_0^{2C_\delta \rho_n\halfpower{}K(n)} \log\left(1 + \left(\frac{2C_\delta \rho_n\halfpower{}K(n)}{\epsilon}\right)^d \right)\diff\epsilon \\
% &\quad\leq C\int_0^{2C_\delta \rho_n\halfpower{}K(n)} \log\left(2\left(\frac{2C_\delta \rho_n\halfpower{}K(n)}{\epsilon}\right)^d \right)\diff\epsilon 
% \\&
% = 
2CC_\delta \rho_n\halfpower{}K(n)\int_0^{1} \log\left(2\left(\frac{1}{\epsilon'}\right)^d \right)\diff\epsilon' \\
&\quad= 2CC_\delta \rho_n\halfpower{}K(n) (\log 2 + d) 
% \\&
< 2(1+d) C C_\delta \rho_n\halfpower{}K(n),
\end{align*}
where $C$ is a constant (related to the function $\psi_1(x)=e^x-1$).

\noindent{}Consider $J_{ijk}(\zero_d)=(A_{ij}-p_{0ij}) 2\tau_n\inverse \rho_n\halfpower x_{0jk}$. We then have
$|J_{ijk}(\zero_d)| \leq 2C_\delta\halfpower\rho_n\invhalfpower$ and $\expect|J_{ijk}(\zero_d)|^2 \leq 4C_\delta$
for all $j\in[n]$. Then by Bernstein's inequality,
\begin{align*}
\prob\left\{ |J_{ink}(\zero_d)| \geq t \right\}
% &= 2\exp\left\{ \frac{-t^2}{8C_\delta n\inverse + (4/3)C_\delta\halfpower n\inverse\rho_n\invhalfpower t} \right\} \\
&\leq 2\exp\left\{ -\min\left( \frac{t^2}{16C_\delta n\inverse},\, \frac{t}{(8/3)C_\delta\halfpower n\inverse\rho_n\invhalfpower}\right) \right\},
\end{align*}
from which, similar to the computation for  $\|J_{ink}(\bx_i) - J_{ink}(\bx_i')\|_{\psi_1}$ shown above, we can obtain
 % (omitting the details)
% \begin{align*}
$\| J_{ink}(\zero_d) \|_{\psi_1} \leq \sqrt{(48/\log 2)} C_\delta\halfpower n\invhalfpower + 8 C_\delta\halfpower n\inverse\rho_n\invhalfpower 
% \\&
= 2C_\delta\halfpower \rho_n\halfpower K(n)$.
% \end{align*}
By triangle inequality, monotonicity of integral, and sub-additivity of norm,
\begin{align*}
\bigg\| \sup_{\bx_i\in B(\zero_d,\rho_n\halfpower)} |J_{ink}(\bx_i)| \bigg\|_{\psi_1}
&\leq \bigg\| \sup_{\bx_i,\,\bx_i'\in B(\zero_d,1)} \bigg|J_{ink}(\bx_i) - J_{ink}(\bx_i')\bigg| \bigg\|_{\psi_1} + \| J_{ink}(\zero_d) \|_{\psi_1}
 \\
% &\leq 2(1+d) C C_\delta \rho_n\halfpower{}K(n) + 2C_\delta\halfpower \rho_n\halfpower K(n) \\
&\leq ((1+d) C + 1) 2C_\delta \rho_n\halfpower K(n),
\end{align*}
and we also have
\[
\expect\bigg[ \sup_{\bx_i\in B(\zero_d,\rho_n\halfpower)} |J_{ink}(\bx_i)| \bigg] = \bigg\|\sup_{\bx_i\in B(\zero_d,\rho_n\halfpower)} |J_{ink}(\bx_i)|\bigg\|_1 \leq ((1+d) C + 1) 2C_\delta \rho_n\halfpower K(n),
\]
since $\|X\|_1\leq\|X\|_{\psi_1}$ (recall that $x\leq e^x-1$ for $x\geq 0$).

\noindent{}Similar to the previous steps, it is straightforward to compute that $|(2C_\delta)\inverse\rho_n\halfpower J_{ijk}(\bx_i)|<1$ for all $\bx_i\in B(\zero_d,\rho_n\halfpower)$ and for all $j\in[n]$, and also that
\begin{align*}
\expect\Bigg[ \sup_{\bx_i\in B(\zero_d,\rho_n\halfpower)}\sum_{j=1}^n((2C_\delta)\inverse\rho_n\halfpower J_{ijk}(\bx_i))^2 \Bigg] &< n\rho_n, \\
\sup_{\bx_i\in B(\zero_d,\rho_n\halfpower)}\sum_{j=1}^n\expect\Bigg[ ((2C_\delta)\inverse\rho_n\halfpower J_{ijk}(\bx_i))^2 \Bigg] &< n\rho_n.
\end{align*}
\noindent{}Let $B'(\zero_d,\rho_n\halfpower)$ be the set of all rational points in the closed ball that is centered at origin and of radius $\rho_n\halfpower$, then $B'(\zero_d,\rho_n\halfpower)$ is countable and is a dense subset of $B(\zero_d,\rho_n\halfpower)$. Since $J_{ink}(\bx_i)$ is continuous in $\bx_i$ surely, for any $\bx_i^*\in B(\zero_d,\rho_n\halfpower)$, there exists a sequence $\{\bx_i^{(m)}\}\subset B'(\zero_d,\rho_n\halfpower)$ with $\lim_{m\to\infty}\bx_i^{(m)} = \bx_i^*$ such that $\lim_{m\to\infty}J_{ink}(\bx_i^{(m)}) = J_{ink}(\bx_i^*)$. So $J_{ink}(\bx_i)$ indexed by $\bx_i\in B(\zero_d,\rho_n\halfpower)$ is a separable process, and we have
\begin{align*}
\sup_{\bx_i\in B(\zero_d,\rho_n\halfpower)}J_{ink}(\bx_i) = \sup_{\bx_i\in B'(\zero_d,\rho_n\halfpower)}J_{ink}(\bx_i), \\
\inf_{\bx_i\in B(\zero_d,\rho_n\halfpower)}J_{ink}(\bx_i) = \inf_{\bx_i\in B'(\zero_d,\rho_n\halfpower)}J_{ink}(\bx_i).
\end{align*}

\noindent{}Then, by Theorem 12.2 in \cite{boucheron-lugosi-massart-2013}, we have
\begin{align*}
&\prob\bigg\{ \sup_{\bx_i\in B(\zero_d,\rho_n\halfpower)}J_{ink}(\bx_i) \geq \expect\bigg[ \sup_{\bx_i\in B(\zero_d,\rho_n\halfpower)} |J_{ink}(\bx_i)| \bigg] + t \bigg\}
 % \\
% &\quad\leq \prob\bigg\{ \sup_{\bx_i\in B(\zero_d,\rho_n\halfpower)}J_{ink}(\bx_i) \geq \expect\bigg[ \sup_{\bx_i\in B(\zero_d,\rho_n\halfpower)} J_{ink}(\bx_i) \bigg] + t \bigg\} 
% \\&\quad
\leq 
% \exp\left\{ \frac{-(n\rho_n\halfpower/(2C_\delta))^2t^2}{2(4n\rho_n + (n\rho_n\halfpower/(2C_\delta))t)} \right\} 
% \\&
% = 
\exp\bigg\{ \frac{-n t^2}{32C_\delta^2 + 4C_\delta\rho_n\invhalfpower t} \bigg\},
\end{align*}
and also
\begin{align*}
& \prob\bigg\{ -\sup_{\bx_i\in B(\zero_d,\rho_n\halfpower)}J_{ink}(\bx_i) \geq \expect\bigg[ \sup_{\bx_i\in B(\zero_d,\rho_n\halfpower)} |J_{ink}(\bx_i)| \bigg] + t \bigg\} \\
% &\quad= \prob\bigg\{ \inf_{\bx_i\in B(\zero_d,\rho_n\halfpower)}-J_{ink}(\bx_i) \geq \expect\bigg[ \sup_{\bx_i\in B(\zero_d,\rho_n\halfpower)} |J_{ink}(\bx_i)| \bigg] + t \bigg\} \\
% &\quad\leq \prob\bigg\{ \sup_{\bx_i\in B(\zero_d,\rho_n\halfpower)}-J_{ink}(\bx_i) \geq \expect\bigg[ \sup_{\bx_i\in B(\zero_d,\rho_n\halfpower)} |J_{ink}(\bx_i)| \bigg] + t \bigg\} \\
&\quad\leq \prob\bigg\{ \sup_{\bx_i\in B(\zero_d,\rho_n\halfpower)}-J_{ink}(\bx_i) \geq \expect\bigg[ \sup_{\bx_i\in B(\zero_d,\rho_n\halfpower)} -J_{ink}(\bx_i) \bigg] + t \bigg\}
 % \\
% &\quad
\leq 
% \exp\bigg\{ \frac{-(n\rho_n\halfpower/(2C_\delta))^2t^2}{2(4n\rho_n + (n\rho_n\halfpower/(2C_\delta))t)} \bigg\} 
% \\&
% = 
\exp\bigg\{ \frac{-n t^2}{32C_\delta^2 + 4C_\delta\rho_n\invhalfpower t} \bigg\},
\end{align*}
from both of which we have
\[
\prob\bigg\{ \bigg|\sup_{\bx_i\in B(\zero_d,1)}J_{ink}(\bx_i)\bigg| \geq \expect\bigg[ \sup_{\bx_i\in B(\zero_d,1)} \bigg|J_{ink}(\bx_i)\bigg| \bigg] + t \bigg\} \leq 2\exp\bigg\{ \frac{-n t^2}{32C_\delta^2 + 4C_\delta\rho_n\invhalfpower t} \bigg\}.
\]
Take $t=8C_tC_\delta\sqrt{(\log n)/n}$, and recall the bound on $\expect[\sup_{\bx_i\in B(\zero_d,\rho_n\halfpower)} |J_{ink}(\bx_i)|]$, we have 
\begin{align*}
&\prob\left\{ \left|\sup_{\bx_i\in B(\zero_d,1)}J_{ink}(\bx_i)\right| \geq ((1+d) C + 1) 2C_\delta \rho_n\halfpower K(n) + 8C_tC_\delta\sqrt{\frac{\log n}{n}} \right\}
 \\
&\quad\leq 2\exp\left\{ \frac{-64C_t^2C_\delta^2n (\log n/n)}{32C_\delta^2 + 32C_tC_\delta^2 \sqrt{(\log n)/(n\rho_n)}} \right\} 
% \\
% &\quad=2 \exp\left\{ \frac{-2C_t^2\log n}{1 + C_t\sqrt{(\log n)/(n\rho_n)}} \right\} 
% \\&
= 2n^{-2C_t^2/(1 + C_t\sqrt{(\log n)/(n\rho_n)})}.
\end{align*}
For any $c>0$, let $N_{c,\delta,\lambda}$ be an integer large enough such that $\rho_n\halfpower{}K(n)<\sqrt{(\log n)/n}$ for all $n\geq N_{c,\delta,\lambda}$ (recall that $K(n)=\sqrt{(12/\log 2)}(n\rho_n)\invhalfpower + 4(n\rho_n)\inverse$ and that we assume $(\log n)/(n\rho_n)\to 0$), and then choose $C_t$ large enough such that $(\log 2)/(\log n) - 2C_t^2/(1 + C_t\sqrt{(\log n)/(n\rho_n)}) < -c$ for all $n\geq N_{c,\delta,\lambda}$ (it is decreasing in $C_t$ and in $n$), we have
\[
\bigg|\sup_{\bx_i\in B(\zero_d,\rho_n\halfpower)}J_{ink}(\bx_i)\bigg|\
\leq (((1+d) C + 1) 2C_\delta + 8C_tC_\delta) \sqrt{\frac{\log n}{n}}
\asymp_{c,\delta,\lambda} \sqrt{\frac{\log n}{n}}
\]
with probability at least $1-n^{-c}$ (the embedding dimension implicitly depends on $\lambda$).

\noindent{}By the fact that $\inf_x f(x) \leq \sup_x f(x)$ and $-\inf_x f(x) = \sup_x -f(x)$, the previous computation also gives (omitting some details)
\[
\prob\bigg\{ \bigg|\inf_{\bx_i\in B(\zero_d,\rho_n\halfpower)}J_{ink}(\bx_i)\bigg| \geq \expect\bigg[ \sup_{\bx_i\in B(\zero_d,\rho_n\halfpower)} \bigg|J_{ink}(\bx_i)\bigg| \bigg] + t \bigg\} \leq 2\exp\bigg\{ \frac{-n t^2}{32C_\delta^2 + 4C_\delta\rho_n\invhalfpower t} \bigg\},
\]
which gives
\[
\bigg|\inf_{\bx_i\in B(\zero_d,\rho_n\halfpower)}J_{ink}(\bx_i)\bigg|\
\lesssim_{c,\delta,\lambda} \sqrt{\frac{\log n}{n}},
\]
with probability at least $1-n^{-c}$. Then by the fact that $\sup_x |f(x)| \leq |\sup_x f(x)| + |\inf_x f(x)|$, we have
\[
\sup_{\bx_i\in B(\zero_d,\rho_n\halfpower)}\left|J_{ink}(\bx_i)\right|\
\lesssim_{c,\delta,\lambda} \sqrt{\frac{\log n}{n}}
\]
with probability at least $1-n^{-c}$. So we have
\begin{align*}
\sup_{\bx_i\in B(\zero_d,\rho_n\halfpower)}\left\|\frac{1}{n}\sum_{j=1}^n (A_{ij}-p_{0ij})g_n(\rho_n\halfpower\bx_i\transpose\eye_{p,q}\bx_{0j})\rho_n\halfpower\eye_{p,q}\bx_{0j}\right\|_2 
% \leq \sum_{k=1}^d \sup_{\bx_i\in B(\zero_d,\rho_n\halfpower)}\left|J_{ink}(\bx_i)\right|\ 
\lesssim_{c,\delta,\lambda} \sqrt{\frac{\log n}{n}}
\end{align*}
with probability at least $1-n^{-c}$. This shows the bound for the fifth term in \ref{eqn:gradient-decomp}.

\noindent{}By a similar approach, for the sixth term, we have (omitting the details to save space)
\begin{align*}
\sup_{\bx_i\in B(\zero_d,\rho_n\halfpower)}\left\|\frac{1}{n}\sum_{j=1}^n (A_{ij}-p_{0ij})g_n(1-\rho_n\halfpower\bx_i\transpose\eye_{p,q}\bx_{0j})\rho_n\halfpower\eye_{p,q}\bx_{0j}\right\|_2 \lesssim_{c,\delta,\lambda} \rho_n\sqrt{\frac{\log n}{n}}
\end{align*}
with probability at least $1-n^{-c}$.
The conclusion of the lemma then follows from applying triangle inequality and combining the six large probability bounds above, then with a union bound over $i\in[n]$.
\end{proof}

\begin{lemma}[Concentration of the Hessian matrix]
\label{lemma:concentration-hessian}
Suppose Assumption \ref{assumption:esl-grdpg} holds.
Then for any $c > 0$, there exists a constant integer $N_{c,\delta,\lambda}\in\mathbb{N}_+$ and a positive constant $C_{c,\delta,\lambda}$ depending on $c,\delta,\lambda$, such that for all $n\geq N_{c,\delta,\lambda}$,
\[
\max_{i\in[n]}
\sup_{\bx_i:||\bW\transpose\bx_i-\rho_n\halfpower\bx_{0i}||_2\leq \varepsilon_n}
\left\Vert
-\frac{1}{n}\bW\transpose\frac{\partial^2\widehat\ell_{in}}{\partial\bx_i\partial\bx_i\transpose}(\bx_i)\bW - \bG_{0in}
\right\Vert_2
\lesssim_{c,\delta,\lambda} \sqrt{\frac{\log n}{n\rho_n}}
\]
with probability at least $1 - n^{-c}$, where $\varepsilon_n=C_{c,\delta,\lambda}\sqrt{(\log n)/n}$.
% \begin{align*}
% \frac{1}{n}\frac{\partial^2\widehat\ell_{in}}{\partial\bx_i\partial\bx_i\transpose}(\bx_i)
% & = \frac{1}{n} \sum_{j=1}^n \{A_{ij}\psi_n''(\bx_i\transpose\widetilde{\bx}_j) + (1 - A_{ij})\psi_n''(1 - \bx_i\transpose\widetilde{\bx}_j)\} \widetilde\bx_j\widetilde\bx_j\transpose,\\
% \bG_{0in}
% &= \frac{1}{n} \sum_{j=1}^n \left\{ \frac{1}{\rho_n\bx_{0i}\transpose\eye_{p,q}\bx_{0j}} + \frac{1}{1-\rho_n\bx_{0i}\transpose\eye_{p,q}\bx_{0j}} \right\} \rho_n\eye_{p,q}\bx_{0j}\bx_{0j}\transpose\eye_{p,q},
% \end{align*}
% where $\widetilde\bX$ denotes the signature-adjusted adjacency spectral embedding.
\end{lemma}

\begin{proof}
% [Proof of Lemma \ref{lemma:concentration-hessian}]
For simplicity of notation,
Let $p_{0ij} = \rho_n\bx_{0i}\transpose\eye_{p,q}\bx_{0j}$. 
% Note that the function $f(t)=1/\max(t,\tau_n)^2$ on $t>0$ is a truncation of the function $f(t)=1/t^2$ on $t>0$.
A simple algebra shows that
\begin{align*}
&-\frac{1}{n}\bW\transpose\frac{\partial^2\widehat\ell_{in}}{\partial\bx_i\partial\bx_i\transpose}(\bx_i)\bW - \bG_{0in} \\
% &\quad = \frac{1}{n} \sum_{j=1}^n \left\{-A_{ij}\psi_n''(\bx_i\transpose\widetilde\bx_j,\tau_n)^2 - (1-A_{ij})\psi_n''(1-\bx_i\transpose\widetilde\bx_j,\tau_n)^2 \right\} \bW\transpose\widetilde\bx_j\widetilde\bx_j\transpose\bW\\
% &\qquad-
% \frac{1}{n} \sum_{j=1}^n \left\{ \frac{1}{p_{0ij}} + \frac{1}{1-p_{0ij}} \right\} \rho_n\eye_{p,q}\bx_{0j}\bx_{0j}\transpose\eye_{p,q} \\
&\quad = \frac{1}{n} \sum_{j=1}^n A_{ij}\left\{-\psi_n''(\bx_i\transpose\widetilde{\bx}_j) + \psi_n''(p_{0ij})
\right\} \bW\transpose\widetilde\bx_j\widetilde\bx_j\transpose\bW \\
&\qquad + \frac{1}{n} \sum_{j=1}^n (1-A_{ij})\left\{ 
-\psi_n''(1 - \bx_i\transpose\widetilde{\bx}_j) + \psi_n''(1 - p_{0ij})
\right\} \bW\transpose\widetilde\bx_j\widetilde\bx_j\transpose\bW \\
&\qquad + \frac{1}{n} \sum_{j=1}^n (A_{ij}-p_{0ij}) \left\{ 
-\psi_n''(p_{0ij}) + \psi_n''(1 - p_{0ij})
\right\} \left( \bW\transpose\widetilde\bx_j\widetilde\bx_j\transpose\bW - \rho_n\eye_{p,q}\bx_{0j}\bx_{0j}\transpose\eye_{p,q} \right) \\
&\qquad + \frac{1}{n} \sum_{j=1}^n (A_{ij}-p_{0ij}) \left\{ 
-\psi_n''(p_{0ij}) + \psi_n''(1 - p_{0ij})
\right\} \rho_n\eye_{p,q}\bx_{0j}\bx_{0j}\transpose\eye_{p,q} \\
&\qquad + \frac{1}{n} \sum_{j=1}^n \frac{\bW\transpose\widetilde\bx_j\widetilde\bx_j\transpose\bW - \rho_n\eye_{p,q}\bx_{0j}\bx_{0j}\transpose\eye_{p,q}}{p_{0ij}(1-p_{0ij})} \\
&\qquad + \frac{1}{n} \sum_{j=1}^n \left\{ 
-p_{0ij}\psi_n''(p_{0ij}) - (1 - p_{0ij})\psi_n''(1 - p_{0ij})
- \left( \frac{1}{p_{0ij}} + \frac{1}{1-p_{0ij}} \right) \right\} \bW\transpose\widetilde\bx_j\widetilde\bx_j\transpose\bW,
\end{align*}
which can be viewed as a sum of six terms. For simplicity of notation, in the remaining of the proof of this lemma, the large probability bounds are stated with respect to all $n\geq{}{N}_{c,\delta,\lambda}$ for some large constant integer ${N}_{c,\delta,\lambda}$ that depends on $c,\delta,\lambda$.

\noindent{}For the first term, with probability at least $1-n^{-c}$,
\begin{align*}
&\sup_{\bx_i\in B(\rho_n\halfpower\bW\bx_{0i}, \varepsilon_n)}
\left\Vert \frac{1}{n} \sum_{j=1}^n A_{ij}\left\{
-\psi_n''(\bx_i\transpose\widetilde{\bx}_j) + \psi_n''(p_{0ij})
\right\} \bW\transpose\widetilde\bx_j\widetilde\bx_j\transpose\bW \right\Vert_2 \\
% &\leq \frac{1}{n}||\bA||_\infty \cdot \max_{j\in[n]} \sup_{\bx_i\in B(\rho_n\halfpower\bW\bx_{0i}, \varepsilon_n)}\left| \frac{1}{\max(\bx_i\transpose\widetilde\bx_j,\tau_n)^2} - \frac{1}{\max(p_{0ij},\tau_n)^2} \right| \cdot \max_{j\in[n]}||\widetilde\bx_j||_2^2 \\
&\quad\leq 
% \frac{1}{n}||\bA||_\infty \cdot \max_{j\in[n]} \sup_{\bx_i\in B(\rho_n\halfpower\bW\bx_{0i}, \varepsilon_n)}\left| \frac{1}{(\bx_i\transpose\widetilde\bx_j)^2} - \frac{1}{p_{0ij}^2} \right| \cdot \max_{j\in[n]}||\widetilde\bx_j||_2^2 \\
% &\quad= 
\frac{1}{n}||\bA||_\infty \cdot \max_{j\in[n]} \sup_{\bx_i\in B(\rho_n\halfpower\bW\bx_{0i}, \varepsilon_n)} \frac{2|\bx_i\transpose\widetilde\bx_j - p_{0ij}|}{|p_{0ij}+\theta(\bx_i\transpose\widetilde\bx_j - p_{0ij})|^3} \cdot \max_{j\in[n]}||\widetilde\bx_j||_2^2
 % \\
% &\quad\lesssim_{c,\delta,\lambda} \frac{1}{n} \cdot n\rho_n \cdot \rho_n\halfpower\sqrt{\frac{\log n}{n}} \cdot \frac{1}{\rho_n^3} \cdot \rho_n 
% \\&
\asymp_{c,\delta,\lambda} \sqrt{\frac{\log n}{n\rho_n}}
\end{align*}
by from Cauchy--Schwarz inequality, the properties of the function $\psi_n(t)$, mean value theorem, $(\log n)/(n\rho_n) \to 0$, and Lemma \ref{lemma:frequently-used-results}.

\noindent{}For the second term, with probability at least $1-n^{-c}$,
\begin{align*}
&\sup_{\bx_i\in B(\rho_n\halfpower\bW\bx_{0i}, \varepsilon_n)}
\left\Vert \frac{1}{n} \sum_{j=1}^n (1-A_{ij})\left\{ 
-\psi_n''(1 - \bx_i\transpose\widetilde{\bx}_j) + \psi_n''(1 - p_{0ij})
\right\} \bW\transpose\widetilde\bx_j\widetilde\bx_j\transpose\bW \right\Vert_2 \\
% &\leq \max_{j\in[n]} \sup_{\bx_i\in B(\rho_n\halfpower\bW\bx_{0i}, \varepsilon_n)}
% \left| \frac{1}{\max(1-\bx_i\transpose\widetilde\bx_j,\tau_n)^2} - \frac{1}{\max(1-p_{0ij},\tau_n)^2} \right| \cdot \max_{j\in[n]}||\widetilde\bx_j||_2^2 \\
&\quad\leq 
% \max_{j\in[n]} \sup_{\bx_i\in B(\rho_n\halfpower\bW\bx_{0i}, \varepsilon_n)}
% \left| \frac{1}{(1-\bx_i\transpose\widetilde\bx_j)^2} - \frac{1}{(1-p_{0ij})^2} \right| \cdot \max_{j\in[n]}||\widetilde\bx_j||_2^2 \\
% &\quad = 
\max_{j\in[n]} \sup_{\bx_i\in B(\rho_n\halfpower\bW\bx_{0i}, \varepsilon_n)} \frac{2|\bx_i\transpose\widetilde\bx_j - p_{0ij}|}{|1-(p_{0ij}+\theta(\bx_i\transpose\widetilde\bx_j - p_{0ij}))|^3} \cdot \max_{j\in[n]}||\widetilde\bx_j||_2^2 \\
&\quad\lesssim_{c,\delta,\lambda} \rho_n\halfpower\sqrt{\frac{\log{}n}{n}} \cdot \rho_n
% \\&
\asymp_{c,\delta,\lambda} \rho_n^2 \sqrt{\frac{\log n}{n\rho_n}},
\end{align*}
where the inequalities follow from Cauchy--Schwarz inequality, the properties of the function $\psi_n(t)$, mean value theorem, $(\log n)/(n\rho_n) \to 0$, and Lemma \ref{lemma:frequently-used-results}.

\noindent{}For the third term, with probability at least $1-n^{-c}$,
\begin{align*}
&\left\Vert \frac{1}{n} \sum_{j=1}^n (A_{ij}-p_{0ij}) \left\{
-\psi_n''(p_{0ij}) + \psi_n''(1 - p_{0ij})
\right\} \left( \bW\transpose\widetilde\bx_j\widetilde\bx_j\transpose\bW - \rho_n\eye_{p,q}\bx_{0j}\bx_{0j}\transpose\eye_{p,q} \right) \right\Vert_2 \\
&\quad\leq \frac{1}{n} \left\Vert \bA-\rho_n\bX_0\eye_{p,q}\bX_0\transpose \right\Vert_\infty \cdot \max_{j\in[n]} \left( \frac{1}{p_{0ij}^2} + \frac{1}{(1-p_{0ij})^2} \right) \cdot \max_{j\in[n]} \left\Vert \bW\transpose\widetilde\bx_j\widetilde\bx_j\transpose\bW - \rho_n\bx_{0j}\bx_{0j}\transpose \right\Vert_2\\
&\quad\lesssim_{c,\delta,\lambda} \frac{1}{n}n\rho_n \cdot \frac{1}{\rho_n^2} \cdot \rho_n\halfpower\sqrt{\frac{\log{}n}{n}} 
% \\&
\asymp_{c,\delta,\lambda} \sqrt{\frac{\log{}n}{n\rho_n}},
\end{align*}
where the inequalities follow from Cauchy--Schwarz inequality, triangle inequality, Lemma \ref{lemma:frequently-used-results}, and Lemma \ref{lemma:results-with-Aij}.

\noindent{}For the fourth term, note that
% \begin{align*}
$\left\{
-\psi_n''(p_{0ij}) + \psi_n''(1 - p_{0ij})
% \frac{1}{\max(p_{0ij},\tau_n)^2} - \frac{1}{\max(1-p_{0ij},\tau_n)^2}
\right\}\rho_n  
= \{p_{0ij}^{-2} - (1-p_{0ij})^{-2}\}\rho_n
\leq C_\delta \rho_n\inverse,
$
% \end{align*}
so
\[
\left\Vert \frac{1}{n} \sum_{j=1}^n (A_{ij}-p_{0ij}) \left\{ 
-\psi_n''(p_{0ij}) + \psi_n''(1 - p_{0ij})
% \frac{1}{\max(p_{0ij},\tau_n)^2} - \frac{1}{\max(1-p_{0ij},\tau_n)^2} 
\right\} \rho_n\bx_{0j}\bx_{0j}\transpose \right\Vert_2
\lesssim_{c,\delta,\lambda}
\sqrt{\frac{\log{}n}{n\rho_n}}
\]
with probability at least $1-n^{-c}$ by Lemma \ref{lemma:results-with-Aij}.

\noindent{}For the fifth term, with probability at least $1-n^{-c}$,
\[
\left\Vert \frac{1}{n} \sum_{j=1}^n \frac{\bW\transpose\widetilde\bx_j\widetilde\bx_j\transpose\bW - \rho_n\eye_{p,q}\bx_{0j}\bx_{0j}\transpose\eye_{p,q}}{p_{0ij}(1-p_{0ij})} \right\Vert_2
\leq
\frac{\max_{j}\left\Vert\bW\transpose\widetilde\bx_j\widetilde\bx_j\transpose\bW - \rho_n\bx_{0j}\bx_{0j}\transpose\right\Vert_2}{\min_{i,j}p_{0ij}(1-p_{0ij})}
\lesssim_{c,\delta,\lambda}
\sqrt{\frac{\log{}n}{n\rho_n}}
\]
by triangle inequality and Lemma \ref{lemma:frequently-used-results}.

\noindent{}For the sixth term, since $\tau_n<\delta\rho_n$ for all $n$ and $p_{0ij}\in[\delta\rho_n, (1-\delta)\rho_n]$ for all $i,j\in[n]$ by assumption, we have
\[
\frac{1}{n} \sum_{j=1}^n \left\{ 
{ -p_{0ij}\psi_n''(p_{0ij})-(1 - p_{0ij})\psi_n''(1 - p_{0ij})}
% \frac{p_{0ij}}{\max(p_{0ij},\tau_n)^2} + \frac{1-p_{0ij}}{\max(1-p_{0ij},\tau_n)^2}
- \left( \frac{1}{p_{0ij}} + \frac{1}{1-p_{0ij}} \right) \right\} \bW\transpose\widetilde\bx_j\widetilde\bx_j\transpose\bW
= 0.
\]

\noindent{}The conclusion follows from applying triangle inequality, combining the six bounds above, and applying a union bound over $i\in[n]$.
\end{proof}

\begin{lemma}[Lipschitz property of the Hessian matrix]
\label{lemma:lipschitz-hessian}
Suppose Assumption \ref{assumption:esl-grdpg} holds. Let $\bx_i,\bx_i'\in\mathbb{R}^d$, then there exists a constant $N_{c,\delta,\lambda}\in\mathbb{N}_+$ depending on $c,\delta,\lambda$, such that for all $n \geq N_{c,\delta,\lambda}$,
\[
\frac{1}{n} \left\|\frac{\partial^2\widehat\ell_{in}}{\partial\bx_i\partial\bx_i\transpose}(\bx_i) - \frac{\partial^2\widehat\ell_{in}}{\partial\bx_i\partial\bx_i\transpose}(\bx_i')\right\|_2
\lesssim_{c,\delta,\lambda} \rho_n^{-3/2}\|\bx_i-\bx_i'\|_2
\]
with probability at least $1 - n^{-c}$.
\end{lemma}

\begin{proof}
% [Proof of Lemma \ref{lemma:lipschitz-hessian}]
Write
\begin{align*}
\frac{1}{n}\frac{\partial^2 \widehat\ell_{in}}{\partial\bx_i\partial\bx_i\transpose}(\bx_i) - \frac{1}{n}\frac{\partial^2 \widehat\ell_{in}}{\partial\bx_i\partial\bx_i\transpose}(\bx_i')
% &\quad = \frac{1}{n}\sum_{j=1}^n \left\{A_{ij}\psi_n''(\bx_i\transpose\widetilde\bx_j) + (1-A_{ij})\psi_n''(1-\bx_i\transpose\widetilde\bx_j)\right\}\widetilde\bx_j\widetilde\bx_j\transpose \\
% &\quad\quad - \frac{1}{n}\sum_{j=1}^n \left\{A_{ij}\psi_n''((\bx_i')\transpose\widetilde\bx_j) + (1-A_{ij})\psi_n''(1-(\bx_i')\transpose\widetilde\bx_j)\right\}\widetilde\bx_j\widetilde\bx_j\transpose \\
&= \frac{1}{n}\sum_{j=1}^n A_{ij}\left\{\psi_n''(\bx_i\transpose\widetilde\bx_j) - \psi_n''((\bx_i')\transpose\widetilde\bx_j)\right\}\widetilde\bx_j\widetilde\bx_j\transpose  \\
&\quad + \frac{1}{n}\sum_{j=1}^n (1-A_{ij})\left\{\psi_n''(1-\bx_i\transpose\widetilde\bx_j) - \psi_n''(1-(\bx_i')\transpose\widetilde\bx_j)\right\}\widetilde\bx_j\widetilde\bx_j\transpose,
\end{align*}
which can be viewed as a sum of two terms.
For the first term,
\begin{align*}
&\left\|\frac{1}{n}\sum_{j=1}^n A_{ij}\left\{\psi_n''(\bx_i\transpose\widetilde\bx_j) - \psi_n''((\bx_i')\transpose\widetilde\bx_j)\right\}\widetilde\bx_j\widetilde\bx_j\transpose\right\|_2  \\
&\quad\leq \frac{1}{n}\|\bA\|_\infty \cdot \max_{j\in[n]}\left|\psi_n'''(\theta\bx_i\transpose\widetilde\bx_j+(1-\theta)(\bx_i')\transpose\widetilde\bx_j)(\bx_i-\bx_i')\transpose\widetilde\bx_j\right| \cdot \max_{j\in[n]}\|\widetilde\bx_j\|_2^2  \\
&\quad\leq \frac{1}{n}\|\bA\|_\infty \cdot \tau_n^{-3} \cdot \|\bx_i-\bx_i'\|_2 \cdot \max_{j\in[n]}\|\widetilde\bx_j\|_2^3  
% \\
% &\quad
\lesssim_{c,\delta,\lambda} 
% \frac{1}{n}n\rho_n \cdot \rho_n^{-3} \cdot \rho_n^{3/2} \cdot \|\bx_i-\bx_i'\|_2  
% \\&
% \asymp_{c,\delta,\lambda}
 \rho_n\invhalfpower \|\bx_i-\bx_i'\|_2
\end{align*}
with probability at least $1 - n^{-c}$.
For the second term,
\begin{align*}
&\left\| \frac{1}{n}\sum_{j=1}^n (1-A_{ij})\left\{\psi_n''(1-\bx_i\transpose\widetilde\bx_j) - \psi_n''(1-(\bx_i')\transpose\widetilde\bx_j)\right\}\widetilde\bx_j\widetilde\bx_j\transpose\right\|_2  \\
&\quad\leq \max_{j\in[n]}\left|\psi_n'''(1-\theta\bx_i\transpose\widetilde\bx_j-(1-\theta)(\bx_i')\transpose\widetilde\bx_j)(\bx_i-\bx_i')\transpose\widetilde\bx_j\right| \cdot \max_{j\in[n]}\|\widetilde\bx_j\|_2^2  \\
&\quad\leq \tau_n^{-3} \cdot \|\bx_i-\bx_i'\|_2 \cdot \max_{j\in[n]}\|\widetilde\bx_j\|_2^3  
% \\&\quad
\lesssim_{c,\delta,\lambda} \rho_n^{-3/2} \|\bx_i-\bx_i'\|_2
\end{align*}
with probability at least $1 - n^{-c}$.
So we have
\[
\frac{1}{n} \left\|\frac{\partial^2\widehat\ell_{in}}{\partial\bx_i\partial\bx_i\transpose}(\bx_i) - \frac{\partial^2\widehat\ell_{in}}{\partial\bx_i\partial\bx_i\transpose}(\bx_i')\right\|_2
\lesssim_{c,\delta,\lambda} \rho_n^{-3/2}\|\bx_i-\bx_i'\|_2
\]
with probability at least $1 - n^{-c}$.
\end{proof}

\begin{theorem}[One-step estimator]
\label{thm:ose-asymptotic-normality}
Suppose Assumption \ref{assumption:esl-grdpg} holds. Let $\breve\bX$ denote the adjacency spectral embedding, and $\widetilde\bX$ the signature-adjusted adjacency spectral embedding.
Let $p_{0ij} = \rho_n\bx_{0i}\transpose\eye_{p,q}\bx_{0j}$, and $\widetilde{p}_{ij} = \breve\bx_i\transpose\widetilde{\bx}_j$, $i,j\in [n]$.
For each $i\in[n]$, define the one-step estimator $\widehat\bx_i^{\mathrm{(OS)}}$ by
\[
\widehat\bx_i^{\mathrm{(OS)}}
= \breve\bx_i + \left\{\frac{1}{n}\sum_{j=1}^n\frac{\widetilde\bx_j\widetilde\bx_j\transpose}{\widetilde{p}_{ij}(1-\widetilde{p}_{ij})} \right\}\inverse \left\{\frac{1}{n}\sum_{j=1}^n\frac{(A_{ij}-\widetilde{p}_{ij})\widetilde\bx_j}{\widetilde{p}_{ij}(1-\widetilde{p}_{ij})} \right\}.
\]
Then
\[
\bG_{0in}\halfpower(\bW\transpose\widehat\bx_i^{\mathrm{(OS)}}-\rho_n\halfpower\bx_{0i})
= \frac{1}{n}\sum_{j=1}^n \frac{(A_{ij}-p_{0ij})\bG_{0in}\invhalfpower\rho_n\halfpower\eye_{p,q}\bx_{0j}}{p_{0ij}(1-p_{0ij})} + \br_{in}^{\mathrm{(OS)}},
\]
% where
% % \[
% $\bG_{0in} = (1/n)\sum_{j=1}^n{\rho_n\eye_{p,q}\bx_{0j}\bx_{0j}\transpose\eye_{p,q}}/\{p_{0ij}(1-p_{0ij})\}$,
% % \]
and for any $c>0$, there exist a constant integer $N_{c,\delta,\lambda}$ and a constant $C_{c,\delta,\lambda}$ that depend on $c,\delta,\lambda$ such that for all $1\leq t\leq C_{c,\delta,\lambda}\log n$ and for all $n\geq N_{c,\delta,\lambda}$,
% \[
$\|\br_{in}^{\mathrm{(OS)}}\|_2\lesssim_{c,\delta,\lambda} {t^2}/{(n\rho_n\halfpower)}$
% \]
with probability at least $1-c_0n^{-c}-c_0e^{-t}$ for some absolute constant $c_0>0$.
Furthermore,
% \[
$\sqrt{n}\bG_{0in}\halfpower(\bW\transpose\widehat\bx_i^{\mathrm{(OS)}}-\rho_n\halfpower\bx_{0i})
\overset{\calL}{\to} \mathrm{N}_d(\zero_d, \eye_d)$ as $n\to\infty$.
% \]
\end{theorem}
\begin{remark}
Theorem \ref{thm:ose-asymptotic-normality} is a generalization of Theorem 4.7 in \cite{xie-2024-spn-bernoulli} to generalized random dot product graphs in our settings. The proof is mostly identical to its original version, with slight modifications such as the presence of the signature matrix $\eye_{p,q}$ and the different definition of the orthogonal alignment matrix $\bW$. We omit the proof and use the theorem directly. 
\end{remark}

\begin{lemma}[Theorem 1 in \cite{pena-montgomery-1995-decoupling}]
\label{lemma:decoupling-u-statistics}
Let $\{X_i\}$ be a sequence of independent random variables on a measurable $(S,\mathscr{S})$ space and let $\{X_i^{(1)}\}$, $\{X_i^{(2)}\}$ be two independent copies of $\{X_i\}$. Let $f_{i_1i_2}$ be families of functions of two variables taking $(S\times S)$ into a Banach space $(B,\|\cdot\|)$. Then, for all $n\geq 2$, $t>0$, there exist a numerical constant $C$ such that
\[
\prob\bigg\{\bigg\Vert\sum_{1\leq i_1\neq i_2\leq n}f_{i_1i_2}(X_{i_1}^{(1)},X_{i_2}^{(1)})\bigg\Vert\geq t\bigg\}
\leq C\prob\bigg\{C\bigg\Vert\sum_{1\leq i_1\neq i_2\leq n}f_{i_1i_2}(X_{i_1}^{(1)},X_{i_2}^{(2)})\bigg\Vert\geq t\bigg\}.
\]
\end{lemma}

\begin{lemma}[A weak law of large numbers]
\label{lemma:grdpg-inid-wlln}
Suppose Assumption \ref{assumption:esl-grdpg} holds.
Let
\[
Z\equiv Z(\bA) =
\sum_{i=1}^n\bigg\Vert \frac{1}{n}\sum_{j=1}^n \frac{(A_{ij}-p_{0ij})\bG_{0in}\inverse\rho_n\halfpower\eye_{p,q}\bx_{0j}}{p_{0ij}(1-p_{0ij})} \bigg\Vert_2^2,
\]
then $Z=\expect[Z]+o_{\prob}(1)$, where $\expect[Z]=(1/n)\sum_{i=1}^n\mathrm{tr}(\bG_{0in}\inverse)$.
\end{lemma}

\begin{proof}
% [Proof of Lemma \ref{lemma:grdpg-inid-wlln}]
Recall $p_{0ij}=\rho_n\bx_{0i}\transpose\eye_{p,q}\bx_{0j}$ and
% \[
$\bG_{0in} = (1/n)\sum_{j=1}^n {\rho_n\eye_{p,q}\bx_{0j}\bx_{0j}\transpose\eye_{p,q}}/\{p_{0ij}(1-p_{0ij})\}$.
% \]
For $i,j\in[n]$, let $E_{ij}=A_{ij}-p_{0ij}$, and
% \[
$\bgamma_{ij} = {\bG_{0in}\inverse\rho_n\halfpower\eye_{p,q}\bx_{0j}}/\{n p_{0ij}(1-p_{0ij})\}$,
% \]
then we can write
\begin{align*}
Z &= \sum_{i=1}^n \bigg\Vert\sum_{j=1}^n E_{ij}\bgamma_{ij}\bigg\Vert_2^2 
% \\
% &= \sum_{i=1}^n\sum_{a=1}^n\sum_{b=1}^n E_{ia}E_{ib}\bgamma_{ia}\transpose\bgamma_{ib} \\
% &
= \sum_{i=1}^n\sum_{j=1}^n E_{ij}^2\|\bgamma_{ij}\|_2^2 + \sum_{i=1}^n\sum_{a=1}^n\sum_{b=1}^n E_{ia}E_{ib}\bgamma_{ia}\transpose\bgamma_{ib}\one(a\neq b).
\end{align*}
We have $\expect[E_{ij}]=0$, $\expect[E_{ij}^2]=p_{0ij}(1-p_{0ij})$, and $\|\bgamma_{ij}\|_2\asymp_{\delta,\lambda}1/(n\rho_n^{0.5})$ by assumption and the result that $\|\bG_{0in}\inverse\|_2\asymp_{\delta,\lambda}1$ which is shown in \eqref{eqn:fisher-information-order} on page \pageref{eqn:fisher-information-order} (in the proof of asymptotic normality part of Theorem \ref{thm:mesle}). By Bernstein's inequality,
\begin{align*}
&\prob\bigg\{\bigg|\sum_{i=1}^n\sum_{j=1}^n E_{ij}^2\|\bgamma_{ij}\|_2^2 - \expect\bigg[\sum_{i=1}^n\sum_{j=1}^n E_{ij}^2\|\bgamma_{ij}\|_2^2\bigg]\bigg| \geq t\bigg\}  \\
&\quad\leq 2\exp\bigg\{ \frac{-3t^2}{6\sum_{i=1}^n\sum_{j=1}^n\var(E_{ij}^2\|\bgamma_{ij}\|_2^2) + 2\max_{i,j\in[n]}\|\bgamma_{ij}\|_2^2t} \bigg\} 
% \\&
\leq 2\exp\bigg\{\frac{-n^2\rho_nt}{C_{\delta,\lambda}+C_{\delta,\lambda}t}\bigg\},
\end{align*}
from which by taking $t=C_{c,\delta,\lambda}\sqrt{(\log n)/n^2\rho_n}$ we have
\[
\bigg|\sum_{i=1}^n\sum_{j=1}^n E_{ij}^2\|\bgamma_{ij}\|_2^2 - \expect\bigg[\sum_{i=1}^n\sum_{j=1}^n E_{ij}^2\|\bgamma_{ij}\|_2^2\bigg]\bigg|
\lesssim_{c,\delta,\lambda} \sqrt{\frac{\log n}{n^2\rho_n}}
\]
with probability at least $1-n^{-c}$.
We then use Lemma \ref{lemma:decoupling-u-statistics} to deal with the sum of cross terms. Consider the sequence of random variables $\{E_{(i,j)}:(i,j)\in[n]^2\}$, with two independent copies $\{E_{ij}\}$ and $\{\bar{E}_{ij}\}$, and equip the index set $[n]^2$ with the lexicographic order, i.e., for $x=(x_1,x_2)$ and $y=(y_1,y_2)$ in $[n]^2$, we have $x<y$ if either $x_1<y_1$ or $x_1=y_1$ and $x_2<y_2$, and we have $x=y$ if $x_1=y_1$ and $x_2=y_2$. And consider the family of functions $f_{(i_1,a),(i_2,b)}(E_{(i_1,a)},E_{(i_2,b)})=E_{(i_1,a)}E_{(i_2,b)}\bgamma_{(i_1,a)}\transpose\bgamma_{(i_2,b)}\one(i_1=i_2)$. Then
\begin{align*}
% &\quad 
\sum_{i=1}^n\sum_{a=1}^n\sum_{b=1}^n E_{ia}E_{ib}\bgamma_{ia}\transpose\bgamma_{ib}\one(a\neq b) 
% \\&
= \sum_{(1,1)\leq(i_1,a)\neq(i_2,b)\leq(n,n)}f_{(i_1,a),(i_2,b)}(E_{(i_1,a)},E_{(i_2,b)}).
\end{align*}
By Lemma \ref{lemma:decoupling-u-statistics}, conditional probability, and Bernstein's inequality,
\begin{align*}
&\prob\bigg\{\bigg| \sum_{i=1}^n\sum_{a=1}^n\sum_{b=1}^n E_{ia}E_{ib}\bgamma_{ia}\transpose\bgamma_{ib}\one(a\neq b) \bigg|\geq t\bigg\}  \\
% &\quad= \prob\left\{\left| \sum_{(1,1)\leq(i_1,a)\neq(i_2,b)\leq(n,n)}f_{(i_1,a),(i_2,b)}(E_{(i_1,a)},E_{(i_2,b)}) \right|\geq t\right\}  \\
&\quad\leq C\prob\bigg\{C\bigg| \sum_{(1,1)\leq(i_1,a)\neq(i_2,b)\leq(n,n)}f_{(i_1,a),(i_2,b)}(E_{(i_1,a)},\bar{E}_{(i_2,b)}) \bigg|\geq t\bigg\}  \\
% &\quad= C\prob\bigg\{C\bigg| \sum_{i=1}^n\sum_{a=1}^n\sum_{b=1}^n E_{ia}\bar{E}_{ib}\bgamma_{ia}\transpose\bgamma_{ib}\one(a\neq b) \bigg|\geq t\bigg\}  \\
% &\quad= \expect\bigg[C\prob\bigg\{C\bigg| \sum_{i=1}^n\sum_{a=1}^n\sum_{b=1}^n E_{ia}\bar{E}_{ib}\bgamma_{ia}\transpose\bgamma_{ib}\one(a\neq b) \bigg|\geq t \bigg|\{\bar{E}_{ib}\}\bigg\}\bigg]  \\
&\quad\leq \expect\bigg[2C\exp\bigg\{ \frac{-3t^2}{6C^2\sum_{i,a,b}|\bar{E}_{ib}|^2|\bgamma_{ia}\transpose\bgamma_{ib}|^2\var(E_{ia}) + 2C\max_{i,a,b}|\bar{E}_{ib}\bgamma_{ia}\transpose\bgamma_{ib}|t} \bigg\}\bigg]  
\\&\quad
\leq 2C\exp\bigg\{ \frac{-n\rho_nt^2}{C_{\delta,\lambda} + C_{\delta,\lambda}t/n} \bigg\},
\end{align*}
from which by taking $t=C_{c,\delta,\lambda}\sqrt{(\log n)/n\rho_n}$ we have
\[
\bigg| \sum_{i=1}^n\sum_{a=1}^n\sum_{b=1}^n E_{ia}E_{ib}\bgamma_{ia}\transpose\bgamma_{ib}\one(a\neq b) \bigg|
\lesssim_{c,\delta,\lambda} \sqrt{\frac{\log n}{n\rho_n}}
\]
with probability at least $1-n^{-c}$. Note that
% \[
$\expect[Z] = \expect[\sum_{i=1}^n\sum_{j=1}^n E_{ij}^2\|\bgamma_{ij}\|_2^2]$  by independence of $E_{ij}$.
% \]
So by combining several previous results, we have
\begin{align*}
Z &= \sum_{i=1}^n\sum_{j=1}^n E_{ij}^2\|\bgamma_{ij}\|_2^2 + \sum_{i=1}^n\sum_{a=1}^n\sum_{b=1}^n E_{ia}E_{ib}\bgamma_{ia}\transpose\bgamma_{ib}\one(a\neq b)  
% \\
% &= \expect\left[\sum_{i=1}^n\sum_{j=1}^n E_{ij}^2\|\bgamma_{ij}\|_2^2\right] + o_{\prob}(1)  \\
% &
= \expect[Z] + o_{\prob}(1).
\end{align*}
Finally, a simple algebra shows that
\begin{align*}
\expect[Z]
% &= 
% \expect\left[ \sum_{i=1}^n\left\Vert \frac{1}{n}\sum_{j=1}^n \frac{(A_{ij}-p_{0ij})\bG_{0in}\inverse\rho_n\halfpower\eye_{p,q}\bx_{0j}}{p_{0ij}(1-p_{0ij})} \right\Vert_2^2 \right]  \\
% &= \frac{1}{n^2}\sum_{i=1}^n\sum_{a=1}^n\sum_{b=1}^n \frac{\expect[(A_{ia}-p_{0ia})(A_{ib}-p_{0ib})]}{p_{0ia}(1-p_{0ia})p_{0ib}(1-p_{0ib})} \rho_n\bx_{0a}\transpose\eye_{p,q}\bG_{0in}^{-2}\eye_{p,q}\bx_{0b}  \\
&= \frac{1}{n^2}\sum_{i=1}^n\sum_{j=1}^n \frac{\rho_n\bx_{0j}\transpose\eye_{p,q}\bG_{0in}^{-2}\eye_{p,q}\bx_{0j}}{p_{0ij}(1-p_{0ij})}  
% \\&
= \frac{1}{n^2}\sum_{i=1}^n \mathrm{tr}\left\{\sum_{j=1}^n \frac{\rho_n\bx_{0j}\transpose\eye_{p,q}\bG_{0in}^{-2}\eye_{p,q}\bx_{0j}}{p_{0ij}(1-p_{0ij})}\right\}  \\
&= \frac{1}{n}\sum_{i=1}^n \mathrm{tr}\left\{\frac{1}{n}\sum_{j=1}^n \frac{\rho_n\eye_{p,q}\bx_{0j}\bx_{0j}\transpose\eye_{p,q}}{p_{0ij}(1-p_{0ij})}\bG_{0in}^{-2}\right\}  
% \\&
= \frac{1}{n}\sum_{i=1}^n\mathrm{tr}\left\{\bG_{0in}\inverse\right\}.
\end{align*}
\end{proof}

\section{Proofs of the Main Results}

\subsection{Proof of Theorem \ref{thm:mesle}}

\begin{proof}
% [Proof of Theorem \ref{thm:mesle}]
 For simplicity of notation, in the proof of this theorem, the large probability bounds with probability at least $1 - n^{-c}$ are stated with respect to all $n\geq{}{N}_{c,\delta,\lambda}$ for some large constant integer ${N}_{c,\delta,\lambda}$ that depends on $c,\delta,\lambda$, where $c>0$ is an arbitrary positive constant. Also, the results that hold for a single $i\in[n]$ with probability at least $1-n^{-c}$ can be strengthened to hold for all $i\in[n]$ by taking a union bound over $i\in[n]$.

\noindent{}
\textit{Proof of existence and uniqueness.} Note that we have the average ESL function for vertex $i$
\[
\frac{1}{n}\widehat\ell_{in}(\bx_i) = \frac{1}{n}\sum_{j=1}^n \left\{A_{ij}\psi_n(\bx_i\transpose\widetilde\bx_j) + (1-A_{ij})\psi_n(1-\bx_i\transpose\widetilde\bx_j)\right\},
\]
and define the population counterpart of the average of the ESL function
\[
M_{in}(\bx_i) = \frac{1}{n}\sum_{j=1}^n \left\{p_{0ij}\psi_n(\rho_n\halfpower\bx_i\transpose\eye_{p,q}\bx_{0j}) + (1-p_{0ij})\psi_n(1-\rho_n\halfpower\bx_i\transpose\eye_{p,q}\bx_{0j})\right\}.
\]
% where $\psi_n(t) = \max\{-t^2/(2\tau_n^2) + 2t/\tau_n + (\log\tau_n - 3/2),\, \log(t)\}$ with the definition of $\log(t)=-\infty$ on $t\in(-\infty,0]$, 
% then it is straightforward to show
% \begin{align*}
% \frac{1}{n}\frac{\partial \widehat\ell_{in}}{\partial\bx_i}(\bx_i) &= \frac{1}{n}\sum_{j=1}^n \left\{A_{ij}\psi_n'(\bx_i\transpose\widetilde\bx_j) - (1-A_{ij})\psi_n'(1-\bx_i\transpose\widetilde\bx_j)\right\}\widetilde\bx_j, \\
% \frac{\partial M_{in}}{\partial\bx_i}(\bx_i) &= \frac{1}{n}\sum_{j=1}^n \left\{p_{0ij}\psi_n'(\rho_n\halfpower\bx_i\transpose\eye_{p,q}\bx_{0j}) - (1-p_{0ij})\psi_n'(1-\rho_n\halfpower\bx_i\transpose\eye_{p,q}\bx_{0j})\right\}\rho_n\halfpower\eye_{p,q}\bx_{0j},\\
% % \end{align*}
% % where $\psi_n'(t) = (-t+2\max(t,\,\tau_n))/\max(t,\,\tau_n)^2$, and
% % \begin{align*}
% \frac{1}{n}\frac{\partial^2 \widehat\ell_{in}}{\partial\bx_i\partial\bx_i\transpose}(\bx_i) &= \frac{1}{n}\sum_{j=1}^n \left\{A_{ij}\psi_n''(\bx_i\transpose\widetilde\bx_j) + (1-A_{ij})\psi_n''(1-\bx_i\transpose\widetilde\bx_j)\right\}\widetilde\bx_j\widetilde\bx_j\transpose, \\
% \frac{\partial^2 M_{in}}{\partial\bx_i\partial\bx_i\transpose}(\bx_i) &= \frac{1}{n}\sum_{j=1}^n \left\{p_{0ij}\psi_n''(\rho_n\halfpower\bx_i\transpose\eye_{p,q}\bx_{0j}) + (1-p_{0ij})\psi_n''(1-\rho_n\halfpower\bx_i\transpose\eye_{p,q}\bx_{0j})\right\}\rho_n\eye_{p,q}\bx_{0j}\bx_{0j}\transpose\eye_{p,q}.
% \end{align*}
% where $\psi_n''(t) = -1/\max(t,\,\tau_n)^2$. 
Note that $\psi_n''(t)\in[-\tau_n^{-2}, -1]$ for all $t\in\mathbb{R}$. Therefore,
\begin{align*}
-\frac{\partial^2 M_{in}}{\partial\bx_i\partial\bx_i\transpose}(\bx_i)
&= -\frac{1}{n}\sum_{j=1}^n \left\{p_{0ij}\psi_n''(\rho_n\halfpower\bx_i\transpose\eye_{p,q}\bx_{0j}) + (1-p_{0ij})\psi_n''(1 - \rho_n\halfpower\bx_i\transpose\eye_{p,q}\bx_{0j})\right\}\\ 
&\quad\times \rho_n\eye_{p,q}\bx_{0j}\bx_{0j}\transpose\eye_{p,q} \\
&\succeq \frac{1}{n}\sum_{j=1}^n \rho_n\eye_{p,q}\bx_{0j}\bx_{0j}\transpose\eye_{p,q} 
% \\&
\succeq \rho_n\lambda_d\bigg(\frac{1}{n}\sum_{j=1}^n \eye_{p,q}\bx_{0j}\bx_{0j}\transpose\eye_{p,q}\bigg) 
% \\&
\succeq \lambda\rho_n\eye_d,\\
-\frac{1}{n}\frac{\partial^2\widehat\ell_{in}}{\partial\bx_i\partial\bx_i\transpose}(\bx_i)
&= -\frac{1}{n} \sum_{j=1}^n \left\{A_{ij}\psi_n''(\bx_i\transpose\widetilde\bx_j) + (1-A_{ij})\psi_n''(1 -\bx_i\transpose\widetilde\bx_j) \right\} \widetilde\bx_j\widetilde\bx_j\transpose \\
&\succeq \frac{1}{n} \sum_{j=1}^n \widetilde\bx_j\widetilde\bx_j\transpose 
% \\&
\succeq \lambda_d\left(\frac{1}{n} \sum_{j=1}^n \widetilde\bx_j\widetilde\bx_j\transpose\right)\eye_d,
\end{align*}
in which by Theorem 5.2 in \cite{lei-rinaldo-2015-sbm} and Weyl's inequality,
\[
\lambda_d\left(\frac{1}{n} \sum_{j=1}^n \widetilde\bx_j\widetilde\bx_j\transpose\right)
= \lambda_d\left(\frac{1}{n} \widetilde\bX\transpose\widetilde\bX\right)
= \sigma_d(\frac{1}{n} \bA)
\geq \frac{1}{2}\sigma_d\left(\frac{1}{n}  \rho_n\bX_0\bX_0\transpose \right)
\geq \frac{1}{2}\lambda\rho_n
> 0
\]
with probability at least $1-n^{-c}$, so we have
\begin{equation}
\label{eqn:hessian-strongly-convex-bound}
-\frac{1}{n}\frac{\partial^2\widehat\ell_{in}}{\partial\bx_i\partial\bx_i\transpose}(\bx_i)
\succeq \frac{1}{2}\lambda\rho_n\eye_d
\end{equation}
for all $\bx_i\in\mathbb{R}^d$ with probability at least $1-n^{-c}$, i.e., $(1/n)\widehat\ell_{in}(\bx_i)$ is strongly concave over $\mathbb{R}^d$ with probability at least $1-n^{-c}$. This implies that $\argmax_{\bx_i\in\mathbb{R}^d}\widehat{\ell}_{in}(\bx_i)$ exists and is unique because $\widehat{\ell}_{in}(\bx_i)$ is clearly bounded from above.

\noindent{} 
Next, we let $\widehat{\bx}_i = \argmax_{\|\bx_i\|_2\leq \rho_n^{1/2}}\widehat{\ell}_{in}(\bx_i)$ be the local maximizer of $\widehat{\ell}_{in}(\bx_i)$ in the closed ball $\{\bx_i\in\mathbb{R}^d:\|\bx_i\|_2\leq \rho_n\halfpower\}$.
In $\{\bx_i:\,\|\bx_i\|_2\leq \rho_n\halfpower\}$,
\begin{align*}
-\frac{\partial^2 M_{in}}{\partial\bx_i\partial\bx_i\transpose}(\bx_i)
&= \frac{1}{n}\sum_{j=1}^n \left\{
-p_{0ij}\psi_n''(\rho_n^{1/2}\bx_i\transpose\eye_{p, q}\bx_{0j}) - (1 - p_{0ij})\psi_n''(1 - \rho_n^{1/2}\bx_i\transpose\eye_{p, q}\bx_{0j})
% \frac{p_{0ij}}{\max(\rho_n\halfpower\bx_i\transpose\eye_{p,q}\bx_{0j},\,\tau_n)^2} +  \frac{1-p_{0ij}}{\max(1-\rho_n\halfpower\bx_i\transpose\eye_{p,q}\bx_{0j},\,\tau_n)^2}
\right\}\\ 
&\quad\times\rho_n\eye_{p,q}\bx_{0j}\bx_{0j}\transpose\eye_{p,q} \\
% &\succeq \frac{1}{n}\sum_{j=1}^n -p_{0ij}\psi_n''(\rho_n^{1/2}\bx_i\transpose\eye_{p,q}\bx_{0j})
% \rho_n\eye_{p,q}\bx_{0j}\bx_{0j}\transpose\eye_{p,q} \\
&\succeq \frac{1}{n}\sum_{j=1}^n \frac{\delta\rho_n}{(1-\delta)\rho_n^2}\rho_n\eye_{p,q}\bx_{0j}\bx_{0j}\transpose\eye_{p,q} 
% \\&
\succeq \frac{\delta}{1-\delta}\lambda_d\left(\frac{1}{n}\sum_{j=1}^n \eye_{p,q}\bx_{0j}\bx_{0j}\transpose\eye_{p,q}\right)\eye_d 
\\
&
= \frac{\delta\lambda}{1-\delta}\eye_d.
\end{align*}

\noindent{}It is easy to see that $\widehat\ell_{in}(\bx_i)$ is continuous over $\mathbb{R}^d$, so there exists a maximizer of $\widehat\ell_{in}(\bx_i)$ in $\{\bx_i:\|\bx_i\|_2\leq\rho_n\halfpower\}$ which is a compact set. Let $\widehat\bx_i$ denote this maximizer of $\widehat\ell_{in}(\bx_i)$ in $\{\bx_i:\|\bx_i\|_2\leq\rho_n\halfpower\}$, and over this set, by Taylor's theorem,
\begin{align*}
&M_{in}(\rho_n\halfpower\bx_{0i}) - M_{in}(\bW\transpose\widehat\bx_i)
 % \\
% &\quad = -\frac{\partial M_{in}}{\partial\bx_i}(\rho_n\halfpower\bx_{0i})(\bW\transpose\widehat\bx_i - \rho_n\halfpower\bx_{0i}) 
% \\&\quad\quad
 % - \frac{1}{2}(\bW\transpose\widehat\bx_i - \rho_n\halfpower\bx_{0i})\transpose\frac{\partial^2 M_{in}}{\partial\bx_i\partial\bx_i\transpose}(\bar\bx_i)(\bW\transpose\widehat\bx_i - \rho_n\halfpower\bx_{0i}) 
 % \\&
% \quad 
\geq \frac{\delta\lambda}{1-\delta}\left\Vert \bW\transpose\widehat\bx_i - \rho_n\halfpower\bx_{0i} \right\Vert_2^2,
\end{align*}
and also by Taylor's theorem and by Lemma \ref{lemma:concentration-gradient}
\begin{align*}
& M_{in}(\rho_n\halfpower\bx_{0i}) - M_{in}(\bW\transpose\widehat\bx_i) \\
% &\quad \leq 
% \frac{1}{n}\widehat\ell(\widehat\bx_i) - M_{in}(\bW\transpose\widehat\bx_i) - \left\{ \frac{1}{n}\widehat\ell(\rho_n\halfpower\bW\bx_{0i}) - M_{in}(\rho_n\halfpower\bx_{0i}) \right\} \\
% &\quad = 
% \left\{ \frac{1}{n}\bW\transpose\frac{\partial\widehat\ell}{\partial\bx_i}(\bW\bar\bx_i) - \frac{\partial M_{in}}{\partial\bx_i}(\bar\bx_i) \right\}\left( \bW\transpose\widehat\bx_i - \rho_n\halfpower\bx_{0i} \right) \\
&\quad \leq \left\Vert \frac{1}{n}\bW\transpose\frac{\partial\widehat\ell_{in}}{\partial\bx_i}(\bW\bar\bx_i) - \frac{\partial M_{in}}{\partial\bx_i}(\bar\bx_i) \right\Vert_2 \cdot \left\Vert \bW\transpose\widehat\bx_i - \rho_n\halfpower\bx_{0i} \right\Vert_2
%  \\
% &\quad \leq \max_{i\in[n]}\sup_{\|\bx_i\|_2\leq \rho_n\halfpower} \left\Vert \frac{1}{n}\bW\transpose\frac{\partial\widehat\ell_{in}}{\partial\bx_i}(\bW\bx_i) - \frac{\partial M_{in}}{\partial\bx_i}(\bx_i) \right\Vert_2 \cdot 2\rho_n\halfpower 
% \\&
\lesssim_{c,\delta,\lambda} \rho_n\halfpower\sqrt{\frac{\log n}{n}}
\end{align*}
for all $i\in[n]$ with probability at least $1-n^{-c}$. Combining the two inequalities above, we have
\[
\max_{i\in[n]} \left\Vert \bW\transpose\widehat\bx_i - \rho_n\halfpower\bx_{0i} \right\Vert_2
\lesssim_{c,\delta,\lambda} (\rho_n)^{1/4}\left(\frac{\log n}{n}\right)^{1/4}
\]
with probability at least $1-n^{-c}$. Then by taking $N_{c,\delta,\lambda}$ large enough such that\\ 
$C_{c,\delta,\lambda}((\log n)/(n\rho_n))^{1/4}<1-\delta/2-\sqrt{1-\delta}$ for all $n\geq N_{c,\delta,\lambda}$, we have
\begin{align*}
\max_{i\in[n]} \|\widehat\bx_i\|_2 
&\leq \max_{i\in[n]} \left\Vert \bW\transpose\widehat\bx_i - \rho_n\halfpower\bx_{0i} \right\Vert_2 + \max_{i\in[n]} \|\rho_n\halfpower\bx_{0i}\|_2  
% \\&
< (1-\delta/2)\rho_n\halfpower,
\end{align*}
which implies that $\widehat\bx_i$ lies in the interior of the ball $\{\bx_i:\|\bx_i\|_2\leq\rho_n\halfpower\}$, i.e., $\widehat\bx_i$ is a local maximizer, with probability at least $1-n^{-c}$. Since $\widehat\ell_{in}(\bx_i)$ is strongly concave over $\mathbb{R}^d$ with probability at least $1-n^{-c}$, the local maximizer $\widehat\bx_i$ is the unique global maximizer with probability at least $1-n^{-c}$. This implies that, for all $i\in[n]$ with probability at least $1-n^{-c}$, $\|\widehat{\bx}_i\|_2 < (1-\delta/2)\rho_n^{1/2}$.
% By taking a union bound over $i\in[n]$, this holds for all $i\in[n]$ with probability at least $1-n^{-c}$. The existence and uniqueness of the maximizer of the extended surrogate log-likelihood function, $\widehat\bx_i$, for all $i\in[n]$, is thus proved.

\noindent\textit{Proof of Consistency.} Over this event that happens with probability at least $1-n^{-c}$, we have $\|\widehat\bx_i\|_2<(1-\delta/2)\rho_n\halfpower$, and $\widehat\bx_i$ is the unique global maximizer and thus $(\partial\widehat\ell_{in})/(\partial\bx_i)(\widehat\bx_i)=\zero_d$.
\noindent{}Now, by mean value theorem,
\[
\frac{\partial M_{in}}{\partial\bx_i}(\bW\transpose\widehat\bx_i)
= \frac{\partial^2 M_{in}}{\partial\bx_i\partial\bx_i\transpose}(\bar\bx_i) (\bW\transpose\widehat\bx_i-\rho_n\halfpower\bx_{0i})
\]
where $\bar\bx_i = \theta\bW\transpose\widehat\bx_i + (1-\theta)\rho_n\halfpower\bx_{0i}$ for some $\theta\in(0,1)$, then it follows that
\begin{align*}
\left\| \frac{\partial M_{in}}{\partial\bx_i}(\bW\transpose\widehat\bx_i) \right\|_2
&= \left\| \frac{\partial^2 M_{in}}{\partial\bx_i\partial\bx_i\transpose}(\bar\bx_i) (\bW\transpose\widehat\bx_i-\rho_n\halfpower\bx_{0i}) \right\|_2 
% \\&
\geq \frac{\delta\lambda}{1-\delta} \left\| \bW\transpose\widehat\bx_i-\rho_n\halfpower\bx_{0i} \right\|_2,
\end{align*}
and by Lemma \ref{lemma:concentration-gradient},
\begin{align*}
&\left\| \frac{\partial M_{in}}{\partial\bx_i}(\bW\transpose\widehat\bx_i) \right\|_2\\
% &\quad = \left\{\left\| \frac{\partial M_{in}}{\partial\bx_i}(\bW\transpose\widehat\bx_i) \right\|_2 - \left\| \frac{\partial M_{in}}{\partial\bx_i}(\rho_n\halfpower\bx_{0i}) \right\|_2\right\} \\
&\quad\leq \left\{\left\| \frac{\partial M_{in}}{\partial\bx_i}(\bW\transpose\widehat\bx_i) \right\|_2 - \left\| \frac{\partial M_{in}}{\partial\bx_i}(\rho_n\halfpower\bx_{0i}) \right\|_2\right. 
% \\&\quad\quad\quad
 + \left.\left\| \frac{\partial \widehat\ell_{in}}{\partial\bx_i}(\rho_n\halfpower\bW\bx_{0i}) \right\|_2 - \left\| \frac{\partial \widehat\ell_{in}}{\partial\bx_i}(\widehat\bx_i) \right\|_2\right\} \\
% &\quad = \left\{\left\| \frac{\partial M_{in}}{\partial\bx_i}(\bW\transpose\widehat\bx_i) \right\|_2 - \left\| \frac{\partial \widehat\ell_{in}}{\partial\bx_i}(\widehat\bx_i) \right\|_2\right. 
% \\&\quad\quad\quad
 % + \left.\left\| \frac{\partial \widehat\ell_{in}}{\partial\bx_i}(\rho_n\halfpower\bW\bx_{0i}) \right\|_2 - \left\| \frac{\partial M_{in}}{\partial\bx_i}(\rho_n\halfpower\bx_{0i}) \right\|_2\right\} \\
&\quad \leq 2 \max_{i\in[n]}\sup_{\|\bx_i\|_2\leq \rho_n\halfpower} \left\Vert \frac{1}{n}\bW\transpose\frac{\partial\widehat\ell_{in}}{\partial\bx_i}(\bW\bx_i) - \frac{\partial M_{in}}{\partial\bx_i}(\bx_i) \right\Vert_2 
% \\&
\lesssim_{c,\delta,\lambda} \sqrt{\frac{\log n}{n}}
\end{align*}
with probability at least $1-n^{-c}$. Combining the two inequalities above, we have
\[
\max_{i\in[n]}\left\| \bW\transpose\widehat\bx_i-\rho_n\halfpower\bx_{0i} \right\|_2 \lesssim_{c,\delta,\lambda} \sqrt{\frac{\log n}{n}}
\]
with probability at least $1-n^{-c}$. The consistency is thus proved.

\noindent
\textit{Proof of Asymptotic normality.}
Let $\breve\bx_i$ denote the adjacency spectral embedding, $i\in[n]$, and let $\widetilde{p}_{ij}=\breve\bx_i\transpose\widetilde\bx_j$, $p_{0ij}=\rho_n\bx_{0i}\transpose\eye_{p,q}\bx_{0j}$, $i,j\in[n]$. Define the 
% Fisher information matrix
% \[
% \bG_{0in} = \frac{1}{n}\sum_{j=1}^n \frac{\rho_n\eye_{p,q}\bx_{0j}\bx_{0j}\transpose\eye_{p,q}}{p_{0ij}(1-p_{0ij})},
% \]
and its plug-in estimate of $\bG_{0in}$ with the adjacency spectral embedding by
% \[
$\widetilde\bG_{in} = (1/n)\sum_{j=1}^n {\widetilde\bx_j\widetilde\bx_j\transpose}/\{\widetilde{p}_{ij}(1-\widetilde{p}_{ij})\}$.
% \]
It is not hard to see that
\begin{equation}
\label{eqn:fisher-information-order}
\begin{split}
\lambda_1(\bG_{0in}) &\leq \frac{1}{\delta(1-\delta)}\lambda_1(\frac{1}{n}\bX_0\transpose\bX_0) \leq \frac{1}{n\delta(1-\delta)}\|\bX_0\|\frobenius^2 \leq \frac{1}{\delta}, \\
\lambda_d(\bG_{0in}) &\geq \frac{1}{1-\delta} \lambda_d(\frac{1}{n}\bX_0\transpose\bX_0) \geq \frac{\lambda}{1-\delta},
\end{split}
\end{equation}
i.e., $\bG_{0in}$ is a positive definite matrix with eigenvalues bounded away from $0$ and $\infty$, and by Lemma \ref{lemma:frequently-used-results}, Theorem 5.2 in \cite{lei-rinaldo-2015-sbm}, and Weyl's inequality, we also have
$\lambda_1(\widetilde\bG_{in})
\leq 2/\delta$ and $\lambda_d(\widetilde\bG_{in})\geq \lambda/\{2(1 - \delta/2)\}$,
% \begin{align*}
% \lambda_1(\widetilde\bG_{in})
% &\leq \frac{1}{(\delta/2)(1-\delta/2)}\lambda_1(\frac{1}{n\rho_n}\widetilde\bX\transpose\widetilde\bX) 
% % \\&
% = \frac{1}{(\delta/2)(1-\delta/2)}\sigma_1(\frac{1}{n\rho_n}\bA) \\
% &\leq \frac{1}{(\delta/2)(1-\delta/2)}\left(\sigma_1(\frac{1}{n\rho_n}\bP)+\sigma_1(\frac{1}{n\rho_n}(\bA-\bP))\right) 
% % \\&
% \leq \frac{2}{\delta}, \\
% \lambda_d(\widetilde\bG_{in})
% &\geq \frac{1}{1-\delta/2} \lambda_d(\frac{1}{n\rho_n}\widetilde\bX\transpose\widetilde\bX) 
% % \\&
% = \frac{1}{1-\delta/2} \left|\sigma_d(\frac{1}{n\rho_n}\bA)\right| \\
% &\geq \frac{1}{1-\delta/2} \left(\left|\sigma_d(\frac{1}{n\rho_n}\bP)\right|-\left|\sigma_d(\frac{1}{n\rho_n}(\bA-\bP))\right|\right) 
% % \\&
% \geq \frac{\lambda/2}{1-\delta/2} ,
% \end{align*}
i.e., $\widetilde\bG_{in}$ is a positive definite matrix with eigenvalues bounded away from $0$ and $\infty$ with probability at least $1-n^{-c}$. Now we have
\begin{align*}
\left\Vert\bW\transpose\widetilde\bG_{in}\bW - \bG_{0in}\right\Vert_2
% &= \left\Vert \frac{1}{n}\sum_{j=1}^n \frac{\bW\transpose\widetilde\bx_j\widetilde\bx_j\transpose\bW}{\widetilde{p}_{ij}(1-\widetilde{p}_{ij})} - \frac{1}{n}\sum_{j=1}^n \frac{\rho_n\eye_{p,q}\bx_{0j}\bx_{0j}\transpose\eye_{p,q}}{p_{0ij}(1-p_{0ij})} \right\Vert_2 \\
&\leq \left\Vert \frac{1}{n}\sum_{j=1}^n \left\{\frac{1}{\widetilde{p}_{ij}(1-\widetilde{p}_{ij})} - \frac{1}{p_{0ij}(1-p_{0ij})}\right\}\bW\transpose\widetilde\bx_j\widetilde\bx_j\transpose\bW \right\Vert_2 \\
&\quad + \left\Vert \frac{1}{n}\sum_{j=1}^n \frac{1}{p_{0ij}(1-p_{0ij})}\left(\bW\transpose\widetilde\bx_j\widetilde\bx_j\transpose\bW - \rho_n\eye_{p,q}\bx_{0j}\bx_{0j}\transpose\eye_{p,q}\right)\right\Vert_2 \\
&\leq \max_{j\in[n]}\left\vert\frac{(1-2\bar{p}_{ij})(\widetilde{p}_{ij} - p_{0ij})}{\bar{p}_{ij}^2(1-\bar{p}_{ij})^2}\right\vert \cdot \max_{j\in[n]}\|\widetilde\bx_j\|_2^2  \\
&\quad + \max_{j\in[n]}\left\vert\frac{1}{p_{0ij}(1-p_{0ij})}\right\vert \cdot \max_{j\in[n]} \left\Vert \bW\transpose\widetilde\bx_j\widetilde\bx_j\transpose\bW - \rho_n\eye_{p,q}\bx_{0j}\bx_{0j}\transpose\eye_{p,q}\right\Vert_2 \\
&\lesssim_{c,\delta,\lambda} \frac{1}{\rho_n^2} \cdot \rho_n\halfpower\sqrt{\frac{\log n}{n}} \cdot \rho_n + \frac{1}{\rho_n} \cdot \rho_n\halfpower\sqrt{\frac{\log n}{n}} 
% \\&
\asymp_{c,\delta,\lambda} \sqrt{\frac{\log n}{n\rho_n}}
\end{align*}
with probability at least $1-n^{-c}$, where the inequalities follows from triangle inequality, mean value theorem, Cauchy-Schwarz inequality, and Lemma \ref{lemma:frequently-used-results}. And we have
\[
\left\Vert\frac{1}{n}\bW\transpose\frac{\partial^2 \widehat\ell_{in}}{\partial\bx_i\partial\bx_i\transpose}(\breve\bx_i)\bW + \bG_{0in}\right\Vert_2
\lesssim_{c,\delta,\lambda} \sqrt{\frac{\log n}{n\rho_n}}
\]
with probability at least $1-n^{-c}$, by Theorem \ref{thm:signed-ase} (bound for $\breve\bx_i$) and Lemma \ref{lemma:concentration-hessian}. So, by triangle inequality and the previous two bounds, we have
\[
\left\Vert\frac{1}{n}\frac{\partial^2 \widehat\ell_{in}}{\partial\bx_i\partial\bx_i\transpose}(\breve\bx_i) + \widetilde\bG_{in}\right\Vert_2
\lesssim_{c,\delta,\lambda} \sqrt{\frac{\log n}{n\rho_n}}
\]
with probability at least $1-n^{-c}$.

\noindent{}Now, we show the asymptotic normality of the maximizer of the ESL function by showing that $\widehat\bx_i$ and $\widehat\bx_i^{\mathrm{(OS)}}$ are close enough and then applying Slutsky's theorem to utilize the asymptotic normality of the one-step estimator $\widehat\bx_i^{\mathrm{(OS)}}$. For each $k\in[d]$, apply Taylor's theorem to $(1/n)(\partial\ell_{in})/(\partial\bx_i)(\widehat\bx_i)=0$ at $\bx_i=\breve\bx_i$ to obtain
\begin{align*}
0 &= \frac{1}{n}\frac{\partial\widehat\ell_{in}}{\partial x_{ik}}(\breve\bx_i) + \frac{1}{n}\frac{\partial^2\widehat\ell_{in}}{\partial x_{ik}\partial\bx_i\transpose}(\breve\bx_i) (\widehat\bx_i - \breve\bx_i)
%  \\
% &\quad\quad
 +\frac{1}{2}(\widehat\bx_i - \breve\bx_i)\transpose \frac{1}{n}\frac{\partial^3\widehat\ell_{in}}{\partial x_{ik}\partial\bx_i\partial\bx_i\transpose}(\bar\bx_i) (\widehat\bx_i - \breve\bx_i),
\end{align*}
where $\bar\bx_i=\theta\widehat\bx_i + (1-\theta)\breve\bx_i$ for some $\theta\in(0,1)$. By triangle inequality, Theorem \ref{thm:signed-ase}, and the consistency result that has been shown above, $\|\bW\bar\bx_i - \rho_n\halfpower\bx_{0i}\|_2\lesssim_{c,\delta,\lambda}\sqrt{(\log n)/n}$ with probability at lease $1-n^{c}$.
It is easy to see that, over such an event,
\[
\frac{1}{n}\frac{\partial^3\widehat\ell_{in}}{\partial x_{ik}\partial\bx_i\partial\bx_i\transpose}(\bar{\bx}_i)
= \frac{2}{n} \sum_{j=1}^n \left\{ \frac{A_{ij}}{(\bar{\bx}_i\transpose\widetilde\bx_j)^3} - \frac{1-A_{ij}}{(1-\bar{\bx}_i\transpose\widetilde\bx_j)^3} \right\} \widetilde{x}_{ik}\widetilde\bx_j\widetilde\bx_j\transpose
% \one(\tau_n\leq\bx_i\widetilde\bx_j\leq 1-\tau_n).
\]
Then, by Lemma \ref{lemma:frequently-used-results}, Lemma \ref{lemma:results-with-Aij}, and Cauchy-Schwarz inequality,
\begin{align*}
\left\Vert\frac{1}{n}\frac{\partial^3\widehat\ell_{in}}{\partial x_{ik}\partial\bx_i\partial\bx_i\transpose}(\bar\bx_i)\right\Vert_2
% &= \left\Vert \frac{2}{n} \sum_{j=1}^n \left\{ \frac{A_{ij}}{(\bar\bx_i\transpose\widetilde\bx_j)^3} - \frac{1-A_{ij}}{(1-\bar\bx_i\transpose\widetilde\bx_j)^3} \right\} \widetilde{x}_{ik}\widetilde\bx_j\widetilde\bx_j\transpose \right\Vert_2 \\
&\leq \left\{\frac{2}{n}\|\bA\|_\infty \cdot \max_{j\in[n]}\frac{1}{(\bar\bx_i\transpose\widetilde\bx_j)^3} + 2\max_{j\in[n]}\frac{1}{(1-\bar\bx_i\transpose\widetilde\bx_j)^3}\right\} \cdot \max_{j\in[n]}\|\widetilde\bx_j\|_2^3 \\
&\lesssim_{c,\delta,\lambda} \left\{\frac{1}{n}n\rho_n \cdot \frac{1}{\rho_n^3} + 1\right\} \cdot \rho_n^{3/2} 
% \\&
\asymp_{c,\delta,\lambda} \rho_n\invhalfpower,
\end{align*}
with probability at lease $1-n^{c}$, and also by triangle inequality, Theorem \ref{thm:signed-ase}, and the consistency result that has been shown above,
\begin{align*}
\|\widehat\bx_i - \breve\bx_i\|_2
&\leq \left\Vert\bW\widehat\bx_i-\rho_n\halfpower\bx_{0i}\right\Vert_2 + \left\Vert\bW\breve\bx_i-\rho_n\halfpower\bx_{0i}\right\Vert_2 
% \\&
\lesssim_{c,\delta,\lambda} \sqrt{\frac{\log n}{n}},
\end{align*}
with probability at lease $1-n^{c}$.
So the Taylor expansion mentioned above can be written as
\[
\left(-\frac{1}{n}\frac{\partial^2\widehat\ell_{in}}{\partial\bx_i\partial\bx_i\transpose}(\breve\bx_i) + \bR_{in1} \right)(\widehat\bx_i - \breve\bx_i) = \frac{1}{n}\frac{\partial\widehat\ell_{in}}{\partial\bx_i}(\breve\bx_i),
\]
where $\bR_{in1}$ is a random matrix with $\|\bR_{in1}\|_2\lesssim_{c,\delta,\lambda}\sqrt{(\log n)/(n\rho_n)}$ with probability at lease $1-n^{c}$. By the approximation of $(1/n)(\partial^2\widehat\ell_{in})/(\partial\bx_i\partial\bx_i\transpose)(\breve\bx_i)$ and $\widetilde\bG_{in}$ that has been shown above, Lemma \ref{lemma:frequently-used-results} (for the large probability bounds on $\widetilde{p}_{ij}$), and the definition of $(1/n)(\partial\widehat\ell_{in})/(\partial\bx_i)(\breve\bx_i)$, we have
\[
\left(\widetilde\bG_{in} + \bR_{in2}\right)(\widehat\bx_i - \breve\bx_i)
= \frac{1}{n}\sum_{j=1}^n \frac{A_{ij}-\widetilde{p}_{ij}}{\widetilde{p}_{ij}(1-\widetilde{p}_{ij})}\widetilde\bx_j,
\]
where $\bR_{in2}$ is a random matrix with $\|\bR_{in2}\|_2\lesssim_{c,\delta,\lambda}\sqrt{(\log n)/(n\rho_n)}$ with probability at lease $1-n^{c}$. Now we have $\|\widetilde\bG_{in}\inverse\bR_{in2}\|_2\lesssim_{c,\delta,\lambda}\sqrt{(\log n)/(n\rho_n)}$ with probability at lease $1-n^{c}$.
Write
\begin{align*}
\widehat\bx_i - \breve\bx_i
% &= \left(\widetilde\bG_{in} + \bR_{in2}\right)\inverse \frac{1}{n}\sum_{j=1}^n \frac{A_{ij}-\widetilde{p}_{ij}}{\widetilde{p}_{ij}(1-\widetilde{p}_{ij})}\widetilde\bx_j  \\
&= \left(\eye_d + \widetilde\bG_{in}\inverse\bR_{in2}\right)\inverse \widetilde\bG_{in}\inverse\frac{1}{n}\sum_{j=1}^n \frac{A_{ij}-\widetilde{p}_{ij}}{\widetilde{p}_{ij}(1-\widetilde{p}_{ij})}\widetilde\bx_j  \\
&= \sum_{m=0}^\infty\left(-\widetilde\bG_{in}\inverse\bR_{in2}\right)^m (\widehat\bx_i^{\mathrm{(OS)}} - \breve\bx_i)  
% \\&
= \widehat\bx_i^{\mathrm{(OS)}} - \breve\bx_i + \sum_{m=1}^\infty\left(-\widetilde\bG_{in}\inverse\bR_{in2}\right)^m (\widehat\bx_i^{\mathrm{(OS)}} - \breve\bx_i),  \\
\end{align*}
then write
\begin{align*}
\|\widehat\bx_i - \widehat\bx_i^{\mathrm{(OS)}}\|_2
% &= \left\Vert \sum_{m=1}^\infty\left(-\widetilde\bG_{in}\inverse\bR_{in2}\right)^m (\widehat\bx_i^{\mathrm{(OS)}} - \breve\bx_i) \right\Vert_2 
% \\&
\leq 
% \sum_{m=1}^\infty \|\widetilde\bG_{in}\inverse\|_2^m \|\bR_{in2}\|_2^m \|\widehat\bx_i^{\mathrm{(OS)}} - \breve\bx_i\|_2 \\
% &= 
\frac{\|\widetilde\bG_{in}\inverse\|_2 \|\bR_{in2}\|_2}{1 - \|\widetilde\bG_{in}\inverse\|_2 \|\bR_{in2}\|_2} \|\widehat\bx_i^{\mathrm{(OS)}} - \breve\bx_i\|_2,
\end{align*}
and under the assumption that $(\log n)/(n\rho_n)\to 0$, by Theorem \ref{thm:signed-ase} and Theorem \ref{thm:ose-asymptotic-normality} (take $t=C_{c,\delta,\lambda}\log(n\rho_n)$), we have
\begin{align*}
\|\widehat\bx_i^{\mathrm{(OS)}} - \breve\bx_i\|_2
&\leq \|\bW\widehat\bx_i^{\mathrm{(OS)}} - \rho_n\halfpower\bx_{0i}\|_2 + \|\bW\breve\bx_i - \rho_n\halfpower\bx_{0i}\|_2 
% \\&
\lesssim_{c,\delta,\lambda} \sqrt{\frac{\log n}{n}}
\end{align*}
with probability at least $1-(n\rho_n)^{-c}$, so
$\|\widehat\bx_i - \widehat\bx_i^{\mathrm{(OS)}}\|_2 \lesssim_{c,\delta,\lambda}\log n/(n\rho_n\halfpower)$ with probability at least $1-(n\rho_n)^{-c}$. By Theorem \ref{thm:ose-asymptotic-normality} and Slutsky's theorem, we have
\[
\bG_{0in}\halfpower(\bW\transpose\widehat\bx_i-\rho_n\halfpower\bx_{0i})
= \frac{1}{n}\sum_{j=1}^n \frac{(A_{ij}-p_{0ij})\bG_{0in}\invhalfpower\rho_n\halfpower\eye_{p,q}\bx_{0j}}{p_{0ij}(1-p_{0ij})} + \br_{in},
\]
where
\[
\|\br_{in}\|_2 \lesssim_{c,\delta,\lambda} \frac{\log n}{n\rho_n\halfpower} + \frac{(\log(n\rho_n))^2}{n\rho_n\halfpower}
\]
with probability at least $1-(n\rho_n)^{-c}$, and
% \[
$\sqrt{n}\bG_{0in}\halfpower(\bW\transpose\widehat\bx_i-\rho_n\halfpower\bx_{0i}) \overset{\calL}{\to} \mathrm{N}_d(\zero_d,\,\eye_d)$.
% \]

\noindent
\textit{Proof of consistency of global error.}
Under the assumption that $(\log n)^4/(n\rho_n)\to 0$, by Theorem \ref{thm:signed-ase} and Theorem \ref{thm:ose-asymptotic-normality} (take $t=C_{c,\delta,\lambda}\log n$ in Theorem \ref{thm:ose-asymptotic-normality}), we have
\begin{align*}
\|\widehat\bx_i^{\mathrm{(OS)}} - \breve\bx_i\|_2
&\leq \|\bW\widehat\bx_i^{\mathrm{(OS)}} - \rho_n\halfpower\bx_{0i}\|_2 + \|\bW\breve\bx_i - \rho_n\halfpower\bx_{0i}\|_2 
% \\&
\lesssim_{c,\delta,\lambda} \sqrt{\frac{\log n}{n}}
\end{align*}
with probability at least $1-n^{-c}$, so
$\|\widehat\bx_i - \widehat\bx_i^{\mathrm{(OS)}}\|_2 \lesssim_{c,\delta,\lambda}{\log n}/{(n\rho_n\halfpower)}$
with probability at least $1-n^{-c}$, then by Theorem \ref{thm:ose-asymptotic-normality} and Slutsky's theorem, we have
\[
\bG_{0in}\halfpower(\bW\transpose\widehat\bx_i-\rho_n\halfpower\bx_{0i})
= \frac{1}{n}\sum_{j=1}^n \frac{(A_{ij}-p_{0ij})\bG_{0in}\invhalfpower\rho_n\halfpower\eye_{p,q}\bx_{0j}}{p_{0ij}(1-p_{0ij})} + \br_{in},
\]
where
$\|\br_{in}\|_2 \lesssim_{c,\delta,\lambda} {(\log n)^2}/{(n\rho_n\halfpower)}$
with probability at least $1-n^{-c}$.
Now write
\begin{align*}
\left\Vert\widehat\bX\bW - \rho_n\halfpower\bX_0\right\Vert\frobenius^2
% &= \sum_{i=1}^n\left\Vert \bW\transpose\widehat\bx_i-\rho_n\halfpower\bx_{0i} \right\Vert_2^2 \\
% &= \sum_{i=1}^n\left\Vert \frac{1}{n}\sum_{j=1}^n \frac{(A_{ij}-p_{0ij})\bG_{0in}\inverse\rho_n\halfpower\eye_{p,q}\bx_{0j}}{p_{0ij}(1-p_{0ij})} + \br_{in} \right\Vert_2^2  \\
&= \sum_{i=1}^n\left\Vert \frac{1}{n}\sum_{j=1}^n \frac{(A_{ij}-p_{0ij})\bG_{0in}\inverse\rho_n\halfpower\eye_{p,q}\bx_{0j}}{p_{0ij}(1-p_{0ij})} \right\Vert_2^2 + \sum_{i=1}^n\left\Vert\bG_{0in}\invhalfpower\br_{in}\right\Vert_2^2  \\
&\quad + 2\sum_{i=1}^n\left<\frac{1}{n}\sum_{j=1}^n \frac{(A_{ij}-p_{0ij})\bG_{0in}\inverse\rho_n\halfpower\eye_{p,q}\bx_{0j}}{p_{0ij}(1-p_{0ij})},\,\bG_{0in}\invhalfpower\br_{in}\right>,
\end{align*}
which can be viewed as a sum of three terms.
For the first term, by Lemma \ref{lemma:grdpg-inid-wlln},
\begin{align*}
\sum_{i=1}^n\left\Vert \frac{1}{n}\sum_{j=1}^n \frac{(A_{ij}-p_{0ij})\bG_{0in}\inverse\rho_n\halfpower\eye_{p,q}\bx_{0j}}{p_{0ij}(1-p_{0ij})} \right\Vert_2^2  
&= \frac{1}{n}\sum_{i=1}^n\mathrm{tr}(\bG_{0in}\inverse) + o_{\prob}(1).
\end{align*}
For the second term, recalling that $\bG_{0in}$ is a positive definite matrix with eigenvalues bounded away from $0$ and $\infty$, we have
\begin{align*}
\sum_{i=1}^n\left\Vert\bG_{0in}\invhalfpower\br_{in}\right\Vert_2^2
&\lesssim_{\delta,\lambda} n\max_{i\in[n]}\|\br_{in}\|_2^2  
% \\&
\lesssim_{c,\delta,\lambda} \frac{(\log n)^4}{n\rho_n},
\end{align*}
with probability at least $1-n^{-c}$.
For the third term, by triangle inequality and Cauchy-Schwarz inequality (twice),
\begin{align*}
&\left|\sum_{i=1}^n\left<\frac{1}{n}\sum_{j=1}^n \frac{(A_{ij}-p_{0ij})\bG_{0in}\inverse\rho_n\halfpower\eye_{p,q}\bx_{0j}}{p_{0ij}(1-p_{0ij})},\,\bG_{0in}\invhalfpower\br_{in}\right>\right| \\
% &\quad \leq \sum_{i=1}^n \left\|\frac{1}{n}\sum_{j=1}^n \frac{(A_{ij}-p_{0ij})\bG_{0in}\inverse\rho_n\halfpower\eye_{p,q}\bx_{0j}}{p_{0ij}(1-p_{0ij})}\right\|_2 \left\|\bG_{0in}\invhalfpower\br_{in}\right\|_2  \\
&\quad \leq \left(\sum_{i=1}^n \left\|\frac{1}{n}\sum_{j=1}^n \frac{(A_{ij}-p_{0ij})\bG_{0in}\inverse\rho_n\halfpower\eye_{p,q}\bx_{0j}}{p_{0ij}(1-p_{0ij})}\right\|_2^2\right)\halfpower \left(\sum_{i=1}^n \left\|\bG_{0in}\invhalfpower\br_{in}\right\|_2^2\right)\halfpower  
% \\
% &= O_{\prob}(1) \cdot o_{\prob}(1) \\
% &
= o_{\prob}(1).
\end{align*}
Hence, we conclude that
% \[
$\Vert\widehat\bX\bW - \rho_n\halfpower\bX_0\Vert\frobenius^2 = (1/n)\sum_{i=1}^n\mathrm{tr}(\bG_{0in}\inverse) + o_{\prob}(1)$.
% \]
\end{proof}

\subsection{Proof of Theorem \ref{thm:bernstein-von-mises}}

\begin{proof}
% [Proof of Theorem \ref{thm:bernstein-von-mises}]
It is sufficient to prove that
\begin{align*}
&\int_{\mathbb{R}^d} \left|\pi_i(\bx_i)\exp\left\{\widehat\ell_{in}(\bx_i)-\widehat\ell_{in}(\widehat\bx_i)\right\}\right.\\ 
&\quad\quad\left. - \pi_i(\rho_n\halfpower\bW\bx_{0i})\exp\left\{-\frac{n}{2}(\bx_i-\widehat\bx_i)\transpose\bW\bG_{0in}\bW\transpose(\bx_i-\widehat\bx_i)\right\}\right|\diff\bx_i  
% \\\&\quad 
\lesssim_{c,\delta,\lambda} n^{-d/2}\frac{1}{\log n}
\end{align*}
with probability at least $1-n^{-c}$.
Recall that the posterior density associated with the ESL function for vertex $i$ is
% \[
$\pi_{in}(\bx_i\mid\bA) = e^{\widehat\ell_{in}(\bx_i)-\widehat\ell_{in}(\widehat\bx_i)}\pi_i(\bx_i)/d_{in}$,
% \]
where $d_{in}=\int_{\mathbb{R}^d}\exp\{\widehat\ell_{in}(\bx_i)-\widehat\ell_{in}(\widehat\bx_i)\}\pi_i(\bx_i)\diff\bx_i$ is the normalizing constant.
Let $\eta_n = C_{c,\delta,\lambda}\sqrt{\log n}$, and partition $\mathbb{R}^d$ as $\calA_{1in}\cup\calA_{2in}$ where
\[
\calA_{1in} = \left\{\bx_i\in\mathbb{R}^d: \sqrt{n}\|\bx_i-\widehat\bx_i\|_2 \leq \eta_n\right\},
\quad \calA_{2in} = \left\{\bx_i\in\mathbb{R}^d: \sqrt{n}\|\bx_i-\widehat\bx_i\|_2 > \eta_n\right\}.
\]
For the integral over $\calA_{2in}$,
\begin{align*}
&\int_{\calA_{2in}} \left|\pi_i(\bx_i)\exp\left\{\widehat\ell_{in}(\bx_i)-\widehat\ell_{in}(\widehat\bx_i)\right\}\right.\\
&\qquad\left. - \pi_i(\rho_n\halfpower\bW\bx_{0i})\exp\left\{-\frac{n}{2}(\bx_i-\widehat\bx_i)\transpose\bW\bG_{0in}\bW\transpose(\bx_i-\widehat\bx_i)\right\}\right|\diff\bx_i  \\
&\quad \leq \int_{\calA_{2in}}\pi_i(\bx_i)\exp\left\{\widehat\ell_{in}(\bx_i)-\widehat\ell_{in}(\widehat\bx_i)\right\}\diff\bx_i\\ 
&\quad\quad  + \int_{\calA_{2in}}\pi_i(\rho_n\halfpower\bW\bx_{0i})\exp\left\{-\frac{n}{2}(\bx_i-\widehat\bx_i)\transpose\bW\bG_{0in}\bW\transpose(\bx_i-\widehat\bx_i)\right\}\diff\bx_i,
\end{align*}
of which for the first term, by \eqref{eqn:hessian-strongly-convex-bound} on page \pageref{eqn:hessian-strongly-convex-bound} and the assumption that $\rho_n=1$, we have
\begin{align*}
&\int_{\calA_{2in}}\pi_i(\bx_i)\exp\left\{\widehat\ell_{in}(\bx_i)-\widehat\ell_{in}(\widehat\bx_i)\right\}\diff\bx_i  \\
&\quad \leq C\int_{\calA_{2in}}\exp\left\{\widehat\ell_{in}(\bx_i)-\widehat\ell_{in}(\widehat\bx_i)\right\}\diff\bx_i  
% \\&
\leq C\int_{\calA_{2in}}\exp\left\{-\frac{n\lambda}{4}\|\bx_i-\widehat\bx_i\|_2^2\right\}\diff\bx_i  \\
&\quad = Cn^{-d/2}\int_{\|\bt_i\|_2\geq\eta_n}\exp\left\{-\frac{\lambda}{4}\|\bt_i\|_2^2\right\}\diff\bt_i 
% \\&
\lesssim_{\lambda} n^{-d/2}\eta_n^{d-2}\exp\{-(\lambda/4)\eta_n^2\} 
% \\&
\lesssim_{\lambda} n^{-d/2}\eta_n^{-2}
\end{align*}
with probability at least $1-n^{-c}$, and for the second term, by \eqref{eqn:fisher-information-order} on page \pageref{eqn:fisher-information-order}, we have
\begin{align*}
&\int_{\calA_{2in}}\pi_i(\rho_n\halfpower\bW\bx_{0i})\exp\left\{-\frac{n}{2}(\bx_i-\widehat\bx_i)\transpose\bW\bG_{0in}\bW\transpose(\bx_i-\widehat\bx_i)\right\}\diff\bx_i  \\
&\quad \leq C\int_{\calA_{2in}}\exp\left\{-\frac{\lambda n}{2}\|\bx_i-\widehat\bx_i\|_2^2\right\}\diff\bx_i  
% \\&
= Cn^{-d/2}\int_{\|\bt_i\|_2\geq\eta_n}\exp\left\{-\frac{\lambda}{2}\|\bt_i\|_2^2\right\}\diff\bt_i \\
&\quad \lesssim_{\lambda} n^{-d/2}\eta_n^{d-2}\exp\{-(\lambda/2)\eta_n^2\} 
% \\&
\lesssim_{\lambda} n^{-d/2}\eta_n^{-2}
\end{align*}
with probability at least $1-n^{-c}$.
Over $\calA_{1in}$,
\begin{align*}
\left\|\bW\transpose\bx_i-\rho_n\halfpower\bx_{0i}\right\|_2
&\leq \left\|\bW\transpose\widehat\bx_i-\rho_n\halfpower\bx_{0i}\right\|_2 + \|\bx_i-\widehat\bx_i\|_2  
% \\&
\lesssim_{c,\delta,\lambda} \sqrt{\frac{\log n}{n}}
\end{align*}
with probability at least $1-n^{-c}$ by Theorem \ref{thm:mesle} and the definition of $\eta_n$.
For the integral over $\calA_{1in}$, with the change of variable $\bt_i=\sqrt{n}\bW\transpose(\bx_i-\widehat\bx_i)$,
\begin{align*}
&\int_{\calA_{1in}} \left|\pi_i(\bx_i)\exp\left\{\widehat\ell_{in}(\bx_i)-\widehat\ell_{in}(\widehat\bx_i)\right\}\right.\\ 
&\quad \quad \quad \left. - \pi_i(\rho_n\halfpower\bW\bx_{0i})\exp\left\{-\frac{n}{2}(\bx_i-\widehat\bx_i)\transpose\bW\bG_{0in}\bW\transpose(\bx_i-\widehat\bx_i)\right\}\right|\diff\bx_i  \\
&= \int_{\calA_{1in}} \left|\pi_i(\bx_i)\exp\left\{\frac{1}{2}(\bx_i-\widehat\bx_i)\transpose\frac{\partial^2\widehat\ell_{in}}{\partial\bx_i\partial\bx_i\transpose}(\bar\bx_i)(\bx_i-\widehat\bx_i)\right\}\right.\\ 
&\quad \quad \quad \left. - \pi_i(\rho_n\halfpower\bW\bx_{0i})\exp\left\{-\frac{n}{2}(\bx_i-\widehat\bx_i)\transpose\bW\bG_{0in}\bW\transpose(\bx_i-\widehat\bx_i)\right\}\right|\diff\bx_i  \\
&= n^{-d/2}\int_{\|\bt_i\|_2\leq\eta_n} \left|\pi_i(\widehat\bx_i+\frac{\bW\bt_i}{\sqrt{n}})\exp\left\{\frac{1}{2}\bt_i\transpose\bW\transpose\frac{1}{n}\frac{\partial^2\widehat\ell_{in}}{\partial\bx_i\partial\bx_i\transpose}(\bar\bx_i)\bW\bt_i\right\}\right.\\ 
&\quad \quad \quad \quad \quad \quad \quad \quad \left. - \pi_i(\rho_n\halfpower\bW\bx_{0i})\exp\left\{-\frac{1}{2}\bt_i\transpose\bG_{0in}\bt_i\right\}\right|\diff\bt_i  \\
&= n^{-d/2}\int_{\|\bt_i\|_2\leq\eta_n} \left|\pi_i(\widehat\bx_i+\frac{\bW\bt_i}{\sqrt{n}})\exp\left\{\frac{1}{2}\bt_i\transpose\bW\transpose\frac{1}{n}\frac{\partial^2\widehat\ell_{in}}{\partial\bx_i\partial\bx_i\transpose}(\bar\bx_i)\bW\bt_i\right\}\right.\\ 
&\quad \quad \quad \quad \quad \quad \quad \quad \left. - \pi_i(\widehat\bx_i+\frac{\bW\bt_i}{\sqrt{n}})\exp\left\{-\frac{1}{2}\bt_i\transpose\bG_{0in}\bt_i\right\}\right|\diff\bt_i  \\
&\quad + n^{-d/2}\int_{\|\bt_i\|_2\leq\eta_n} \left|\pi_i(\widehat\bx_i+\frac{\bW\bt_i}{\sqrt{n}})\exp\left\{-\frac{1}{2}\bt_i\transpose\bG_{0in}\bt_i\right\}
% \right.
% \\ 
% &\quad \quad \quad \quad \quad \quad \quad \quad \quad \left. 
- \pi_i(\rho_n\halfpower\bW\bx_{0i})e^{-\bt_i\transpose\bG_{0in}\bt_i/2}\right|\diff\bt_i,  \\
\end{align*}
which is a sum of two terms, in which $\bar\bx_i$ is a convex combination of $\bx_i$ and $\widehat\bx_i$. For the first term of this integral over $\calA_{1in}$ we have
\begin{align*}
&\int_{\|\bt_i\|_2\leq\eta_n} \left|\pi_i(\widehat\bx_i+\frac{\bW\bt_i}{\sqrt{n}})\exp\left\{\frac{1}{2}\bt_i\transpose\bW\transpose\frac{1}{n}\frac{\partial^2\widehat\ell_{in}}{\partial\bx_i\partial\bx_i\transpose}(\bar\bx_i)\bW\bt_i\right\}\right.\\ 
&\quad \quad \quad \quad \left. - \pi_i(\widehat\bx_i+\frac{\bW\bt_i}{\sqrt{n}})\exp\left\{-\frac{1}{2}\bt_i\transpose\bG_{0in}\bt_i\right\}\right|\diff\bt_i  \\
&\quad \leq C\int_{\|\bt_i\|_2\leq\eta_n} \left|\exp\left\{\frac{1}{2}\bt_i\transpose\bW\transpose\frac{1}{n}\frac{\partial^2\widehat\ell_{in}}{\partial\bx_i\partial\bx_i\transpose}(\bar\bx_i)\bW\bt_i\right\} - \exp\left\{-\frac{1}{2}\bt_i\transpose\bG_{0in}\bt_i\right\}\right|\diff\bt_i  \\
&\quad = C\int_{\|\bt_i\|_2\leq\eta_n} \left|\exp\left\{\frac{1}{2}\bt_i\transpose\left(\bW\transpose\frac{1}{n}\frac{\partial^2\widehat\ell_{in}}{\partial\bx_i\partial\bx_i\transpose}(\bar\bx_i)\bW+\bG_{0in}\right)\bt_i\right\} - 1\right|e^{-\bt_i\transpose\bG_{0in}\bt_i/2}\diff\bt_i  \\
&\quad \leq C\int_{\|\bt_i\|_2\leq\eta_n} \left(\exp\left\{\frac{1}{2}C_{c,\delta,\lambda}\sqrt{\frac{\log n}{n\rho_n}}\eta_n^2\right\} - 1\right)e^{-\bt_i\transpose\bG_{0in}\bt_i/2}\diff\bt_i  \\
&\quad\lesssim_{c,\delta,\lambda} \sqrt{\frac{\log n}{n\rho_n}}\eta_n^2 \int_{\|\bt_i\|_2\leq\eta_n}e^{-\bt_i\transpose\bG_{0in}\bt_i/2}\diff\bt_i  
% \\&
\asymp_{c,\delta,\lambda} \sqrt{\frac{(\log n)^3}{n}}
\end{align*}
with probability at least $1-n^{-c}$ by Lemma \ref{lemma:concentration-hessian}, and for the second term of this integral over $\calA_{1in}$ we have
\begin{align*}
&\int_{\|\bt_i\|_2\leq\eta_n} \left|\pi_i(\widehat\bx_i+\frac{\bW\bt_i}{\sqrt{n}})\exp\left\{-\frac{1}{2}\bt_i\transpose\bG_{0in}\bt_i\right\} - \pi_i(\rho_n\halfpower\bW\bx_{0i})\exp\left\{-\frac{1}{2}\bt_i\transpose\bG_{0in}\bt_i\right\}\right|\diff\bt_i \\
&\quad= \int_{\|\bt_i\|_2\leq\eta_n} \left|\nabla\pi_i\left(\theta\bx_i+(1-\theta)\rho_n\halfpower\bW\bx_{0i}\right)\transpose(\bx_i-\rho_n\halfpower\bW\bx_{0i})\right|\exp\left\{-\frac{1}{2}\bt_i\transpose\bG_{0in}\bt_i\right\}\diff\bt_i \\
&\quad\lesssim_{c,\delta,\lambda} \sqrt{\frac{\log n}{n}}\int_{\|\bt_i\|_2\leq\eta_n}\exp\left\{-\frac{1}{2}\bt_i\transpose\bG_{0in}\bt_i\right\}\diff\bt_i 
% \\&
\asymp_{c,\delta,\lambda} \sqrt{\frac{\log n}{n}}
\end{align*}
with probability at least $1-n^{-c}$ by Assumption \ref{assumption:prior}.

\noindent{} Hence, the integral in the statement of this theorem is bounded by
\begin{align*}
&\int_{\mathbb{R}^d} \left|\pi_i(\bx_i)\exp\left\{\widehat\ell_{in}(\bx_i)-\widehat\ell_{in}(\widehat\bx_i)\right\}\right.\\ 
&\quad\quad\left. - \pi_i(\rho_n\halfpower\bW\bx_{0i})\exp\left\{-\frac{n}{2}(\bx_i-\widehat\bx_i)\transpose\bW\bG_{0in}\bW\transpose(\bx_i-\widehat\bx_i)\right\}\right|\diff\bx_i  \\
&\quad\lesssim_{c,\delta,\lambda} n^{-d/2}\frac{1}{\log n} + n^{-d/2}\sqrt{\frac{(\log n)^3}{n}} + n^{-d/2}\sqrt{\frac{\log n}{n}} 
% \\&
\asymp_{c,\delta,\lambda} n^{-d/2}\frac{1}{\log n}
\end{align*}
with probability at least $1-n^{-c}$.
\end{proof}

\subsection{Proof of Theorem \ref{thm:vb-posterior}}
The proof of Theorem \ref{thm:vb-posterior} can be done by the triangle inequality, Pinsker's inequality, and Lemma \ref{lemma:KL-subgaussian-error} below. 
\begin{lemma}\label{lemma:KL-subgaussian-error}
Suppose the conditions in Theorem \ref{thm:bernstein-von-mises} hold. 
% Let $\calQ$ be a family of probability measures on $\mathbb{R}^d$ that are absolutely continuous with respect to Lebesgue measure such that
% \[
% Q_i\left(\left\{\bx_i\in\mathbb{R}^d:\|\bx_i-\widehat\bx_i\|_2\geq \frac{t}{\sqrt{n}}\right\}\right)
% \leq C\exp(-ct^2)
% \]
% for all $t>0$, where $C,c>0$ are absolute constants, for all $Q_i\in\calQ$ with density $q_i(\bx_i)$. 
Then
\[
D_\mathrm{{KL}}(\phi_d(\bx_i|\widehat\bx_i,(n\bW\bG_{0in}\bW\transpose)\inverse)\|\pi_{in}(\bx_i|\bA))
\lesssim_{c,\delta,\lambda}  \frac{1}{\log n}.
\]
\end{lemma}

\begin{proof}
% [Proof of Lemma \ref{lemma:KL-subgaussian-error}]
Note that Theorem \ref{thm:bernstein-von-mises} implies that
\begin{align*}
\left|\frac{\int_{\mathbb{R}^d}\pi_i(\bx_i)\exp\left\{\widehat\ell_{in}(\bx_i)-\widehat\ell_{in}(\widehat\bx_i)\right\}\diff\bx_i}{\pi_{in}(\rho_n\halfpower\bW\bx_{0i})\det\{2\pi(n\bW\bG_{0in}\bW\transpose)\inverse\}\halfpower} - 1\right|
\lesssim_{c,\delta,\lambda} \frac{1}{\log n}
\end{align*}
with probability at least $1-n^{-c}$.
% And note that
% \begin{align*}
% &\quad \left|D_\mathrm{{KL}}(q_i(\bx_i)\|\pi_{in}(\bx_i|\bA)) - D_\mathrm{{KL}}(q_i(\bx_i)\|\phi_d(\bx_i|\widehat\bx_i,(n\bW\bG_{0in}\bW\transpose)\inverse))\right|  \\
% &= \left|\int_{\mathbb{R}^d} \log\left(\frac{q_i(\bx_i)}{\pi_{in}(\bx_i|\bA)}\right)q_i(\bx_i)\diff\bx_i - \int_{\mathbb{R}^d} \log\left(\frac{q_i(\bx_i)}{\phi_d(\bx_i|\widehat\bx_i,(n\bW\bG_{0in}\bW\transpose)\inverse)}\right)q_i(\bx_i)\diff\bx_i\right|  \\
% &= \left|\int_{\mathbb{R}^d} \log\left(\frac{\phi_d(\bx_i|\widehat\bx_i,(n\bW\bG_{0in}\bW\transpose)\inverse)}{\pi_{in}(\bx_i|\bA)}\right)q_i(\bx_i)\diff\bx_i \right|,
% \end{align*}
% in which
\begin{align*}
&\log\left(\frac{\phi_d(\bx_i|\widehat\bx_i,(n\bW\bG_{0in}\bW\transpose)\inverse)}{\pi_{in}(\bx_i|\bA)}\right)  \\
% &= \log\det\{2\pi(n\bW\bG_{0in}\bW\transpose)\inverse\}\invhalfpower -\frac{n}{2}(\bx_i-\widehat\bx_i)\transpose\bW\bG_{0in}\bW\transpose(\bx_i-\widehat\bx_i)  \\
% &\quad + \log\left(\int_{\mathbb{R}^d}\pi_i(\bx_i)\exp\left\{\widehat\ell_{in}(\bx_i)-\widehat\ell_{in}(\widehat\bx_i)\right\}\diff\bx_i\right) - \log\pi_i(\bx_i) - \widehat\ell_{in}(\bx_i) + \widehat\ell_{in}(\widehat\bx_i)  \\
&\quad= - \widehat\ell_{in}(\bx_i) + \widehat\ell_{in}(\widehat\bx_i) - \frac{n}{2}(\bx_i-\widehat\bx_i)\transpose\bW\bG_{0in}\bW\transpose(\bx_i-\widehat\bx_i)  \\
&\qquad + \log\left\{\frac{\int_{\mathbb{R}^d}\pi_i(\bx_i)\exp\left\{\widehat\ell_{in}(\bx_i)-\widehat\ell_{in}(\widehat\bx_i)\right\}\diff\bx_i}{\pi_i(\rho_n\halfpower\bW\bx_{0i})\det\{2\pi(n\bW\bG_{0in}\bW\transpose)\inverse\}\halfpower}\right\} - \log\pi_i(\bx_i) + \log\pi_i(\rho_n\halfpower\bW\bx_{0i}) \\
&\quad= - \frac{n}{2}(\bx_i-\widehat\bx_i)\transpose\left\{\frac{1}{n}\frac{\partial^2\widehat\ell_{in}}{\partial\bx_i\partial\bx_i\transpose}(\bx_i') + \bW\bG_{0in}\bW\transpose\right\}(\bx_i-\widehat\bx_i) \\
&\qquad + \log\left(1+\frac{1}{\log n}\theta_{in}\right) - \frac{\partial\log\pi_i}{\partial\bx_i}(\rho_n\halfpower\bW\bx_{0i})(\bx_i-\rho_n\halfpower\bW\bx_{0i}) \\
&\qquad - (\bx_i-\rho_n\halfpower\bW\bx_{0i})\transpose\frac{\partial^2\log\pi_i}{\partial\bx_i\partial\bx_i\transpose}(\bx_i'')(\bx_i-\rho_n\halfpower\bW\bx_{0i}),
\end{align*}
where $\bx_i'$ and $\bx_i''$ lie between $\bx_i$ and $\rho_n\halfpower\bW\bx_{0i}$, and $\theta_{in}=O(1)$ with probability at least $1-n^{-c}$ by Theorem \ref{thm:bernstein-von-mises}. Let $\eta_n = C_{c,\delta,\lambda}\sqrt{\log n}$, and partition $\mathbb{R}^d$ as $\calA_{1in}\cup\calA_{2in}$ where
\[
\calA_{1in} = \left\{\bx_i\in\mathbb{R}^d: \sqrt{n}\|\bx_i-\widehat\bx_i\|_2 \leq \eta_n\right\},
\quad \calA_{2in} = \left\{\bx_i\in\mathbb{R}^d: \sqrt{n}\|\bx_i-\widehat\bx_i\|_2 > \eta_n\right\}.
\]
On $\calA_{1in}$,
\[
\|\bW\transpose\bx_i-\rho_n\halfpower\bx_{0i}\|_2 \leq \|\bx_i-\widehat\bx_i\|_2 + \|\bW\transpose\widehat\bx_i-\rho_n\halfpower\bx_{0i}\|_2 
\lesssim_{c,\delta,\lambda} \sqrt{\frac{\log n}{n}}
\]
with probability at least $1-n^{-c}$ by Theorem \ref{thm:mesle}, and the same bound also holds for $\|\bW\transpose\bx_i'-\rho_n\halfpower\bx_{0i}\|_2$ and $\|\bW\transpose\bx_i''-\rho_n\halfpower\bx_{0i}\|_2$. Then by Lemma \ref{lemma:concentration-hessian}, Taylor's theorem, and Assumption \ref{assumption:prior},
\begin{align*}
% &\quad 
\sup_{\bx_i\in\calA_{1in}}\left|\log\left(\frac{\phi_d(\bx_i|\widehat\bx_i,(n\bW\bG_{0in}\bW\transpose)\inverse)}{\pi_{in}(\bx_i|\bA)}\right)\right|  
% \\
&\lesssim_{c,\delta,\lambda} \sqrt{\frac{(\log n)^3}{n}} + \frac{1}{\log n} + \sqrt{\frac{\log n}{n}}  
% \\&
\asymp_{c,\delta,\lambda} \frac{1}{\log n}
\end{align*}
with probability at least $1-n^{-c}$. So the integral over $\calA_{1in}$
\begin{align*}
\left|\int_{\calA_{1in}} \log\left(\frac{\phi_d(\bx_i|\widehat\bx_i,(n\bW\bG_{0in}\bW\transpose)\inverse)}{\pi_{in}(\bx_i|\bA)}\right)\phi_d(\bx_i|\widehat\bx_i,(n\bW\bG_{0in}\bW\transpose)\inverse)\diff\bx_i \right|
\lesssim_{c,\delta,\lambda} \frac{1}{\log n}
\end{align*}
with probability at least $1-n^{-c}$. We then consider the integral over $\calA_{2in}$. Note that over $\calA_{2in}$, by Theorem \ref{thm:mesle},
\begin{align*}
\|\bx_i'-\widehat\bx_i\|_2
&\leq \|\bx_i'-\bx_i\|_2 + \|\bx_i-\widehat\bx_i\|_2  
% \\&
\leq \|\bx_i - \rho_n\halfpower\bW\bx_{0i}\|_2 + \|\bx_i-\widehat\bx_i\|_2  \\
&\leq 2\|\bx_i-\widehat\bx_i\|_2 + \|\bW\transpose\widehat\bx_i-\rho_n\halfpower\bx_{0i}\|_2  
% \\&
\lesssim_{c,\delta,\lambda} \|\bx_i-\widehat\bx_i\|_2
\end{align*}
with probability at least $1-n^{-c}$. Then by triangle inequality, Lemma \ref{lemma:concentration-hessian}, and Lemma \ref{lemma:lipschitz-hessian}, over $\calA_{2in}$ we have
\begin{align*}
& \left|\frac{n}{2}(\bx_i-\widehat\bx_i)\transpose\left\{\frac{1}{n}\frac{\partial^2\widehat\ell_{in}}{\partial\bx_i\partial\bx_i\transpose}(\bx_i') + \bW\bG_{0in}\bW\transpose\right\}(\bx_i-\widehat\bx_i)\right| \\
&\quad\leq \left|\frac{n}{2}(\bx_i-\widehat\bx_i)\transpose\left\{\frac{1}{n}\frac{\partial^2\widehat\ell_{in}}{\partial\bx_i\partial\bx_i\transpose}(\widehat\bx_i) + \bW\bG_{0in}\bW\transpose\right\}(\bx_i-\widehat\bx_i)\right|  \\
&\quad\quad + \left|\frac{n}{2}(\bx_i-\widehat\bx_i)\transpose\left\{\frac{1}{n}\frac{\partial^2\widehat\ell_{in}}{\partial\bx_i\partial\bx_i\transpose}(\bx_i') - \frac{1}{n}\frac{\partial^2\widehat\ell_{in}}{\partial\bx_i\partial\bx_i\transpose}(\widehat\bx_i)\right\}(\bx_i-\widehat\bx_i)\right| \\
&\quad\lesssim_{c,\delta,\lambda} n\sqrt{\frac{\log n}{n}}\|\bx_i-\widehat\bx_i\|_2^2 + n\|\bx_i-\widehat\bx_i\|_2^2 \cdot \|\bx_i'-\widehat\bx_i\|_2  
% \\&
\lesssim_{c,\delta,\lambda} n\|\bx_i-\widehat\bx_i\|_2^3
\end{align*}
with probability at least $1-n^{-c}$. And by Assumption \ref{assumption:prior} we have
\begin{align*}
&\left|\frac{\partial\log\pi_i}{\partial\bx_i}(\rho_n\halfpower\bW\bx_{0i})(\bx_i-\rho_n\halfpower\bW\bx_{0i}) + (\bx_i-\rho_n\halfpower\bW\bx_{0i})\transpose\frac{\partial^2\log\pi_i}{\partial\bx_i\partial\bx_i\transpose}(\bx_i'')(\bx_i-\rho_n\halfpower\bW\bx_{0i})\right|  \\
&\quad\leq C(\|\bx_i-\rho_n\halfpower\bW\bx_{0i}\|_2 + \|\bx_i-\rho_n\halfpower\bW\bx_{0i}\|_2^2) \\
&\quad\leq C(\|\bx_i-\widehat\bx_i\|_2 + \|\bW\transpose\widehat\bx_i - \rho_n\halfpower\bx_{0i}\|_2) \\
&\quad\quad + C(\|\bx_i-\widehat\bx_i\|_2^2 + \|\bW\transpose\widehat\bx_i - \rho_n\halfpower\bx_{0i}\|_2^2 + 2\|\bx_i-\widehat\bx_i\|_2\cdot\|\bW\transpose\widehat\bx_i - \rho_n\halfpower\bx_{0i}\|_2)  \\
&\quad\lesssim_{c,\delta,\lambda} \|\bx_i-\widehat\bx_i\|_2 + \|\bx_i-\widehat\bx_i\|_2^2
\end{align*}
with probability at least $1-n^{-c}$. So the integral over $\calA_{2in}$
\begin{align*}
&\left|\int_{\calA_{2in}} \log\left(\frac{\phi_d(\bx_i|\widehat\bx_i,(n\bW\bG_{0in}\bW\transpose)\inverse)}{\pi_{in}(\bx_i|\bA)}\right)\phi_d(\bx_i|\widehat\bx_i,(n\bW\bG_{0in}\bW\transpose)\inverse)\diff\bx_i \right|  \\
% &\quad\lesssim_{c,\delta,\lambda} \int_{\calA_{2in}} \left(n\|\bx_i-\widehat\bx_i\|_2^3 + \frac{1}{\log n} + \|\bx_i-\widehat\bx_i\|_2 + \|\bx_i-\widehat\bx_i\|_2^2\right)\phi_d(\bx_i|\widehat\bx_i,(n\bW\bG_{0in}\bW\transpose)\inverse)\diff\bx_i  \\
&\quad\lesssim_{c,\delta,\lambda} \frac{1}{\sqrt{n}} \int_{\sqrt{n}\|\bx_i-\widehat\bx_i\|_2>\eta_n}(\sqrt{n}\|\bx_i-\widehat\bx_i\|_2)^3\phi_d(\bx_i|\widehat\bx_i,(n\bW\bG_{0in}\bW\transpose)\inverse)\diff\bx_i\\ 
&\qquad  + \frac{1}{\log{n}} \int_{\sqrt{n}\|\bx_i-\widehat\bx_i\|_2>\eta_n}\phi_d(\bx_i|\widehat\bx_i,(n\bW\bG_{0in}\bW\transpose)\inverse)\diff\bx_i 
% \\&\quad
% \lesssim_{c,\delta,\lambda} \frac{1}{\sqrt{n}} + \frac{1}{\log n}\right) 
% \\&
\lesssim_{c,\delta,\lambda} \frac{1}{\log n}
\end{align*}
with probability at least $1-n^{-c}$. This completes the proof.
\end{proof}
% We are now in a position to prove Theorem \ref{thm:vb-posterior}.
% \begin{proof}
% Let $\Delta_{in}(q)$ denote the error
% \[
% \Delta_{in}(q) = \left|D_\mathrm{{KL}}(q(\bx_i)\|\pi_{in}(\bx_i|\bA)) - D_\mathrm{{KL}}(q(\bx_i)\|\phi_d(\bx_i|\widehat\bx_i,(n\bW\bG_{0in}\bW\transpose)\inverse))\right|.
% \]
% By $q_{in}^*(\bx_i)=\argmin_{q\in\calQ_d}D_{\mathrm{KL}}(q(\bx_i)\|\pi_{in}(\bx_i|\bA))$, triangle inequality, and Lemma \ref{lemma:KL-subgaussian-error},
% \begin{align*}
% D_{\mathrm{KL}}(q_{in}^*(\bx_i)\|\pi_{in}(\bx_i|\bA))
% &\leq D_{\mathrm{KL}}(\phi_d(\bx_i|\widehat\bx_i,(n\bW\bG_{0in}\bW\transpose)\inverse)\|\pi_{in}(\bx_i|\bA)) \\
% &\leq \Delta_{in}(\phi_d(\bx_i|\widehat\bx_i,(n\bW\bG_{0in}\bW\transpose)\inverse))\\
% &\quad + D_{\mathrm{KL}}(\phi_d(\bx_i|\widehat\bx_i,(n\bW\bG_{0in}\bW\transpose)\inverse)\|\phi_d(\bx_i|\widehat\bx_i,(n\bW\bG_{0in}\bW\transpose)\inverse)) \\
% &\lesssim_{c,\delta,\lambda} \frac{1}{\log n}
% \end{align*}
% with probability at least $1-n^{-c}$.

% \noindent{}By Pinsker's inequality,
% \begin{align*}
% \int_{\mathbb{R}^d} \left|q_{in}^*(\bx_i) - \pi_{in}(\bx_i|\bA)\right| \diff\bx_i\leq \sqrt{D_{\mathrm{KL}}(q_{in}^*(\bx_i)\|\pi_{in}(\bx_i|\bA))} \lesssim_{c,\delta,\lambda} \sqrt{\frac{1}{\log n}}
% \end{align*}
% with probability at least $1-n^{-c}$.
% \end{proof}

\subsection{Proof of Theorem \ref{thm:vb-posterior-mean}}

\begin{proof}
Let $Q_{in}^*$ denote the variational posterior distribution $N(\bx_i^*,\bSigma_{in}^*)$, with density $q_{in}^*(\bx_i)$, and let $N_{in}^*$ denote the normal distribution $N(\widehat\bx_i,\{n\bW\bG_{0in}\bW\transpose\}\inverse)$, with density $\phi_d(\bx_i|\widehat\bx_i,(n\bW\bG_{0in}\bW\transpose)\inverse)$.
Note that Theorem \ref{thm:bernstein-von-mises} implies that
\begin{align*}
\int_{\mathbb{R}^d}\left|\pi_{in}(\bx_i|\bA) - \phi_d(\bx_i|\widehat\bx_i,(n\bW\bG_{0in}\bW\transpose)\inverse)\right| \diff\bx_i
\lesssim_{c,\delta,\lambda} \frac{1}{\log n}
\end{align*}
with probability at least $1-n^{-c}$, and Theorem \ref{thm:vb-posterior} implies that
% \[
% \int_{\mathbb{R}^d} \left|q_{in}^*(\bx_i) - \pi_{in}(\bx_i|\bA)\right| \diff\bx_i
% \lesssim_{c,\delta,\lambda} \sqrt{\frac{1}{\log n}}
% \]
% with probability at least $1-n^{-c}$. So
\[
\int_{\mathbb{R}^d} \left|q_{in}^*(\bx_i) - \phi_d(\bx_i|\widehat\bx_i,(n\bW\bG_{0in}\bW\transpose)\inverse)\right| \diff\bx_i
\lesssim_{c,\delta,\lambda} \sqrt{\frac{1}{\log n}}
\]
with probability at least $1-n^{-c}$ by triangle inequality.

\noindent{}Let $\varphi_Q(\bt)$ denote the characteristic function of a distribution $Q$. We have
\begin{align*}
\sup_{\bt\in\mathbb{R}^d}|\varphi_{Q_{in}^*}(\bt) - \varphi_{N_{in}^*}(\bt)|
&= \sup_{\bt\in\mathbb{R}^d}\left|\int_{\mathbb{R}^d}e^{i\bt\transpose\bx_i} \{q_{in}^*(\bx_i) - \phi_d(\bx_i|\widehat\bx_i,(n\bW\bG_{0in}\bW\transpose)\inverse)\} \diff\bx_i\right|  \\
&\leq \int_{\mathbb{R}^d} \left|q_{in}^*(\bx_i) - \phi_d(\bx_i|\widehat\bx_i,(n\bW\bG_{0in}\bW\transpose)\inverse)\right| \diff\bx_i 
\lesssim_{c,\delta,\lambda} \sqrt{\frac{1}{\log n}}
\end{align*}
with probability at least $1-n^{-c}$. Recall that both $Q_{in}^*$ and $N_{in}^*$ are $d$-dimensional normal distributions, so
\begin{align*}
\varphi_{Q_{in}^*}(\bt) &= \exp\left\{i\bt\transpose\bx_i^* - \frac{1}{2}\bt\transpose\bSigma_{in}^*\bt\right\}, \quad
\varphi_{N_{in}^*}(\bt) = \exp\left\{i\bt\transpose\widehat\bx_i - \frac{1}{2}\bt\transpose\{n\bW\bG_{0in}\bW\transpose\}\inverse\bt\right\}.
\end{align*}
By triangle inequality and the fact that $|\exp\{ix\}|=1$ for all $x\in\mathbb{R}$, we have
\[
\sup_{\bt\in\mathbb{R}^d}\left|\exp\left(- \frac{1}{2}\bt\transpose\bSigma_{in}^*\bt\right) - \exp\left( - \frac{1}{2}\bt\transpose\{n\bW\bG_{0in}\bW\transpose\}\inverse\bt\right)\right|
\lesssim_{c,\delta,\lambda} \sqrt{\frac{1}{\log n}}
\]
with probability at least $1-n^{-c}$. Also note that, for $\bu\in\mathbb{R}^d$ with $\|\bu\|_2=1$,
\[
\frac{1}{2}\bu\transpose(n\bW\bG_{0in}\bW\transpose)\inverse\bu
\asymp_{\delta,\lambda} \frac{1}{n},
\]
since $\bG_{0in}$ has all eigenvalues positive and bounded away from $0$ and $+\infty$ by \eqref{eqn:fisher-information-order} on page \pageref{eqn:fisher-information-order}.
Then,
\begin{align*}
&\exp\left\{-\frac{1}{2}\bt\transpose(n\bW\bG_{0in}\bW\transpose)\inverse\bt\right\}|\exp(i\bt\transpose\bx_i^*) - \exp(i\bt\transpose\widehat\bx_i)| \\
&\quad= \left|\exp\left\{i\bt\transpose\bx_i^*-\frac{1}{2}\bt\transpose(n\bW\bG_{0in}\bW\transpose)\inverse\bt\right\} - \exp\left\{i\bt\transpose\widehat\bx_i-\frac{1}{2}\bt\transpose(n\bW\bG_{0in}\bW\transpose)\inverse\bt\right\}\right|  \\
&\quad\leq \left|\exp\left\{i\bt\transpose\bx_i^*-\frac{1}{2}\bt\transpose(n\bW\bG_{0in}\bW\transpose)\inverse\bt\right\} - \exp\left\{i\bt\transpose\bx_i^*-\frac{1}{2}\bt\transpose\bSigma_{in}^*\bt\right\}\right|  \\
&\quad\quad +\left|\exp\left\{i\bt\transpose\bx_i^*-\frac{1}{2}\bt\transpose\bSigma_{in}^*\bt\right\} - \exp\left\{i\bt\transpose\widehat\bx_i-\frac{1}{2}\bt\transpose(n\bW\bG_{0in}\bW\transpose)\inverse\bt\right\}\right|  \\
&\quad\leq \left|\exp\left\{-\frac{1}{2}\bt\transpose(n\bW\bG_{0in}\bW\transpose)\inverse\bt\right\} - \exp\left\{-\frac{1}{2}\bt\transpose\bSigma_{in}^*\bt\right\}\right| + |\varphi_{N_{in}^*}(\bt) - \varphi_{Q_{in}^*}(\bt)|,
\end{align*}
which, by taking $\bt=\sqrt{n}\bu$, implies that
\[
|\exp(i\bu\transpose\sqrt{n}\bx_i^*) - \exp(i\bu\transpose\sqrt{n}\widehat\bx_i)|
\lesssim_{c,\delta,\lambda} \sqrt{\frac{1}{\log n}}
\]
for all $\|\bu\|_2=1$, $\bu\in\mathbb{R}^d$, with probability at least $1-n^{-c}$. Namely, 
% \[
$\sqrt{n}(\widehat\bx_i - \bx_i^*) = o_{\prob}(1)$.
% \]
By Theorem \ref{thm:mesle} and Slutsky's theorem,
% \[
$\sqrt{n}\bG_{0in}\halfpower(\bW\transpose\bx_i^*-\rho_n\halfpower\bx_{0i}) \overset{\calL}{\to} \mathrm{N}_d(\zero_d,\,\eye_d)$.
% \]
% If furthermore $(\log n)^4/(n\rho_n)\to 0$, then by Theorem \ref{thm:mesle} and $\sqrt{n}(\widehat\bx_i - \bx_i^*) = o_{\prob}(1)$, we have
% \[
% \left|\left<\widehat\bX\bW-\rho_n\halfpower\bX_0, \bX^*\bW-\widehat\bX\bW\right>\frobenius\right|
% \leq \|\widehat\bX\bW-\rho_n\halfpower\bX_0\|\frobenius \cdot \|\bX^*-\widehat\bX\|\frobenius
% = o_{\prob}(1),
% \]
% with which we have
% \begin{align*}
% \|\bX^*\bW-\rho_n\halfpower\bX_0\|\frobenius^2
% &= \|\widehat\bX\bW-\rho_n\halfpower\bX_0\|\frobenius^2 + \|\bX^*-\widehat\bX\|\frobenius^2 + 2\left<\widehat\bX\bW-\rho_n\halfpower\bX_0, \bX^*\bW-\widehat\bX\bW\right>\frobenius \\
% &= \frac{1}{n}\sum_{i=1}^n\mathrm{tr}(\bG_{0in}\inverse) + o_{\prob}(1).
% \end{align*}
\end{proof}

\subsection{Proof of Theorem \ref{thm:vb-objective-fun}}

\begin{proof}
We borrow the idea in the proof of Theorem 11 in \cite{xu-campbell-2023} to proof the strong convexity.
By Assumption \ref{assumption:prior}, $-\log\pi_i(\bx_i)$ is convex in $\bx_i$. By Exercise 12.21 in \cite{abadir-magnus-2005}, $-\log\det(\bL_i)$ is convex in $\bL_i$.
By \eqref{eqn:hessian-strongly-convex-bound} in the proof of Theorem \ref{thm:mesle} on page \pageref{eqn:hessian-strongly-convex-bound}, $-\widehat\ell(\bx_i)$ is strongly convex in $\bx_i$ with strong convexity parameter $n\lambda\rho_n$ with probability at least $1-n^{-c}$. Let $\bD_n=n\lambda\rho_n\eye_d$, and note that
\begin{align*}
\expect\left[\frac{1}{2}\left(\bmu_i+\frac{1}{\sqrt{n}}\bL_i\bZ\right)\transpose \bD_n\left(\bmu_i+\frac{1}{\sqrt{n}}\bL_i\bZ\right)\right]
&= \frac{1}{2}\bmu_i\transpose \bD_n\bmu_i + \frac{1}{2n}\mathrm{tr}\left(\bL_i\transpose \bD_n\bL_i\right).
\end{align*}
Then by the linearity of expectation and the strong convexity of $-\widehat\ell_{in}(\bx_i)$, the function
\begin{align*}
&\expect_\bZ\left[- \widehat\ell_{in}\left(\bmu_i+\frac{1}{\sqrt{n}}\bL_i\bZ\right) - \frac{1}{2}\left(\bmu_i+\frac{1}{\sqrt{n}}\bL_i\bZ\right)\transpose \bD_n\left(\bmu_i+\frac{1}{\sqrt{n}}\bL_i\bZ\right)\right]  \\
&\quad = - \expect_\bZ\left[\widehat\ell_{in}\left(\bmu_i+\frac{1}{\sqrt{n}}\bL_i\bZ\right)\right] - \frac{1}{2}\bmu_i\transpose \bD_n\bmu_i - \frac{1}{2n}\mathrm{tr}\left(\bL_i\transpose \bD_n\bL_i\right)
\end{align*}
is convex in $(\bmu_i,\bL_i)\in\mathbb{R}^d\times\calL_{d\times d}$ with probability at least $1-n^{-c}$.
So the function
\begin{align*}
F_{in}(\bmu_i,\bL_i) - \frac{1}{2}\bmu_i\transpose \bD_n\bmu_i - \frac{1}{2n}\mathrm{tr}\left(\bL_i\transpose \bD_n\bL_i\right)
\end{align*}
is convex in $(\bmu_i,\bL_i)\in\mathbb{R}^d\times\calL_{d\times d}$, with probability at least $1-n^{-c}$, which equivalently means that the function $F_{in}(\bmu_i,\bL_i)$ is strongly convex in $(\bmu_i,\bL_i)\in\mathbb{R}^d\times\calL_{d\times d}$, with probability at least $1-n^{-c}$.

\noindent{}We next prove the interchange of derivative and expectation.
By the definition of $\widehat\ell(\bx_i)$,
\[
\|\nabla_{\bx_i}\widehat\ell(\bx_i)\|_2 \leq \sum_{j=1}^n \frac{1}{\tau_n^2}(|\bx_i\transpose\widetilde\bx_j-0.5|+2)\|\widetilde\bx_j\|_2,
\]
which is polynomial in $\bx_i$, so $\|\nabla_{\bmu_i,\bL_i}\widehat\ell_{in}\left(\bmu_i+\frac{1}{\sqrt{n}}\bL_i\bZ\right)\|_2$ is dominated by a function that is integrable with respect to the standard normal measure. By Assumption \ref{assumption:prior},
\begin{align*}
\|\nabla_{\bx_i}\log\pi_i(\bx_i)\|_2
&= \left\|\frac{\partial}{\partial\bx_i}\log\pi_i(\bx_{0i}) + \frac{\partial^2}{\partial\bx_i\partial\bx_i\transpose}\log\pi_i(\bar\bx_i)(\bx_i-\bx_{0i})\right\|_2  
% \\&
\leq C + C\|\bx_i-\bx_{0i}\|_2,
\end{align*}
which is polynomial in $\bx_i$, so $\|\nabla_{\bmu_i,\bL_i}\log\pi_i\left(\bmu_i+\frac{1}{\sqrt{n}}\bL_i\bZ\right)\|_2$ is dominated by a function that is integrable with respect to the standard normal measure. Hence,
\begin{align*}
\expect_\bZ\left[\nabla_{\bmu_i,\bL_i}\widehat\ell_{in}\left(\bmu_i+\frac{1}{\sqrt{n}}\bL_i\bZ\right)\right] &= \nabla_{\bmu_i,\bL_i}\expect_\bZ\left[\widehat\ell_{in}\left(\bmu_i+\frac{1}{\sqrt{n}}\bL_i\bZ\right)\right]  \\
\expect_\bZ\left[\nabla_{\bmu_i,\bL_i}\log\pi_i\left(\bmu_i+\frac{1}{\sqrt{n}}\bL_i\bZ\right)\right] &= \nabla_{\bmu_i,\bL_i}\expect_\bZ\left[\log\pi_i\left(\bmu_i+\frac{1}{\sqrt{n}}\bL_i\bZ\right)\right].
\end{align*}
So $\nabla_{\bmu_i,\bL_i}F_{in}(\bmu_i,\bL_i) = \expect_\bZ\left[\nabla_{\bmu_i,\bL_i}f_{in}(\bmu_i,\bL_i)\right]$.
\end{proof}

% \newpage
%% if your bibliography is in bibtex format, uncomment commands:
\bibliographystyle{apalike} % Style BST file
\bibliography{reference_VB,reference_VBGraph}       % Bibliography file (usually '*.bib')

\begin{thebibliography}{}

\bibitem[Abadir and Magnus, 2005]{abadir-magnus-2005}
Abadir, K.~M. and Magnus, J.~R. (2005).
\newblock {\em Matrix Algebra}.
\newblock Econometric Exercises. Cambridge University Press.

\bibitem[Abbe, 2018]{abbe-2018-cdsbm}
Abbe, E. (2018).
\newblock Community detection and stochastic block models: Recent developments.
\newblock {\em Journal of Machine Learning Research}, 18(177):1--86.

\bibitem[Abbe et~al., 2016]{abbe-bandeira-hall-2016-sbm}
Abbe, E., Bandeira, A.~S., and Hall, G. (2016).
\newblock Exact recovery in the stochastic block model.
\newblock {\em IEEE Transactions on Information Theory}, 62(1):471--487.

\bibitem[Adamic and Glance, 2005]{adamic-glance-2005-pb}
Adamic, L.~A. and Glance, N. (2005).
\newblock The political blogosphere and the 2004 u.s. election: divided they
  blog.
\newblock In {\em Proceedings of the 3rd International Workshop on Link
  Discovery}, LinkKDD '05, page 36–43, New York, NY, USA. Association for
  Computing Machinery.

\bibitem[Airoldi et~al., 2008]{airoldi-blei-fienberg-2008}
Airoldi, E.~M., Blei, D.~M., Fienberg, S.~E., and Xing, E.~P. (2008).
\newblock Mixed membership stochastic blockmodels.
\newblock {\em Journal of Machine Learning Research}, 9(65):1981--2014.

\bibitem[Athreya et~al., 2016]{athreya-priebe-tang-2016}
Athreya, A., Priebe, C.~E., Tang, M., Lyzinski, V., Marchette, D.~J., and
  Sussman, D.~L. (2016).
\newblock A limit theorem for scaled eigenvectors of random dot product graphs.
\newblock {\em Sankhya A}, 78(1):1--18.

\bibitem[Athreya et~al., 2021]{athreya-tang-park-2021}
Athreya, A., Tang, M., Park, Y., and Priebe, C.~E. (2021).
\newblock {On Estimation and Inference in Latent Structure Random Graphs}.
\newblock {\em Statistical Science}, 36(1):68 -- 88.

\bibitem[Bhattacharya et~al., 2025]{bhattacharya-pati-yang-2025}
Bhattacharya, A., Pati, D., and Yang, Y. (2025).
\newblock On the convergence of coordinate ascent variational inference.
\newblock {\em The Annals of Statistics}, 53(3):929--962.

\bibitem[Bickel and Doksum, 2015]{bickel-doksum-2015}
Bickel, P.~J. and Doksum, K.~A. (2015).
\newblock {\em Mathematical statistics: Basic ideas and selected topics. volume
  I}.
\newblock CRC Press.

\bibitem[Blei et~al., 2017]{blei-kucukelbir-mcAuliffe-2017}
Blei, D.~M., Kucukelbir, A., and and, J. D.~M. (2017).
\newblock Variational inference: A review for statisticians.
\newblock {\em Journal of the American Statistical Association},
  112(518):859--877.

\bibitem[Boucheron et~al., 2013]{boucheron-lugosi-massart-2013}
Boucheron, S., Lugosi, G., and Massart, P. (2013).
\newblock {\em Concentration Inequalities: A Nonasymptotic Theory of
  Independence}.
\newblock Oxford University Press.

\bibitem[Caron and Fox, 2017]{caron-fox-2017}
Caron, F. and Fox, E.~B. (2017).
\newblock Sparse graphs using exchangeable random measures.
\newblock {\em Journal of the Royal Statistical Society Series B: Statistical
  Methodology}, 79(5):1295--1366.

\bibitem[Chernozhukov and Hong, 2003]{chernozhukov-hong-2003}
Chernozhukov, V. and Hong, H. (2003).
\newblock An {MCMC} approach to classical estimation.
\newblock {\em Journal of Econometrics}, 115(2):293--346.

\bibitem[de~la Pena and Montgomery-Smith,
  1995]{pena-montgomery-1995-decoupling}
de~la Pena, V.~H. and Montgomery-Smith, S.~J. (1995).
\newblock {Decoupling Inequalities for the Tail Probabilities of Multivariate
  $U$-Statistics}.
\newblock {\em The Annals of Probability}, 23(2):806 -- 816.

\bibitem[Fan et~al., 2022]{fan-fan-han-2022}
Fan, J., Fan, Y., Han, X., and Lv, J. (2022).
\newblock Simple: Statistical inference on membership profiles in large
  networks.
\newblock {\em Journal of the Royal Statistical Society Series B: Statistical
  Methodology}, 84(2):630--653.

\bibitem[Girvan and Newman, 2002]{girvan-newman-2002}
Girvan, M. and Newman, M. E.~J. (2002).
\newblock Community structure in social and biological networks.
\newblock {\em Proceedings of the National Academy of Sciences},
  99(12):7821--7826.

\bibitem[Han and Yang, 2019]{han-yang-2019}
Han, W. and Yang, Y. (2019).
\newblock Statistical inference in mean-field variational {B}ayes.

\bibitem[Hinton and van Camp, 1993]{hinton-camp-1993}
Hinton, G.~E. and van Camp, D. (1993).
\newblock Keeping the neural networks simple by minimizing the description
  length of the weights.
\newblock In {\em Proceedings of the Sixth Annual Conference on Computational
  Learning Theory}, COLT '93, page 5–13, New York, NY, USA. Association for
  Computing Machinery.

\bibitem[Hoff et~al., 2002]{hoff-raftery-handcock-2002}
Hoff, P.~D., Raftery, A.~E., and and, M. S.~H. (2002).
\newblock Latent space approaches to social network analysis.
\newblock {\em Journal of the American Statistical Association},
  97(460):1090--1098.

\bibitem[Holland et~al., 1983]{holland-1983-sbm}
Holland, P.~W., Laskey, K.~B., and Leinhardt, S. (1983).
\newblock Stochastic blockmodels: First steps.
\newblock {\em Social Networks}, 5(2):109--137.

\bibitem[Janson and Diaconis, 2008]{janson-2008-graph}
Janson, S. and Diaconis, P. (2008).
\newblock Graph limits and exchangeable random graphs.
\newblock {\em Rendiconti di Matematica e delle sue Applicazioni. Serie VII},
  28:33--61.

\bibitem[Jin et~al., 2023]{jin-ke-luo-2023}
Jin, J., Ke, Z.~T., and Luo, S. (2023).
\newblock Mixed membership estimation for social networks.
\newblock {\em Journal of Econometrics}.

\bibitem[Jordan et~al., 1998]{jordan-ghahramani-jaakkola-1998}
Jordan, M.~I., Ghahramani, Z., Jaakkola, T.~S., and Saul, L.~K. (1998).
\newblock {\em An Introduction to Variational Methods for Graphical Models},
  pages 105--161.
\newblock Springer Netherlands, Dordrecht.

\bibitem[Karrer and Newman, 2011]{karrer-newman-2011-sbm}
Karrer, B. and Newman, M. E.~J. (2011).
\newblock Stochastic blockmodels and community structure in networks.
\newblock {\em Phys. Rev. E}, 83:016107.

\bibitem[Katsevich and Rigollet, 2024]{katsevich-rigollet-2024}
Katsevich, A. and Rigollet, P. (2024).
\newblock On the approximation accuracy of gaussian variational inference.
\newblock {\em The Annals of Statistics}, 52(4):1384--1409.

\bibitem[Kingma and Ba, 2015]{kingma-ba-2017-adam}
Kingma, D.~P. and Ba, J. (2015).
\newblock Adam: A method for stochastic optimization.
\newblock {\em International Conference on Learning Representations}.

\bibitem[Koo et~al., 2023]{koo-tang-trosset-2023}
Koo, J., Tang, M., and and, M. W.~T. (2023).
\newblock Popularity adjusted block models are generalized random dot product
  graphs.
\newblock {\em Journal of Computational and Graphical Statistics},
  32(1):131--144.

\bibitem[Kucukelbir et~al., 2017]{kucukelbir-tran-ranganath-2017}
Kucukelbir, A., Tran, D., Ranganath, R., Gelman, A., and Blei, D.~M. (2017).
\newblock Automatic differentiation variational inference.
\newblock {\em Journal of Machine Learning Research}, 18(14):1--45.

\bibitem[Lei, 2016]{lei-2016-testsbm}
Lei, J. (2016).
\newblock {A goodness-of-fit test for stochastic block models}.
\newblock {\em The Annals of Statistics}, 44(1):401 -- 424.

\bibitem[Lei, 2021]{lei-2021-graphroot}
Lei, J. (2021).
\newblock {Network representation using graph root distributions}.
\newblock {\em The Annals of Statistics}, 49(2):745 -- 768.

\bibitem[Lei and Rinaldo, 2015]{lei-rinaldo-2015-sbm}
Lei, J. and Rinaldo, A. (2015).
\newblock {Consistency of spectral clustering in stochastic block models}.
\newblock {\em The Annals of Statistics}, 43(1):215 -- 237.

\bibitem[Levin and Levina, 2025]{levin-levina-2025}
Levin, K. and Levina, E. (2025).
\newblock {Bootstrapping networks with latent space structure}.
\newblock {\em Electronic Journal of Statistics}, 19(1):745 -- 791.

\bibitem[Levin et~al., 2021]{levin-roosta-tang-2021}
Levin, K.~D., Roosta, F., Tang, M., Mahoney, M.~W., and Priebe, C.~E. (2021).
\newblock Limit theorems for out-of-sample extensions of the adjacency and
  laplacian spectral embeddings.
\newblock {\em Journal of Machine Learning Research}, 22(194):1--59.

\bibitem[Li et~al., 2020]{li-levina-zhu-2020}
Li, T., Levina, E., and Zhu, J. (2020).
\newblock Network cross-validation by edge sampling.
\newblock {\em Biometrika}, 107(2):257--276.

\bibitem[Lov{\'a}sz, 2012]{lovasz-2012-large}
Lov{\'a}sz, L. (2012).
\newblock {\em Large networks and graph limits}, volume~60.
\newblock American Mathematical Soc.

\bibitem[Loyal, 2024]{loyal2024fast}
Loyal, J.~D. (2024).
\newblock Fast variational inference of latent space models for dynamic
  networks using bayesian p-splines.
\newblock {\em arXiv preprint:2401.09715}.

\bibitem[Lyzinski et~al., 2014]{lyzinski-sussman-tang-2014}
Lyzinski, V., Sussman, D.~L., Tang, M., Athreya, A., and Priebe, C.~E. (2014).
\newblock {Perfect clustering for stochastic blockmodel graphs via adjacency
  spectral embedding}.
\newblock {\em Electronic Journal of Statistics}, 8(2):2905 -- 2922.

\bibitem[Lyzinski et~al., 2017]{lyzinski-tang-athreya-2017}
Lyzinski, V., Tang, M., Athreya, A., Park, Y., and Priebe, C.~E. (2017).
\newblock Community detection and classification in hierarchical stochastic
  blockmodels.
\newblock {\em IEEE Transactions on Network Science and Engineering},
  4(1):13--26.

\bibitem[Neil et~al., 2013]{neil-uphoff-hash-2013-6623779}
Neil, J., Uphoff, B., Hash, C., and Storlie, C. (2013).
\newblock Towards improved detection of attackers in computer networks: New
  edges, fast updating, and host agents.
\newblock In {\em 2013 6th International Symposium on Resilient Control Systems
  (ISRCS)}, pages 218--224.

\bibitem[Peterson and Anderson, 1987]{peterson-anderson-1987}
Peterson, C. and Anderson, J. (1987).
\newblock A mean field theory learning algorithm for neural networks.
\newblock {\em Complex Systems}, 1:995--1019.

\bibitem[Rubin-Delanchy et~al., 2016]{rubin-adams-heard-2016-7745482}
Rubin-Delanchy, P., Adams, N.~M., and Heard, N.~A. (2016).
\newblock Disassortativity of computer networks.
\newblock In {\em 2016 IEEE Conference on Intelligence and Security Informatics
  (ISI)}, pages 243--247.

\bibitem[Rubin-Delanchy et~al., 2022]{rubin-cape-tang-2022}
Rubin-Delanchy, P., Cape, J., Tang, M., and Priebe, C.~E. (2022).
\newblock A statistical interpretation of spectral embedding: The generalised
  random dot product graph.
\newblock {\em Journal of the Royal Statistical Society Series B: Statistical
  Methodology}, 84(4):1446--1473.

\bibitem[Sengupta and Chen, 2017]{sengupta-chen-2017}
Sengupta, S. and Chen, Y. (2017).
\newblock A block model for node popularity in networks with community
  structure.
\newblock {\em Journal of the Royal Statistical Society Series B: Statistical
  Methodology}, 80(2):365--386.

\bibitem[Sussman et~al., 2012]{sussman-tang-fishkind-2012}
Sussman, D.~L., Tang, M., Fishkind, D.~E., and and, C. E.~P. (2012).
\newblock A consistent adjacency spectral embedding for stochastic blockmodel
  graphs.
\newblock {\em Journal of the American Statistical Association},
  107(499):1119--1128.

\bibitem[Sussman et~al., 2014]{sussman-tang-priebe-2014}
Sussman, D.~L., Tang, M., and Priebe, C.~E. (2014).
\newblock Consistent latent position estimation and vertex classification for
  random dot product graphs.
\newblock {\em IEEE Transactions on Pattern Analysis and Machine Intelligence},
  36(1):48--57.

\bibitem[Tang et~al., 2017a]{tang-athreya-sussman-2017-sptest}
Tang, M., Athreya, A., Sussman, D.~L., Lyzinski, V., Park, Y., and and, C.
  E.~P. (2017a).
\newblock A semiparametric two-sample hypothesis testing problem for random
  graphs.
\newblock {\em Journal of Computational and Graphical Statistics},
  26(2):344--354.

\bibitem[Tang et~al., 2017b]{tang-athreya-sussman-2017-nptest}
Tang, M., Athreya, A., Sussman, D.~L., Lyzinski, V., and Priebe, C.~E. (2017b).
\newblock {A nonparametric two-sample hypothesis testing problem for random
  graphs}.
\newblock {\em Bernoulli}, 23(3):1599 -- 1630.

\bibitem[Tang and Priebe, 2018]{tang-priebe-2018}
Tang, M. and Priebe, C.~E. (2018).
\newblock {Limit theorems for eigenvectors of the normalized Laplacian for
  random graphs}.
\newblock {\em The Annals of Statistics}, 46(5):2360 -- 2415.

\bibitem[Tang et~al., 2013]{tang-sussman-priebe-2013}
Tang, M., Sussman, D.~L., and Priebe, C.~E. (2013).
\newblock {Universally consistent vertex classification for latent positions
  graphs}.
\newblock {\em The Annals of Statistics}, 41(3):1406 -- 1430.

\bibitem[Tang et~al., 2019]{tang-ketcha-badea-2019}
Tang, R., Ketcha, M., Badea, A., Calabrese, E.~D., Margulies, D.~S.,
  Vogelstein, J.~T., Priebe, C.~E., and Sussman, D.~L. (2019).
\newblock Connectome smoothing via low-rank approximations.
\newblock {\em IEEE Transactions on Medical Imaging}, 38(6):1446--1456.

\bibitem[van~der Vaart and Wellner, 2023]{vaart-wellner-2023-weak}
van~der Vaart, A. and Wellner, J. (2023).
\newblock {\em Weak Convergence and Empirical Processes: With Applications to
  Statistics}.
\newblock Springer Series in Statistics. Springer International Publishing.

\bibitem[Wang and Blei, 2019]{wang-blei-2019}
Wang, Y. and Blei, D.~M. (2019).
\newblock Frequentist consistency of variational bayes.
\newblock {\em Journal of the American Statistical Association},
  114(527):1147--1161.

\bibitem[Wasserman and Faust, 1994]{wasserman-faust-1994-social}
Wasserman, S. and Faust, K. (1994).
\newblock {\em Social network analysis: Methods and applications}, volume~8.
\newblock Cambridge university press.

\bibitem[Wu and Xie, 2025]{wu-xie-2022-sl}
Wu, D. and Xie, F. (2025+).
\newblock Statistical inference of random graphs with a surrogate likelihood
  function.
\newblock {\em Journal of Machine Learning Research, accepted conditioned on
  minor revision}.

\bibitem[Xie, 2023]{xie-2023-euclidean}
Xie, F. (2023).
\newblock Euclidean representation of low-rank matrices and its geometric
  properties.
\newblock {\em SIAM Journal on Matrix Analysis and Applications},
  44(2):822--866.

\bibitem[Xie, 2024]{xie-2024-spn-bernoulli}
Xie, F. (2024).
\newblock {Entrywise limit theorems for eigenvectors of signal-plus-noise
  matrix models with weak signals}.
\newblock {\em Bernoulli}, 30(1):388 -- 418.

\bibitem[Xie and Wu, 2023]{xie-wu-2023-eigen}
Xie, F. and Wu, D. (2023).
\newblock An eigenvector-assisted estimation framework for signal-plus-noise
  matrix models.
\newblock {\em Biometrika}, 111(2):661--676.

\bibitem[Xie and Xu, 2020]{xie-xu-2020-bayes}
Xie, F. and Xu, Y. (2020).
\newblock Optimal bayesian estimation for random dot product graphs.
\newblock {\em Biometrika}, 107(4):875--889.

\bibitem[Xie and Xu, 2023]{xie-xu-2023-os}
Xie, F. and Xu, Y. (2023).
\newblock Efficient estimation for random dot product graphs via a one-step
  procedure.
\newblock {\em Journal of the American Statistical Association},
  118(541):651--664.

\bibitem[Xu and Campbell, 2023]{xu-campbell-2023}
Xu, Z. and Campbell, T. (2023).
\newblock The computational asymptotics of gaussian variational inference and
  the laplace approximation.

\bibitem[Young and Scheinerman, 2007]{young-scheinerman-2007-rdpg}
Young, S.~J. and Scheinerman, E.~R. (2007).
\newblock Random dot product graph models for social networks.
\newblock In Bonato, A. and Chung, F. R.~K., editors, {\em Algorithms and
  Models for the Web-Graph}, pages 138--149, Berlin, Heidelberg. Springer
  Berlin Heidelberg.

\bibitem[Zhang and Yang, 2024]{zhang-yang-2024-jrsssb}
Zhang, Y. and Yang, Y. (2024).
\newblock Bayesian model selection via mean-field variational approximation.
\newblock {\em Journal of the Royal Statistical Society Series B: Statistical
  Methodology}, 86(3):742--770.

\bibitem[Zhao et~al., 2024]{JMLR:v25:22-0514}
Zhao, P., Bhattacharya, A., Pati, D., and Mallick, B.~K. (2024).
\newblock Structured optimal variational inference for dynamic latent space
  models.
\newblock {\em Journal of Machine Learning Research}, 25(259):1--55.

\end{thebibliography}

\end{document}